\documentclass[a4paper, 12pt, final]{report}

\usepackage[T1]{fontenc}
\usepackage{mathrsfs}
\usepackage{enumerate}
\usepackage[titletoc]{appendix}
\usepackage{graphicx}
\usepackage{amsmath, amsthm, amssymb, amsfonts}
\usepackage{charter}
\usepackage{setspace}
\usepackage[parfill]{parskip}
\usepackage{float}
\usepackage[bf]{caption}
\usepackage[all]{nowidow}
\usepackage[a4paper, twoside, inner=3.2 cm, outer=3.5 cm, 
top=2.5cm, bottom=2.7cm, 
includefoot, includehead, heightrounded]{geometry}
\usepackage[nottoc]{tocbibind}
\usepackage[pdfstartpage=1, pdfstartview=FitBV, 
pdfauthor={Vasileios Iliopoulos}, pdfkeywords={Quicksort,
Information Theory, Combinatorics, Posets},
colorlinks=false,        
unicode=false, linktoc=page,
    pdftitle={The Quicksort algorithm and related topics},    
    hidelinks,        
    filecolor=black,    
    urlcolor=black]{hyperref} 
   
\usepackage{titlesec}
\titlespacing*{\chapter}{0pt}{10pt}{30pt}
\titleformat{\chapter}[display]{\normalfont\huge\bfseries}
{\chaptertitlename\ \thechapter}{10pt}{\huge}

\begingroup
\makeatletter
 \@for\theoremstyle:=definition,remark,plain, Lemma, Theorem, Corollary\do{
    \expandafter\g@addto@macro\csname th@\theoremstyle\endcsname{
       \addtolength\thm@preskip\parskip
     }
   }
\endgroup

\usepackage{fancyhdr}

\fancypagestyle{main}{
\fancyhf{}

\fancyhead[LE, RO]{\bf \thepage} 
\lfoot{}
\cfoot{}
\rfoot{}
}
\fancypagestyle{plain}{
\fancyhf{}

\cfoot{\bf \thepage}
}

\newtheorem{theorem}{Theorem}[section]
\newtheorem{Definition}[theorem]{Definition}
\newtheorem{Lemma}[theorem]{Lemma}
\newtheorem{Corollary}[theorem]{Corollary}
\newtheorem{remark}[theorem]{Remark}
\newtheorem{question}[theorem]{Question}
\theoremstyle{definition}

\renewcommand{\qedsymbol}{\rule{.07in}{.1in}}

\begin{document}

\pagestyle{plain}

\pagenumbering{roman}

\newgeometry{left=3.2cm, right=3.2cm, top=2.5cm, vmarginratio=1:1, centering}

\begin{titlepage}

\begin{center}
\vspace*{.5cm}
{\Huge {\bfseries The Quicksort algorithm \\ 
and related topics}}~ \\ [2cm]

\begin{figure}[H]
\centering
\includegraphics[height=8cm]
{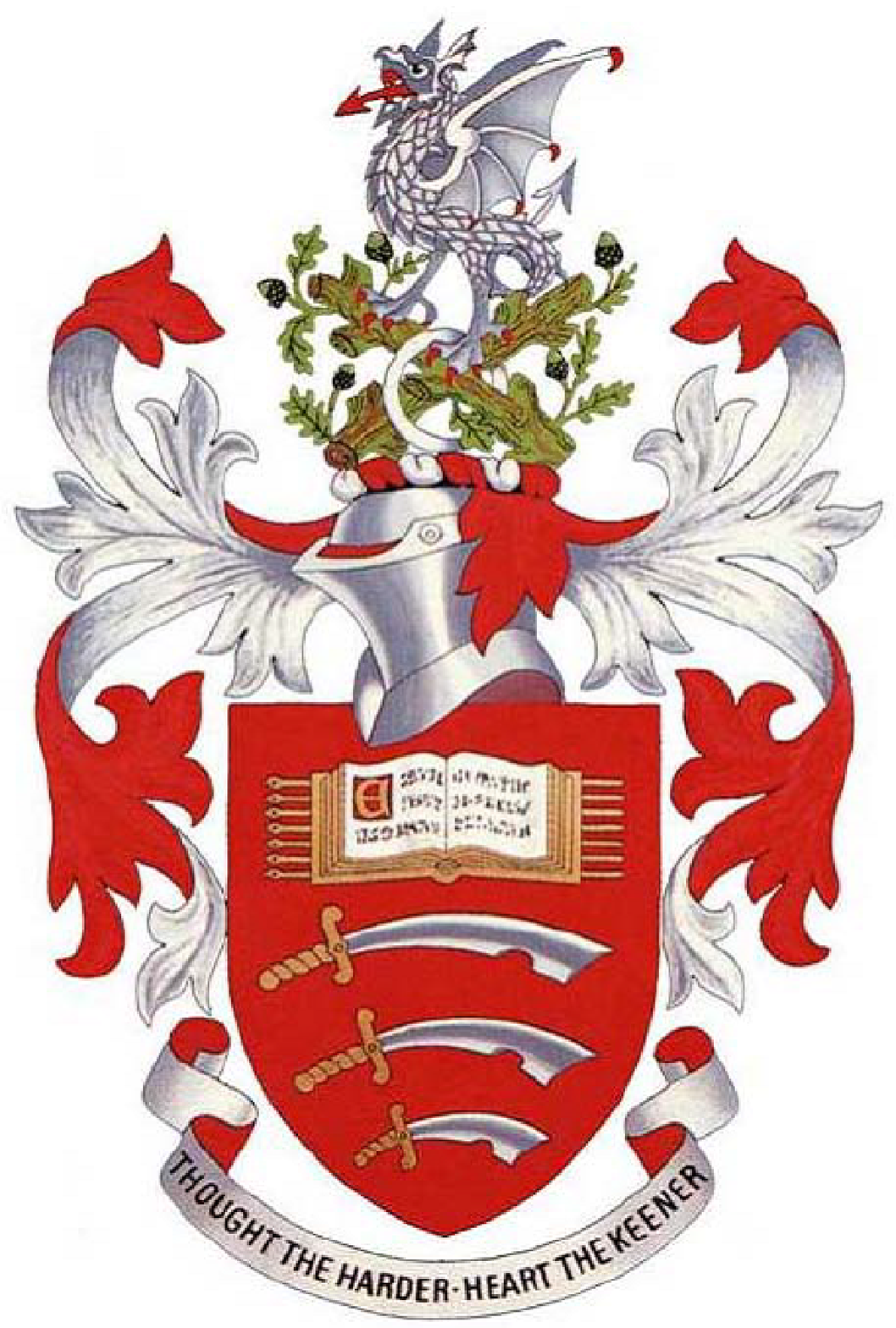} ~ \\ [1.5cm]
\end{figure}

{\large {\bfseries Vasileios Iliopoulos} ~ \\ [1.5cm]
A thesis submitted for the degree of \\
~ Doctor of Philosophy ~ \\ [0.7cm]
Department of Mathematical Sciences  \\
~ University of Essex ~ \\ [0.7cm]
~~~~ June 2013
}

\end{center}

\end{titlepage}

\restoregeometry

\newpage
\setcounter{page}{2}
\thispagestyle{empty}
\mbox{}
\newpage

\setstretch{1.7}

\section*{Abstract}

Sorting algorithms have attracted a great deal of attention and 
study, as they have numerous applications to Mathematics,
Computer Science and related fields. In this thesis, we first deal with the
mathematical analysis of the Quicksort algorithm and its variants.
Specifically, we study the
time complexity of the algorithm and we provide a
complete demonstration of the variance of the number of comparisons required,
a known result but one whose detailed proof is not easy to read out of the
literature. We also examine variants of Quicksort, where multiple pivots are chosen
for the partitioning of the array.

The rest of this work is dedicated to the analysis of finding the true order
by further pairwise comparisons when a partial order compatible with the true order
is given in advance. We discuss a number of cases where the partially 
ordered sets arise at random. To this end, we employ results 
from Graph and Information Theory. 
Finally, we obtain an alternative bound on the number of linear 
extensions when the partially ordered set arises from a random graph, and discuss
the possible application of Shellsort in merging chains.

\newpage
\thispagestyle{empty}
\mbox{}
\newpage

\section*{Acknowledgements}

I would like to thank Dr. David B. Penman for his meticulous advise and guidance, 
as I pursued this work. Our fruitful conversations helped me to clearly 
express my ideas in this thesis. Thanks are also due to Dr. Gerald Williams for his
comments and suggestions during board meetings.

I dedicate this work to my family. Without their support, 
I wouldn't be able to carry out and accomplish my research.

\newpage

\tableofcontents

\newpage

\pagestyle{main}

\pagenumbering{arabic}

\chapter{Preface} 

The first Chapter serves as an introduction to this work, where in a simple manner, 
useful notions of algorithmic analysis in 
general are presented, along with definitions that will be 
used throughout the thesis. We present an introduction to Quicksort and this Chapter
ends with an outline and a summary of the main contributions of this thesis.

\section{Preliminaries}

An algorithm is a valuable tool for solving computational problems. It is a
well defined procedure that takes some value or a set of values as input and
is guaranteed to produce an answer in finite time. See e.g. \cite{intro}.
However the time taken may be impractically long.

As a useful introduction, we present a simple and intuitive algorithm known
as Euclid's algorithm which is still used today to determine the
greatest common divisor of two integers. Its definition follows \cite{knuth 2}:
\begin{Definition}
Given two integers $A \geq C$ greater than unity, we want to find their greatest
common divisor, g.c.d.(A, C).
\begin{enumerate}
\item If C divides A, then the algorithm terminates with C as the greatest common divisor. 
\item If $A\bmod C$ is equal to unity, the numbers are either prime or relatively prime and the
algorithm terminates. Otherwise set $A\leftarrow C$, $C\leftarrow A\bmod C$ 
and return to step $1$.
\end{enumerate}
\end{Definition}
By $A\bmod C$, we denote the remainder of the division of $A$ by $C$, namely
\begin{equation*}
A\bmod C:=A-\left (C\cdot \left \lfloor {\frac{A}{C}} \right \rfloor \right),
\end{equation*}
which is between $0$ and $C-1$. Here the floor function $\lfloor x\rfloor$ of
a real number $x$ is the largest integer
less than or equal to $x$. We observe that this algorithm operates
recursively by successive divisions, until obtaining remainder equal to $0$
or to $1$. Since the remainder strictly reduces at each stage, the process
is finite and eventually terminates.

Two main issues in relation to an algorithm are its running time or time
complexity, which is the amount of time necessary to solve the
problem, and its space complexity, which is the amount of memory needed for
the execution of the algorithm in a computer. In case of Euclid's algorithm,
$3$ locations of memory are required for storing the integer numbers $A$,
$C$ and $A\bmod C$. In this thesis, we will mainly be concerned with time
complexity.

The aim is usually to relate the time complexity to some measure of the
size of the instance of the problem we are considering. Very often we are particularly
concerned with what happens when the size $n$ is large -- tending to infinity.
The notations $O$, $\Omega$, $\Theta$
and $o$, $\omega$ are useful in this context.

\section{Definitions}

In this section, we present primary definitions to the analysis presented in the thesis. 
These definitions come from \cite{intro}, \cite{con}.
\begin{Definition}
Let $f(n)$ and $g(n)$ be two functions with $n \in \mathbf N$. 
We say that $f(n) = O\bigl(g(n)\bigr)$ 
if and only if there exists a constant $c>0$ and $n_{0} \in \mathbf N$, such that
$\vert f(n)\vert \leq  c \cdot \left \vert g(n)\right \vert$, $\forall n \ge n_{0}$.
\end{Definition}
\begin{Definition}
Also, we say that $f(n)= \Omega\bigl(g(n)\bigr)$ if and only if there exists a constant $c > 0$ 
and $n_{0} \in \mathbf N$, such that
$\vert f(n)\vert \geq  c \cdot \vert g(n)\vert$, $\forall n \ge n_{0}$. 
\end{Definition}
\begin{Definition}
Furthermore, we say that $f(n) =\Theta \bigl(g(n) \bigr)$ if and only if there exist 
positive constants $c_{1}$, $c_{2}$ and $n_{0} \in \mathbf N$, such that 
$c_{1} \cdot \vert g(n)\vert \leq \vert f(n) \vert \leq c_{2} 
\cdot \vert g(n)\vert$, $\forall n \ge n_{0}$.
Equivalently, we can state that if $f(n) =  O \bigl(g(n) \bigr)$ and $f(n) =  
\Omega \bigl(g(n) \bigr)$, then
$f(n) =  \Theta \bigl(g(n) \bigr)$.
\end{Definition}

We will also use the notations  $o$, $\omega$. They provide the
same kind of limiting bounds with the respective upper case notations. The
difference is that for two functions $f(n)$ and $g(n)$, the upper case
notation holds when it does exist some positive constant $c$. Whereas, the respective lower 
case notation is true for every positive constant $c$ \cite{intro}.
In other words, $o$ is stronger statement than $O$, since $f(n) = o\bigl (g(n) \bigr)$
implies that $f$ is dominated by $g$. Equivalently, this can be stated as
\begin{eqnarray*}
f(n) = o \bigl (g(n) \bigr) \Longleftrightarrow \lim_{n \to \infty} \frac{f(n)}{g(n)} = 0,
\end{eqnarray*}
provided that $g(n)$ is non-zero. 

The relation $f(n) = \omega \bigl (g(n) \bigr)$ implies that $f$ dominates $g$, i.e.
\begin{eqnarray*}
f(n) = \omega \bigl (g(n) \bigr) \Longleftrightarrow \lim_{n \to \infty} \frac{f(n)}{g(n)} = \infty.
\end{eqnarray*}
The relation $f(n) \sim g(n)$ denotes the fact that $f(n)$ and $g(n)$ 
are asymptotically equivalent, i.e.
\begin{equation*}
\lim_{n \to  \infty} \dfrac{f(n)}{g(n)} =1.
\end{equation*}
Further, best, worst and average case performance denote the resource usage,
e.g. amount of memory in computer or running time, of a given algorithm
at least, at most and on average, respectively.
In other words, these terms describe the behaviour of an algorithm under
optimal circumstances (e.g. best--case scenario), worst circumstances and on
average \cite{intro}. In this work, the study will be concentrated on the average
and worst case performance of Quicksort and its variants.

In our analysis, we will frequently come across with harmonic numbers, whose
definition we now present. 
\begin{Definition}
The sum
\begin{eqnarray*} 
H_{n}^{(k)}:=\sum_{i=1}^{n}\frac{1}{i^{k}}=
\frac{1}{1^{k}}+\frac{1}{2^{k}}+\ldots+\frac{1}{n^{k}}
\end{eqnarray*}
is defined to be the generalised $n_{th}$
harmonic number of order $k$. When $k=1$, the sum denotes the $n_{th}$
harmonic number, which we simply write $H_{n}$. We define also $H_{0}:=0$.
\end{Definition}

There are numerous interesting properties of harmonic numbers, which are
not yet fully investigated and understood. Harmonic series have links with
Stirling numbers \cite{con} and arise frequently in the analysis of
algorithms. For $n$ large, it is well-known that \cite{knuth 1}, 
\begin{eqnarray*}
H_{n}= \log_{e}(n)+\gamma+\frac{1}{2n}-\frac{1}{12n^{2}
}+\frac{1}{120n^{4}}+O \left(\frac{1}{n^{6}} \right),
\end{eqnarray*}
where $\gamma=0.57721\ldots$ is the Euler--Mascheroni constant. We will
use this often, especially in the form $H_{n} = \log_{e}(n)+ \gamma + o(1)$.

Note that throughout this thesis, we shall adopt the convention of writing 
explicitly the base of logarithms.
For example, the natural logarithm of $n$ is denoted by $\log_{e}(n)$, instead of $\ln(n)$.
Also, the end of a proof will be denoted by the symbol ~\qedsymbol ~.

\section{Introduction to Quicksort}

Sorting an array of items is clearly a fundamental
problem, directly linked to efficient searching with numerous
applications. The problem is that given an array of keys, we want
to rearrange these in non-decreasing order. Note that the order 
may be numerical, alphabetical or any other transitive relation defined on the
keys \cite{knuth 3}. In this work, the analysis deals with numerical order,
where the keys are decimal numbers and
we particularly focus on Quicksort algorithm and variants of it. 
Quicksort was invented by C. A. R. Hoare \cite{hoar,hoare}.
Here is the detailed definition.
\begin{Definition} {\ \\} 
The steps taken by the Quicksort algorithm are:
\begin{enumerate}
\item Choose an element from the array, called pivot. 

\item Rearrange the array by comparing every element to the pivot,
so all elements smaller than or equal to the pivot come before the pivot and all 
elements greater than or equal to the pivot come after the pivot.

\item Recursively apply steps $1$ and $2$ to the subarray of the elements 
smaller than or equal to the pivot
and to the subarray of the elements greater than or equal to the pivot.
\end{enumerate}
\end{Definition}
Note that the original problem is divided into smaller ones, with
(initially) two subarrays, the keys smaller than the pivot, and those bigger
than it. Then recursively these are divided into smaller subarrays by further
pivoting, until we get trivially sorted subarrays, which 
contain one or no elements. Given an array of $n$ distinct keys 
$A=\{a_{1}, a_{2}, \ldots, a_{n}\}$
that we want to quick sort, with all the $n!$ permutations
equally likely, the aim is to finding the unique permutation
out of all the $n!$ possible, such that the keys are in increasing order. 
The essence of Quicksort is the 
partition operation, where by a series of pairwise
comparisons, the pivot is brought to its final place, with smaller
elements on its left and greater elements to the right. Elements
equal to pivot can be on either or both sides.

As we shall see, there are numerous partitioning schemes, and while the
details of them are not central to this thesis, we should describe the
basic ideas. A straightforward
and natural way (see e.g. \cite{knuth 3}) uses two pointers -- a left
pointer, initially at the left end of the array and a right pointer,
initially at the right end of the array. We pick the leftmost element of the
array as pivot and the right pointer scans from the right end of the
array for a key less than the pivot. If it finds such a key,
the pivot is swapped with that key. Then, the left pointer is increased
by one and starts its scan,
searching for a key greater than the pivot: if 
such a key is found, again the pivot is exchanged with it. When the
pointers are crossed, the pivot by repeated exchanges will ``float'' to its
final position and the keys which are on its left are smaller and keys on
its right are greater. The data movement of this scheme is quite large,
since the pivot is swapped with the other elements.

A different partitioning scheme, described in \cite{hoare} is the following.
Two pointers $i$ (the left pointer, initially $1$) and $j$ (the right pointer,
initially $n$) are set and a key is arbitrarily chosen as pivot. The left
pointer goes to the right until a key is found which is greater than the
pivot. If one is found, its scan is stopped and the right pointer
scans to the left until a key less than the pivot is found. If such a key
is found, the right pointer stops and those two keys are
exchanged. After the exchange, both pointers are stepped down one
position and the lower one starts its scan. When pointers are
crossed, i.e. when $i\geq j$, the final exchange places the pivot in
its final position, completing the partitioning.
The number of comparisons required to partition an array of $n$ keys is
at least $n-1$ and the expected number of exchanges is $\frac{n}{6}+ \frac{5}{6n}$.

A third partitioning routine, called Lomuto's partition, is mentioned in
\cite{bentley} -- this involves exactly $n-1$ comparisons, which
is clearly best possible, but the downside is the increased number of exchanges.
The expected number of key exchanges of this scheme is $\frac{n-1}{2}$, 
\cite{mahmoud}. 

We now consider the
worst case and best case, analysis of Quicksort. 
Suppose we want to sort the following array,
$\{a_{1} < a_{2} < \ldots < a_{n}\}$
and we are very unlucky and our initial choice of pivot is
the largest element $a_{n}$. Then of course we only
divide and conquer in a rather trivial sense: every element is below
the pivot, and it has taken us $n-1$ comparisons with $a_{n}$ to get
here. Suppose we now try again and are unlucky again, choosing
$a_{n-1}$ as pivot this time. Again the algorithm performs $n-2$
comparisons and we are left with everything less than $a_{n-1}$. If
we keep being unlucky in our choices of pivot, and keep choosing the
largest element of what is left, after $i$ recursive calls the running time of the algorithm
will be equal to $(n-1)+(n-2)+\ldots +(n-i)$ comparisons, so the
overall number of comparisons made is 
\begin{equation*}
1+2+\ldots +(n-1)=\frac{n\cdot(n-1)}{2}.
\end{equation*}
Thus Quicksort needs quadratic time to sort already sorted or reverse-sorted
arrays if the choice of pivots is unfortunate. 

If instead we always made good choices,
choosing each pivot to be roughly in the middle of the array we are
considering at present, then in the first round we make $n-1$ comparisons,
then in the two subarrays of size about $n/2$ we make about
$n/2$ comparisons, then in each of the four subarrays of size about $n/4$ we
make $n/4$ comparisons, and so on. So we make about $n$ comparisons in total
in each round. The number of rounds will be roughly $\log_{2}(n)$ as we are
splitting the arrays into roughly equally-sized subarrays at each stage,
and it will take $\log_{2}(n)$ recursions of this to get down to trivially
sorted arrays. 

Thus, in this good case we will need $O\bigl(n\log_{2}(n) \bigr )$
comparisons. This is of course a rather informal argument, but does illustrate
that the time complexity can be much smaller than the quadratic run-time in
the worst case. This is already raising the question of what the typical
time complexity will be: we address this in the next Chapter. 

We briefly discuss the space complexity of the algorithm. There are $n$
memory locations occupied by the keys. Moreover,
the algorithm, due to its recursive nature, needs additional space for the
storage of subarrays. The subarrays' boundaries are saved on to a stack,
which is a data structure providing temporary storage. At the end of the
partition routine, the pivot is placed in its final position between two
subarrays (one of them possibly empty). Recursively, the algorithm is
applied to the smaller subarray and the other one is pushed on to stack.
Since, in best and average
case of Quicksort, we have $O\bigl(\log_{2}(n)\bigr)$ recursive calls,
the required stack space is $O\bigl(\log_{2}(n)\bigr)$ locations in memory.
However, in worst case the stack may require $O(n)$ locations, if the
algorithm is applied to the larger subarray and the smaller one is saved
to the stack \cite{sedgewick}. 

This discussion makes it clear that the pivot selection plays a vital
role in the performance of the algorithm. Many authors have
proposed various techniques to remedy this situation and to avoid worst case
behaviour, see \cite{intro}, \cite{hoare}, \cite{knuth 3}, \cite{Scow},
\cite{sedgewick} and \cite{Singleton}. These include the random
shuffling of the array prior to initialisation of the algorithm, choosing
as pivot the median of the array, or the median of a random sample of
keys.

Scowen in his paper \cite{Scow}, suggested choosing as pivot the middle
element of the array: his variant is dubbed ``Quickersort''. Using this rule
for the choice of partitioning element, the aim is the splitting of the array
into two halves of equal size. Thus, in case where the array is nearly
sorted,
quadratic time is avoided but if the chosen pivot is the minimum or
maximum key, the algorithm's running time attains its worst case and this
variant does not offer any more than choosing the pivot randomly. Singleton
\cite{Singleton} suggested a better estimate of the median, by selecting as
pivot the median of leftmost, rightmost and middle keys of the input array.
Hoare \cite{hoare} suggested the pivot may be chosen as the median of a
random sample from the keys to be sorted, but he didn't analyse this approach.

One point is that Quicksort is not always very fast at sorting small arrays.
Knuth \cite{knuth 3} presented and analysed a partitioning scheme, which takes
$n+1$ instead of $n-1$ comparisons and the sorting of small subarrays (usually
from about 9 to 15 elements) is implemented using insertion sort,
since the recursive structure of Quicksort is better suited to large arrays.
Insertion sort is a simple sorting algorithm, which gradually
`constructs' a sorted array from left to right,
in the following manner. The first two 
elements are compared and exchanged, in case that are not in order. Then, the
third element is compared with the element on its left. If it is greater, it 
is left at its initial location, otherwise is compared with the first element and 
accordingly is inserted to its position in the sorted array of 3 elements. 
This process is iteratively applied to the
remaining elements, until the array is sorted. 
See as well in Cormen {\em et al.} \cite{intro}, 
for a description of the algorithm.

\section{Outline and contributions of thesis}

This thesis consists of seven Chapters and one Appendix. After the first, 
introductory Chapter, the rest of the thesis is organised as follows:

In {\bf Chapter 2}, we consider the first and second moments of the number of
comparisons made when pivots are chosen randomly. The result for the mean
is known and easy: the result for the variance is known, but less easy to
find a full proof of in the literature. We supply one. 
We briefly discuss the skewness of the number of comparisons and we
study the asymptotic behaviour of the algorithm.  

In {\bf Chapter 3}, we analyse the idea of choosing the pivot as a centered
statistic of a random  sample of the keys to be sorted and we obtain the
average number of comparisons required by these
variants, showing that the running time can be greatly improved. Moreover, we
present relevant definitions of entropy. Not much of this is original, but
some details about why various differential equations that arise in the
analysis have the solutions they do (i.e. details about roots of indicial polynomials) 
are not in literature. 

In {\bf Chapter 4}, we analyse extensions of Quicksort, where multiple pivots
are used for the partitioning of the array.
The main contributions in this Chapter are in sections {\bf 4.1} 
and {\bf 4.2}. The results in the former section were
published in the paper \cite{iliop}, where the expected costs related to the time 
complexity and the 
second moment of the number of comparisons are computed.
The latter section contains the analysis of the generalisation of the
algorithm. We study the general recurrence model, giving the expected cost of the variant,
provided that the cost during partitioning is linear, with respect to the number of keys. 
We also present the 
application of Vandermonde matrices for the computation of the constants 
involved to the cost of these variants. 

In {\bf Chapter 5}, various cases of partially ordered sets are discussed and the number of comparisons 
needed for the complete sorting is studied. The `information--theoretic lower bound'
is always $\omega(n)$ in these cases and we show that the time needed for the sorting of partial orders
is $O\bigl(n\log_{2}(n)\bigr)$. The main contribution of this Chapter,
is the derivation of the asymptotic number of comparisons needed, for the
sorting of various partially ordered sets. The basic ideas used here
are due to, amongst others, Cardinal {\em et al.} \cite{Cardinal}, Fredman \cite{fred}, 
Kahn and Kim \cite{kk}, Kislitsyn \cite{kisl}, but the working out of the detailed consequences for these
partial orders seems to be new. 

In {\bf Chapter 6}, we consider random graph orders, where the
`information--theoretic lower bound' is of the same order of magnitude as
the number of keys being sorted. We derive a new bound on the number of
linear extensions using entropy arguments, though it is not at present
competitive with an older bound in the literature \cite{abbj}. 

In {\bf Chapter 7}, we conclude the thesis, presenting future research directions. 
At this final Chapter, we 
derive another bound of the number of comparisons required to sort a random interval order and
we discuss the merging of linearly ordered sets.

In {\bf Appendix A}, we present the \textsc{Maple} calculations, regarding the derivation of
the variance of the number of comparisons of dual pivot Quicksort, analysed
in subsection {\bf 4.1.1}.

\chapter{Random selection of pivot}

In this Chapter, the mathematical analysis of Quicksort is presented, under the
assumption that the pivots are uniformly selected at random. Specifically, 
the major expected costs regarding the time complexity of the algorithm and the second
moment are computed. The derivation of the average costs is 
unified under a general recurrence relation, demonstrating the amenability of the algorithm
to a complete mathematical treatment. We also address the asymptotic analysis of the algorithm 
and we close this Chapter considering the presence of equal keys.

\section{Expected number of comparisons}

This discussion of lucky and unlucky choices of pivot suggests the idea
of selecting the pivot at random, as randomisation often helps to improve
running time in algorithms with bad worst-case, but good average-case
complexity \cite{de}. For example, we could choose the pivots 
randomly for a discrete uniform distribution on the array
we are looking at each stage. Recall that the uniform distribution on a finite set assigns
equal probability to each element of it.
\begin{Definition}
$C_{n}$ is the random variable giving the number of comparisons 
in Quicksort of $n$ distinct elements when
all the $n!$ permutations of the keys are equiprobable. 
\end{Definition}

It is clear that for $n=0$ or $n=1$, $C_{0} = C_{1} = 0$ as there is
nothing to sort. These are the initial or ``seed'' values of the recurrence relation for
the number of comparisons, given in the following Lemma.
\begin{Lemma}
The random number of comparisons $C_{n}$ for the sorting
of an array consisting of $n \geq 2$ keys, is given by
\begin{equation*}
C_{n}=C_{U_{n}-1}+{C^{\star}_{n-U_{n}}} + n-1,
\end{equation*}
where $U_{n}$ follows the uniform distribution over the set $\{1, 2,
\ldots ,n\}$ and ${C^{\star}_{n-U_{n}}}$ is identically distributed
to $C_{U_{n}-1}$ and independent of it conditional on $U_{n}$. 
\end{Lemma}
\begin{proof}
The choice of $U_{n}$ as pivot, and comparing the other
$n-1$ elements with it, splits the array into two subarrays. There is 
one subarray of all $U_{n}-1$ elements
smaller than the pivot and another one of all $n-U_{n}$
elements greater than the pivot. Obviously
these two subarrays are disjoint. Then recursively two pivots are randomly
selected from the two subarrays, until the array is sorted, and so we get the
equation. 
\end{proof}

This allows us to find that the expected complexity of Quicksort applied to $n$
keys is:
\begin{align*}
\mathbb {E}(C_{n})& = \mathbb {E}(C_{U_{n}-1}+{C^{\star}_{n-U_{n}}}+n-1) \\
& = \mathbb {E}(C_{U_{n}-1})+ \mathbb {E}({C^{\star}_{n-U_{n}}})+n-1.
\end{align*}
Using conditional expectation and noting that $U_{n}=k$ has
probability $1/n$, we get, writing
$a_{k}$ for $\mathbb {E}(C_{k})$,
that
\begin{eqnarray*}
a_{n}=\sum_{k=1}^{n}\frac{1}{n}(a_{k-1}+a_{n-k})+n-1
\Longrightarrow a_{n}=\frac{2}{n}\sum_{k=0}^{n-1}a_{k}+n-1. 
\end{eqnarray*}
We have to solve this recurrence relation, in order to obtain a closed form
for the expected number of comparisons. The following result is well-known (e.g. see
in \cite{intro}, \cite{sedgewick}):
\begin{theorem}
The expected number $a_{n}$ of comparisons for Quicksort with uniform
selection of pivots is $a_{n}=2(n+1)H_{n}-4n$.
\end{theorem}
\begin{proof}
We multiply both sides of the formula for $a_{n}$ by $n$, getting 
\begin{eqnarray*}
na_{n}=2\sum_{k=0}^{n-1}a_{k}+n(n-1)
\end{eqnarray*}
and similarly, multiplying by $n-1$, 
\begin{eqnarray*}
(n-1)a_{n-1}=2\sum_{k=0}^{n-2}a_{k}+(n-2)(n-1).
\end{eqnarray*}

Subtracting $(n-1)a_{n-1}$ from $na_{n}$ in order to eliminate the sum -- 
see \cite{knuth 3}, we obtain
\begin{align*}
na_{n}-(n-1)a_{n-1}&=2a_{n-1}+2(n-1) \\
& \Longrightarrow na_{n}=(n+1)a_{n-1}+2(n-1) \\
& \Longrightarrow \frac{a_{n}}{n+1}=\frac{a_{n-1}}{n}+\frac{2(n-1)}{n(n+1)}.
\end{align*}
``Unfolding'' the recurrence we get 
\begin{align*}
\frac{a_{n}}{n+1}&=2\sum_{j=2}^{n}\frac{(j-1)}{j(j+1)}=2\sum_{j=2}^{n}\left 
(\frac{2}{j+1}-\frac{1}{j}\right) \\
& =4H_{n}+\frac{4}{n+1}-4-2H_{n} \\
& =2H_{n}+\frac{4}{n+1}-4.
\end{align*}
Finally, 
\begin{equation*}
a_{n}=2(n+1)H_{n}-4n. \qedhere
\end{equation*}
\end{proof}

We now show a slick way of solving the recurrence about the 
expected number of comparisons using generating
functions. This approach is also noted in various places, e.g. \cite{knuth 1},
\cite{nip}, \cite{sedgewick}. We again start from 
\begin{eqnarray*}
a_{n}=\frac{2}{n}\sum_{j=0}^{n-1}a_{j}+n-1.
\end{eqnarray*}
We multiply through by $n$ to clear fractions, getting
\begin{align*}
& na_{n} =n(n-1)+2\sum_{j=0}^{n-1}a_{j} \\
& \quad {}\Longrightarrow \sum_{n=1}^{\infty}na_{n}x^{n-1}=\sum_{n=1}^{\infty}n(n-1)x^{n-1}
+2\sum_{n=1}^{\infty}\sum_{j=0}^{n-1}a_{j}x^{n-1} \\
& \quad {} \Longrightarrow \sum_{n=1}^{\infty}na_{n}x^{n-1}=x\sum_{n=2}^{\infty}n(n-1)x^{n-2}
+2\sum_{n=1}^{\infty}\sum_{j=0}^{n-1}a_{j}x^{n-1}.
\end{align*}
Letting $f(x)=\sum_{n=0}^{\infty}a_{n}x^{n}$,
\begin{eqnarray*}
f^{\prime}(x)=x\sum_{n=2}^{\infty}n(n-1)x^{n-2}
+2\sum_{j=0}^{\infty}\sum_{n=j+1}^{\infty}a_{j}x^{n-1},
\end{eqnarray*}
changing round the order of summation. To evaluate 
$\displaystyle \sum_{n=2}^{\infty}n(n-1)x^{n-2}$, 
simply note this is the
2nd derivative of $\displaystyle \sum_{n=0}^{\infty}x^{n}$. Since the latter sum is equal to
$1/(1-x)$, for $\vert x\vert <1$, its second derivative is easily checked to
be $2(1-x)^{-3}$. Multiplying the sum by $x$ now gives
\begin{eqnarray*}
f^{\prime}(x)=2x(1-x)^{-3}+
2\sum_{j=0}^{\infty}a_{j}\sum_{n=j+1}^{\infty}x^{n-1}.
\end{eqnarray*}

We evaluate the last double-sum. The inner sum is of course $x^{j}/(1-x)$
being a geometric series. Thus we get
\begin{align*}
f^{\prime}(x)&=2x(1-x)^{-3}+
\frac{2}{1-x}\sum_{j=0}^{\infty}a_{j}x^{j} \\
&=2x(1-x)^{-3}+\frac{2}{1-x}f(x).
\end{align*}
This is now a fairly standard kind of differential equation. Multiplying
both sides by $(1-x)^{2}$, we see
\begin{align*}
(1-x)^{2}f^{\prime}(x)&=2x(1-x)^{-1}+2(1-x)f(x) \\
\Longrightarrow &(1-x)^{2}f^{\prime}(x)-2(1-x)f(x)=\frac{2}{1-x}-2 \\
\Longrightarrow &\left((1-x)^{2}f(x)\right)^{\prime}=\frac{2}{1-x}-2 \\
\Longrightarrow &(1-x)^{2}f(x)=-2\log_{e}(1-x)-2x+c.
\end{align*}
Setting $x=0$, we get $f(0)=0$ on the left-hand side, and on the
right-hand side we get $c$, so $c=0$. Therefore
\begin{eqnarray*}
f(x)=\frac{-2\Bigl(\log_{e}(1-x)+x\Bigr)}{(1-x)^{2}}.
\end{eqnarray*}

Expanding out $\log_{e}(1-x)$ as a series, $-x-x^{2}/2-x^{3}/3-\ldots $, 
and similarly
writing $1/(1-x)^{2}$ as the derivative of $1/(1-x)$, we obtain
\begin{eqnarray*}
f(x)=2\sum_{k=2}^{\infty}\frac{x^{k}}{k}\sum_{j=0}^{\infty}(j+1)x^{j}.
\end{eqnarray*}
Thus, looking at coefficients of $x^{n}$ on both sides, we get on the
left-hand side $a_{n}$. On the right-hand side, we get the coefficient for each
$x^{j}$ in the first series (which is $1/j$) times the term for the
$x^{n-j}$ in the other, namely $(n-j+1)$. So we get
\begin{align*}
a_{n}&= 2\sum_{k=2}^{n}\frac{n-k+1}{k}=2\sum_{k=2}^{n}\frac{n+1}{k}-2\sum_{k=2}^{n}1 \\
&= 2(n+1)\biggl(\sum_{k=1}^{n}\frac{1}{k}-1 \biggr) -2(n-1) \\
&= 2(n+1)H_{n}-4n.
\end{align*}

Some texts give a
different argument for this, as follows. 
\begin{theorem}[Mitzenmacher and Upfal \cite{prob}]
Suppose that a pivot is chosen independently and uniformly at random from an
array of $n$ keys, in which Quicksort is applied. Then, for any
input, the expected number of comparisons made by randomised Quicksort
is $2n \log_{e}(n) + O(n)$.
\end{theorem}
\begin{proof} 
Let $\{x_{1}, x_{2}, \ldots, x_{n}\}$ be the input
values and the output, after the termination of Quicksort, (i.e.
keys in increasing order) be a permutation of the initial array $\{y_{1}, y_{2},
\ldots, y_{n}\}$. For $i<j$, let $X_{ij}$ be a $\{0,1\}$-valued
random variable, that takes the value $1$, if $y_{i}$ and $y_{j}$ are
compared over the course of algorithm and $0$, otherwise. Then, the
total number of comparisons $X$ satisfies
\begin{align*}
X &= \sum_{i=1}^{n-1}\sum_{j=i+1}^{n}X_{ij} \\
\Longrightarrow \mathbb {E}(X) &= \mathbb {E}\left(\sum_{i=1}^{n-1}\sum_{j=i+1}^{n}X_{ij}\right)
=\sum_{i=1}^{n-1}\sum_{j=i+1}^{n}\mathbb {E}(X_{ij}). 
\end{align*}
Since $X_{ij}\in \{0,1\}$, then $\mathbb{E}(X_{ij})=\mathbb{P}(X_{ij}=1)$. This
event occurs when $y_{i}$ and $y_{j}$ are compared. Clearly, this
happens when $y_{i}$ or $y_{j}$ is the first number chosen as pivot from the
set $Y=\{y_{i}, y_{i+1}, \ldots, y_{j-1}, y_{j}\}$. (Because otherwise
if some element between them is chosen, $y_{k}$ say, $y_{i}$ and $y_{j}$
are compared to $y_{k}$ and $y_{i}$ is put below $y_{k}$ and $y_{j}$ above it
with the result that $y_{i}$ and $y_{j}$ are never compared).

Since the pivot is chosen uniformly at random, the probability the one of these two elements
is the first of the $j-i+1$ elements chosen is equal to
$\displaystyle \frac{1}{j-i+1}+\frac{1}{j-i+1}= \frac{2}{j-i+1}$. Substituting
$k=j-i+1$, we obtain
\begin{eqnarray*}
\mathbb {E}(X)=\mathbb {E} \left(\sum_{i=1}^{n-1}\sum_{j=i+1}^{n}\frac{2}{j-i+1}\right)
=\sum_{i=1}^{n-1}\sum_{k=2}^{n-i+1}\frac{2}{k}
=\sum_{k=2}^{n}\sum_{i=1}^{n+1-k}\frac{2}{k}.
\end{eqnarray*}
The change of sum is justified by writing out the possible rows for fixed
$i$ and columns for a fixed $k$, then changing round from summing
rows first to summing columns first. This is, as $2/k$ does not depend
on $i$, equal to
\begin{eqnarray*}
\sum_{k=2}^{n}(n+1-k)\frac{2}{k}=(n+1)\sum_{k=2}^{n}\frac{2}{k}-2\sum_{k=2}^{n}1.
\end{eqnarray*}
Thus,
\begin{align*}
2(n+1)\sum_{k=2}^{n}\frac{1}{k}-2n+2 & =2(n+1)(H_{n}-1)-2n+2 \\
& =2(n+1)H_{n}-4n,
\end{align*}
proving the claim. 
\end{proof}

Also of some interest is the mean number of partitioning stages $\mathbb{E}(P_{n})$ of the
algorithm applied to $n$ keys as input. For the case where we simply use Quicksort 
for all the sorting,
it is obvious that we will have $P_{0}=0$ and for $n \geq 1$, the number of
partitioning stages $P_{n}$ obeys the following recurrence conditional on
$j$ being the rank of the pivot
\begin{equation*}
P_{n}=1+P_{j-1}+P_{n-j}.
\end{equation*}
Observe that $P_{1}=1$ and $P_{2}=2$.
Taking expectations, and noting that the pivot is uniformly chosen at 
random and that the two sums $\sum_{j}\mathbb{E}(P_{j-1})$ and
$\sum_{j}\mathbb{E}(P_{n-j})$ are equal, we see 
\begin{equation*}
\mathbb {E}(P_{n})=1+\frac{2}{n}\sum_{j=0}^{n-1}\mathbb {E}(P_{j}).
\end{equation*}
Multiplying both sides by $n$ and differencing the recurrence relation, 
as we did for the derivation of the expected number of comparisons, we have
\begin{align*}
n\mathbb {E}(P_{n})-(n-1)\mathbb {E}(P_{n-1})& = 1+2\mathbb {E}(P_{n-1}) \\
\Longrightarrow \dfrac{\mathbb {E}(P_{n})}{n+1}& =\dfrac{1}{n(n+1)}+\dfrac{\mathbb {E}(P_{n-1})}{n} \\
\Longrightarrow \mathbb {E}(P_{n})& = (n+1)\biggl (\sum_{j=2}^{n}
\Bigl(\dfrac{1}{j}-\dfrac{1}{j+1}\Bigr)+\dfrac{1}{2} \biggr).
\end{align*}
Finally, $\mathbb{E}(P_{n})=n$.

\section{Expected number of exchanges}

Here we consider the number of exchanges or swaps performed by the algorithm,
which is mentioned by Hoare \cite{hoare} as a relevant quantity.
We assume that each swap has a fixed cost and as in the previous
section, we assume that the keys are distinct and that all $n!$ permutations
are equally likely to be the input: this in particular implies that the pivot is 
chosen uniformly at random from the array.

We should specify the partitioning procedure. Assume that we have to sort
$n$ distinct keys, where their locations in the array are 
numbered from left to right by
$1, 2, \ldots, n$.
Set two pointers $i \leftarrow 1$ and
$j \leftarrow n-1$ and select the element at location $n$ as a pivot. First, compare
the element at location $1$ with the pivot. If this key is less than the 
pivot, increase $i$ by one until an
element greater than the pivot is found. If an element greater than the
pivot is found, stop and compare the element at location $n-1$ with the
pivot. If this key is greater than the pivot, then
decrease $j$ by one and compare the next element to the pivot. 
If an element less than the pivot is found, then the $j$ pointer stops its
scan and the keys that the two pointers refer are exchanged.

Increase $i$ by one, decrease $j$ by one and 
in the same manner continue the scanning of the array until $i \geq j$.
At the end of the partitioning operation, the pivot is placed in
its final position $k$, where $1 \leq k \leq n$, and Quicksort is recursively
invoked to sort the subarray of $k-1$ keys less than the pivot and the
subarray of $n-k$ keys greater than the pivot \cite{hoar, hoare}.

Note that the probability of a
key being greater than the pivot is
\begin{eqnarray*}
\frac{n-k}{n-1}.
\end{eqnarray*}
The number of keys which are greater than pivot, and were moved during
partition is 
\begin{equation*}
\frac{n-k}{n-1}\cdot (k-1).
\end{equation*}
Therefore, considering also that pivots are uniformly chosen and noting that
we have to count the final swap with the pivot at the end of partition
operation, we obtain
\begin{eqnarray*}
\sum_{k=1}^{n}\frac{(n-k)(k-1)}{n(n-1)}+1=\dfrac{n}{6}+\dfrac{2}{3}.
\end{eqnarray*}

Let $S_{n}$ be the total number of exchanges, when the algorithm is 
applied to an array of $n$ distinct keys.
We have that $S_{0}=S_{1}=0$ and 
for $n \geq 2$, the following recurrence holds 
$$S_{n}=\mbox{\sl ``Number~of~exchanges~during~partition~routine''}+S_{k-1}+S_{n-k}.$$ 
Since the pivot is chosen uniformly at random, the recurrence for the expected number of exchanges is
\begin{align*}
\mathbb {E}(S_{n}) &=\frac{n}{6}+\frac{2}{3}+\frac{1}{n}\sum_{k=1}^{n}\Bigl(\mathbb {E}(S_{k-1})
+\mathbb {E}(S_{n-k})\Bigr) \\
\Longrightarrow \mathbb {E}(S_{n})&=\frac{n}{6}+\frac{2}{3} 
+\frac{2}{n}\sum_{k=0}^{n-1}\mathbb {E}(S_{k}).
\end{align*}

This recurrence relation is similar to the recurrences about the mean number
of comparisons and will be solved by the same way. Subtracting $(n-1)\mathbb {E}(S_{n-1})$ 
from $n\mathbb {E}(S_{n})$, the recurrence becomes
\begin{align*}
n\mathbb {E}(S_{n})-(n-1)\mathbb {E}(S_{n-1})&=\frac{2n+3}{6}+2\mathbb {E}(S_{n-1}) \\
\Longrightarrow \frac{\mathbb {E}(S_{n})}{n+1}&=\frac{2n+3}{6n(n+1)}+\frac{\mathbb {E}(S_{n-1})}{n}.
\end{align*}
Telescoping, the last relation yields
\begin{align*}
\frac{\mathbb {E}(S_{n})}{n+1} &=\sum_{j=3}^{n}\frac{2j+3}{6j(j+1)} + \frac{1}{3}=
\sum_{j=3}^{n}\frac{1}{3(j+1)}+\sum_{j=3}^{n}\frac{1}{2j(j+1)} + \frac{1}{3} \\
&=\frac{1}{3}\sum_{j=3}^{n}\frac{1}{j+1} + \frac{1}{2}\Biggl(\sum_{j=3}^{n}\frac{1}{j}-
\sum_{j=3}^{n}\frac{1}{j+1}\Biggr)+ \frac{1}{3} \\
&=\frac{1}{3} \biggl (H_{n+1}- \frac{11}{6}\biggr )
+\frac{1}{2}\biggl(\frac{1}{3} - \frac{1}{n+1} \biggr )+ \frac{1}{3}.
\end{align*}
Tidying up, the average number of exchanges in course of the algorithm is
\begin{eqnarray*}
\mathbb {E}(S_{n})=\frac{(n+1)H_{n}}{3}-\frac{n}{9}-\frac{5}{18}. 
\end{eqnarray*}
Its asymptotic value is 
\begin{equation*}
\frac{n\log_{e}(n)}{3}.
\end{equation*}
It follows that asymptotically the mean number of exchanges is about 
$1/6$ of the average number of comparisons. 

In a variant of the algorithm analysed in
\cite{knuth 3} and \cite{sedgewick}, which we briefly mentioned in the introduction,
partitioning of $n$ keys takes $n+1$ comparisons and subfiles of $m$ or fewer
elements are sorted using insertion sort. Then the average number of
comparisons, partitioning stages and exchanges respectively, are
\begin{align*}
\mathbb {E}(C_{n})&=(n+1)(2H_{n+1}-2H_{m+2}+1), \\
\mathbb {E}(P_{n})&=2\cdot \frac{n+1}{m+2}-1 \mbox{~~and~} \\
\mathbb {E}(S_{n})&=(n+1)\Bigl(\frac{1}{3}H_{n+1}-\frac{1}{3}H_{m+2}+\frac{1}{6}-
\frac{1}{m+2}\Bigr)+\frac{1}{2}, \mbox{~for~$n>m.$}
\end{align*}
For $m=0$, we obtain the average quantities when there is no switch to insertion sort. 
Note that in this case the expected costs are
\begin{align*}
\mathbb {E}(C_{n})&=2(n+1)H_{n}-2n,  \\
\mathbb {E}(P_{n})&= n \mbox{~~and~} \\
\mathbb {E}(S_{n})&=\frac{2(n+1)H_{n}-5n}{6}.
\end{align*}

\section{Variance}

We now similarly discuss the variance of the number of comparisons in
Quicksort. Although the result has been known for many years -- see
\cite{knuth 3}, exercise 6.2.2-8 for an outline -- there is not
really a full version of all the details written down conveniently
that we are aware of, so we have provided such an account -- this summary
has been put on the {\bf arXiv}, \cite{viliopo}. The sources \cite{sk},
\cite{nip} and \cite{sedgewick} were useful in putting the argument together.
Again, generating functions will be used. The result is: 
\begin{theorem}
The variance of the number of comparisons of Quicksort on $n$ keys,
with a pivot chosen uniformly at random is
\begin{eqnarray*}
\operatorname {Var}(C_{n})=7n^{2}-4(n+1)^{2}H_{n}^{(2)}-2(n+1)H_{n}+13n.
\end {eqnarray*}
\end{theorem}

We start with a recurrence for the generating function of $C_{n}$, namely
$\displaystyle f_{n}(z)=\sum_{k=0}^{\infty}\mathbb{P}(C_{n}=k)z^{k}$. We will use this to
reduce the proof of the Theorem to proving a certain recurrence formula
for the expression $\frac{f_{n}^{''}(1)}{2}$.
\begin{theorem}
In Random Quicksort of $n$ keys, the generating functions $f_{i}$ satisfy
$$f_{n}(z)=\frac{z^{n-1}}{n}\sum_{j=1}^{n}f_{j-1}(z)f_{n-j}(z).$$
\end{theorem}
\begin{proof} 
We have, using the equation
$$C_{n}=C_{U_{n}-1}+{C^{*}_{n-U_{n}}} +n-1$$
that
\begin{align*}
\mathbb{P}(C_{n}=k)&=\sum_{m=1}^{n}\mathbb{P}(C_{n}=k\vert U_{n}=m)\frac{1}{n} \\
&=\sum_{m=1}^{n}\sum_{j=1}^{k-(n-1)}\mathbb{P}(C_{m-1}=j)\mathbb{P}(C_{n-m}=k-(n-1)-j)\frac{1}{n},
\end{align*}
noting that $C_{m-1}$ and $C_{n-m}$ are conditionally independent subject to the pivot.
Thus
\begin{align*} 
\mathbb{P}(C_{n}=k)z^{k}=\frac{1}{n}\sum_{m=1}^{n}\sum_{j=1}^{k-(n-1)} 
\mathbb{P}(C_{m-1}=j)z^{j} 
\mathbb{P}(C_{n-m}=k-(n-1)-j)z^{k-(n-1)-j}z^{n-1}.
\end{align*}
Multiplying by $z^{k}$ and summing over $k$, so as to get the generating
function $f_{n}$ of $C_{n}$ on the left, we obtain
\begin{align*}
f_{n}(z)
&=\frac{1}{n}\sum_{k=1}^{n-1+j}\sum_{m=1}^{n}\sum_{j=1}^{k-(n-1)}\mathbb{P}(C_{m-1}=j)z^{j}
\mathbb{P}(C_{n-m}=k-(n-1)-j)z^{k-(n-1)-j}z^{n-1} \\
&=\frac{z^{n-1}}{n}\sum_{m=1}^{n}\sum_{j=1}^{k-(n-1)}\mathbb{P}(C_{m-1}=j)z^{j}
\sum_{k=1}^{n-1+j}\mathbb{P}(C_{n-m}=k-(n-1)-j)z^{k-(n-1)-j} \\
&=\frac{z^{n-1}}{n}\sum_{m=1}^{n}f_{m-1}(z)f_{n-m}(z), \tag{2.1} 
\end{align*} 
as required. 
\end{proof}

This of course will give a recursion for the variance, using the
well-known formula for variance in terms of the generating function $f_{X}(z)$:
\begin{eqnarray*}
\operatorname {Var}(X)=f_{X}^{\prime \prime}(1)+f_{X}^{\prime}(1)-\bigl(f_{X}^{\prime}(1) \bigr)^{2}.
\end{eqnarray*}
We use this formula together with Eq. (2.1). 
The first order derivative of $f_{n}(z)$ is
\begin{align*}
f_{n}^{\prime}(z)& = 
\frac{(n-1)z^{n-2}}{n}\sum_{j=1}^{n}f_{j-1}(z)f_{n-j}(z)
+\frac{z^{n-1}}{n}\sum_{j=1}^{n}f^{\prime}_{j-1}(z)f_{n-j}(z) \\
&\quad {}+\frac{z^{n-1}}{n}\sum_{j=1}^{n}f_{j-1}(z)f^{\prime}_{n-j}(z).
\end{align*}
From standard properties of generating functions, it holds that 
\begin{equation*}
f^{\prime}_{n}(1) = \mathbb {E}(C_{n}).
\end{equation*}
Differentiating again we obtain
\begin{align*}
f''_{n}(z)&=\frac{(n-1)(n-2)z^{n-3}}{n}\sum_{j=1}^{n}f_{j-1}(z)f_{n-j}(z)
+\frac{(n-1)z^{n-2}}{n}\sum_{j=1}^{n}f^{\prime}_{j-1}(z)f_{n-j}(z) \\
&\phantom{=}\,+\frac{(n-1)z^{n-2}}{n}\sum_{j=1}^{n}f_{j-1}(z)f^{\prime}_{n-j}(z)
+\frac{(n-1)z^{n-2}}{n}\sum_{j=1}^{n}f^{\prime}_{j-1}(z)f_{n-j}(z) \\
&\phantom{=}\,+\frac{z^{n-1}}{n}\sum_{j=1}^{n}f''_{j-1}(z)f_{n-j}(z)
+\frac{z^{n-1}}{n}\sum_{j=1}^{n}f^{\prime}_{j-1}(z)f^{\prime}_{n-j}(z) \\
&\phantom{=}\,+\frac{(n-1)z^{n-2}}{n}\sum_{j=1}^{n}f_{j-1}(z)f^{\prime}_{n-j}(z)
+\frac{z^{n-1}}{n}\sum_{j=1}^{n}f^{\prime}_{j-1}f^{\prime}_{n-j}(z) \\
&\phantom{=}\,+\frac{z^{n-1}}{n}\sum_{j=1}^{n}f_{j-1}(z)f''_{n-j}(z). 
\end{align*}

Setting $z=1$ \cite{nip},
\begin{align*}
f''_{n}(1)&=(n-1)(n-2)+\frac{2}{n}(n-1)\sum_{j=1}^{n}M_{j-1}
+\frac{2}{n}(n-1)\sum_{j=1}^{n}M_{n-j} \\
&\phantom{=}\,+\frac{1}{n}\sum_{j=1}^{n} \bigl(f''_{j-1}(1)+f''_{n-j}(1) \bigr) 
+\frac{2}{n}\sum_{j=1}^{n}M_{j-1}M_{n-j},
\end{align*}
where $M_{j-1}, M_{n-j}$ are $f^{\prime}_{j-1}(1)$, $f^{\prime}_{n-j}(1)$,
i.e. the mean number of comparisons to sort an array of $(j-1)$ and $(n-j)$ elements
respectively. Setting $B_{n}=\frac{f_{n}^{''}(1)}{2}$, we obtain
\begin{align*}
B_{n}= \dbinom{n-1}{2} +\dfrac{2(n-1)}{n}\sum_{j=1}^{n}M_{j-1}+\dfrac{2}{n}\sum_{j=1}^{n}B_{j-1}
+\dfrac{1}{n}\sum_{j=1}^{n}M_{j-1}M_{n-j},
\end{align*}
using the symmetry of the sums. What this argument has shown for us is the following 
-- compare \cite{sk}
where it is also shown that this recurrence has to be solved, 
though no details
of how to solve it are given.

\begin{Lemma}
In order to prove Theorem $2.3.1$, it is sufficient to show 
that this recurrence is given by
\begin{eqnarray*}
B_{n}=2(n+1)^{2}H_{n}^{2}-(8n+2)(n+1)H_{n}
+\frac{n(23n+17)}{2}-2(n+1)^{2}H_{n}^{(2)}.
\end{eqnarray*}
\end{Lemma}
\begin{proof} 
Combining the equations $B_{n} = \frac{f_{n}^{''}(1)}{2}$ and 
$\operatorname {Var} (X)=f_{X}^{\prime \prime}(1)+f_{X}^{\prime}(1)-\bigl(f_{X}^{\prime}(1) \bigr)^{2}$,
the result follows. 
\end{proof}

It will be convenient first to develop some theory on
various sums involving harmonic numbers. Often we used \textsc{Maple} 
to verify these
relations initially, but we provide complete proofs here.
As a brief reminder, $M_{j}$ denotes the expected number of 
comparisons needed to sort an array of $j$ keys. 
Recall that
\begin{equation*}
M_{j}=2(j+1)H_{j}-4j.
\end{equation*}
Thus,
\begin{align*}
\sum_{j=1}^{n}M_{j-1}&=\sum_{j=1}^{n}\bigl(2jH_{j-1}-4(j-1)\bigr)=
2\sum_{j=1}^{n}jH_{j-1}-4\sum_{j=1}^{n}(j-1). \tag{2.2}
\end{align*}
For the computation of the first sum of Eq. (2.2), a Lemma follows
\begin{Lemma}
For $n \in \mathbf N$
\begin{eqnarray*}
\sum_{j=1}^{n}jH_{j-1}=\frac{n(n+1)H_{n+1}}{2}-\frac{n(n+5)}{4}.
\end{eqnarray*}
\end{Lemma}
\begin{proof}
The sum can be written as
\begin{align*}
\sum_{j=1}^{n}jH_{j-1} & =2 + 3\biggl(1+\frac{1}{2} \biggr)+ \ldots + n\biggl 
(1 + \frac{1}{2} + \ldots +\frac{1}{n-1 } \biggr ) \\
& =H_{n-1}(1+2+ \ldots + n) - \biggl (1 + \frac{1+2}{2}+ \ldots + \frac{1+2+\ldots + n-1}{n-1} \biggr ) \\
& = \dfrac{n(n+1)}{2}H_{n-1} - \sum_{j=1}^{n-1}\biggl (\dfrac{\sum_{i=1}^{j} i}{j} \biggr ) \\
& = \dfrac{n(n+1)}{2}H_{n-1} - \dfrac{1}{2} \biggl (\dfrac{n(n+1)}{2} -1 \biggr ).
\end{align*}
The last equation can be easily seen to be
equivalent with the statement of the Lemma. 
\end{proof}

Thus we can find out about the sum of the $M_{j}$s, 
that it holds for $n \in\mathbf N$
\begin{Corollary}
\begin{eqnarray*}
\sum_{j=1}^{n}M_{j-1}=n(n+1)H_{n+1}-\frac{5n^{2}+n}{2}.
\end{eqnarray*}
\end{Corollary}
\begin{proof}
Using Lemma 2.3.4 and Eq. (2.2), the proof is immediate. 
\end{proof} 
Now, we will compute the term $\displaystyle \sum_{j=1}^{n}M_{j-1}M_{n-j}$. 
We shall use three Lemmas for the following proof.
\begin{Lemma}
For $n \in \mathbf N$, it holds that
\begin{align*}
\sum_{j=1}^{n}M_{j-1}M_{n-j}&= 4\sum_{j=1}^{n}jH_{j-1}(n-j+1)H_{n-j}\\
& \quad {}-\frac{8}{3}n(n^{2}-1)H_{n+1}+\frac{44n}{9}(n^{2}-1).
\end{align*}
\end{Lemma}
\begin{proof} 
To do this, we will again use the formula obtained previously for
$M_{j}$. We have
\begin{align*}
\sum_{j=1}^{n}M_{j-1}M_{n-j}&=\sum_{j=1}^{n}\biggl(\bigl(2jH_{j-1}-4j+4 \bigr)
\bigl(2(n-j+1)H_{n-j}-4n+4j \bigr)\biggr)\\
&=4\sum_{j=1}^{n}jH_{j-1}(n-j+1)H_{n-j}-8n\sum_{j=1}^{n}jH_{j-1}+8\sum_{j=1}^{n}j^{2}H_{j-1} \\
&\quad {} - 8\sum_{j=1}^{n}j(n-j+1)H_{n-j}+16n\sum_{j=1}^{n}j-16\sum_{j=1}^{n}j^{2} \\
&\quad{} +8\sum_{j=1}^{n}(n-j+1)H_{n-j} -16n^{2}+16\sum_{j=1}^{n}j.
\end{align*}

We need to work out the value of $\displaystyle \sum_{j=1}^{n}j^{2}H_{j-1}$:
\begin{Lemma}
For $n \in \mathbf N$ holds
\begin{equation*}
\sum_{j=1}^{n}j^{2}H_{j-1}=\frac{6n(n+1)(2n+1)H_{n+1}-n(n+1)(4n+23)}{36}.
\end{equation*}
\end{Lemma}
\begin{proof}
Using the same reasoning as in Lemma $2.3.4$, 
\begin{align*}
\sum_{j=1}^{n}j^{2}H_{j-1} & = 2^{2} + 3^{2}\biggl(1+\frac{1}{2} \biggr)+ \ldots + 
n^{2}\biggl (1 + \frac{1}{2} + \ldots +\frac{1}{n-1 } \biggr ) \\
& =H_{n-1}(1^{2}+2^{2}+ \ldots + n^{2}) \\
& \quad{} - \biggl (1 + \frac{1^{2}+2^{2}}{2}+ \ldots + 
\frac{1^{2}+2^{2}+\ldots + (n-1)^{2}}{n-1} \biggr ) \\
& = \dfrac{n(n+1)(2n+1)}{6}H_{n-1} - \sum_{j=1}^{n-1}\biggl 
(\dfrac{\sum_{i=1}^{j} i^{2}}{j} \biggr ) \\
& =\dfrac{n(n+1)(2n+1)}{6}H_{n-1} - \dfrac{1}{36} (4n^{3}+3n^{2}-n-6 ),
\end{align*} 
completing the proof.
\end{proof}

We also need to compute $\displaystyle \sum_{j=1}^{n}j(n-j+1)H_{n-j}$. 
A Lemma follows
\begin{Lemma}
For $n \in \mathbf N$
\begin{eqnarray*}
\sum_{j=1}^{n}j(n-j+1)H_{n-j}=\frac{6nH_{n+1}(n^{2}+3n+2)-5n^{3}-27n^{2}-22n}{36}.
\end{eqnarray*}
\end{Lemma}
\begin{proof} 
We can write $j=n+1-(n-j+1)$. Then, substituting $k=n-j+1$
\begin{align*}
\sum_{j=1}^{n}j(n-j+1)H_{n-j}&=\sum_{j=1}^{n}\bigl((n+1)-(n-j+1)\bigr)(n-j+1)H_{n-j} \\
&=(n+1)\sum_{j=1}^{n}(n-j+1)H_{n-j}-\sum_{j=1}^{n}(n-j+1)^{2}H_{n-j} \\
&=(n+1)\sum_{k=1}^{n}kH_{k-1}-\sum_{k=1}^{n}k^{2}H_{k-1}.
\end{align*}
These sums can be computed using Lemmas $2.3.4$ and $2.3.7$. 
\end{proof}

In the same manner we shall compute $\displaystyle \sum_{j=1}^{n}(n-j+1)H_{n-j}.$ 
Changing variables, the expression becomes $\displaystyle \sum_{k=1}^{n}kH_{k-1}$. 
Using the previous results, we have
\begin{align*}
\sum_{j=1}^{n}M_{j-1}M_{n-j}&=4\sum_{j=1}^{n}jH_{j-1}(n-j+1)H_{n-j} 
-8n\biggl(\frac{n(n+1)H_{n+1}}{2}-\frac{n(n+5)}{4} \biggr) \\
&\quad {}+8\biggl(\frac{6n(n+1)(2n+1)H_{n+1}-n(n+1)(4n+23)}{36} \biggr)+16n\sum_{j=1}^{n}j \\
&\quad {}-8\biggl(\frac{6nH_{n+1}(n^{2}+3n+2)-5n^{3}-27n^{2}-22n}{36}\biggr)-16\sum_{j=1}^{n}j^{2} \\
&\quad {} +8\biggl(\frac{n(n+1)H_{n+1}}{2}-\frac{n(n+5)}{4}\biggr)-16n^{2}+16\sum_{j=1}^{n}j \\
&=4\sum_{j=1}^{n}jH_{j-1}(n-j+1)H_{n-j}-4n(n+1)(n-1)H_{n+1} \\
&\quad {} +\frac{4}{3}n(n^{2}-1)H_{n+1} +\frac{176n^{3}}{36}-\frac{176n}{36}, 
\end{align*}
finishing the proof of Lemma $2.3.6$. 
\end{proof}
After some tedious calculations, the recurrence
relation becomes
\begin{align*}
B_{n} &=\frac{4\displaystyle \sum_{j=1}^{n} jH_{j-1}(n-j+1)H_{n-j}}{n}+\dfrac{2}{n}
\sum_{j=1}^{n}B_{j-1}+\dfrac{-9n^{2}+5n+4}{2} \\  
&\quad {} -\dfrac{2}{3}(n^{2}-1)H_{n+1}+\dfrac{44}{9}(n^{2}-1).
\end{align*}

Subtracting $nB_{n}$ from $(n+1)B_{n+1}$,
\begin{align*}
&(n+1)B_{n+1}-nB_{n} \\
&=4\biggl(\sum_{j=1}^{n+1}jH_{j-1}(n-j+2)H_{n+1-j}-\sum_{j=1}^{n}jH_{j-1}(n-j+1)H_{n-j} \biggr) \\
&\phantom{=}\, +2B_{n}+(n+1)\frac{-9(n+1)^{2}+5(n+1)+4}{2}-n\frac{-9n^{2}+5n+4}{2 } \\
&\phantom{=}\, -(n+1)\frac{2}{3}\bigl((n+1)^{2}-1 \bigr)H_{n+2}+\frac{2}{3}n(n^{2}-1)H_{n+1} \\
&\phantom{=}\,+(n+1)\frac{44}{9}\bigl((n+1)^{2}-1 \bigr)-n\frac{44}{9}(n^{2}-1) \\
&=4\biggl(\sum_{j=1}^{n}jH_{j-1}(n-j+2)H_{n+1-j}-\sum_{j=1}^{n}jH_{j-1}(n-j+1)H_{n-j} \biggr) \\
&\phantom{=}\, +2B_{n}-\frac{27n^{2}+17n}{2}-\frac{2}{3}n(n+1)(n+2)H_{n+2} \\
& \qquad +\frac{2}{3}nH_{n+1}(n^{2}-1)+\frac{44n(n+1)}{3}.
\end{align*}
We obtain
\begin{align*}
(n+1)B_{n+1}-nB_{n}& =4\biggl(\sum_{j=1}^{n}jH_{j-1}+\sum_{j=1}^{n}jH_{j-1}H_{n-j+1}\biggr) \\
&\phantom{=}\, +2B_{n}-2n(n+1)H_{n+1} 
+\frac{1}{2}n(n+11).
\end{align*}
We have to work out the following sum
\begin{eqnarray*}
\sum_{j=1}^{n}jH_{j-1}H_{n+1-j}.
\end{eqnarray*}
We note that
\begin{align*}
\sum_{j=1}^{n}jH_{j-1}H_{n+1-j} & = 2H_{1}H_{n-1} + 3H_{2}H_{n-2} + \ldots + 
(n-1)H_{n-2}H_{2} + nH_{n-1}H_{1} \\
& =\frac{n+2}{2}\sum_{j=1}^{n}H_{j}H_{n-j}. \, \, \tag{2.3}
\end{align*} 

Sedgewick \cite{sedgewick}, presents and proves the following result:
\begin{Lemma}
\begin{eqnarray*}
\sum_{i=1}^{n}H_{i}H_{n+1-i}= (n+2)(H^{2}_{n+1}-H^{(2)}_{n+1})-2(n+1)(H_{n+1}-1).
\end{eqnarray*}
\end{Lemma}
\begin{proof}
By the definition of harmonic numbers, we have
\begin{align*}
\sum_{i=1}^{n}H_{i}H_{n+1-i}
=\sum_{i=1}^{n}H_{i}\sum_{j=1}^{n+1-i}\frac{1}{j}
\end{align*}
and the above equation becomes
\begin{align*}
\sum_{j=1}^{n}\frac{1}{j}\sum_{i=1}^{n+1-j}H_{i}
=\sum_{j=1}^{n}\frac{1}{j}\bigl((n+2-j)H_{n+1-j}-(n+1-j) \bigr), \, \, \tag{2.4}
\end{align*}
using the identity \cite{knuth 1},
\begin{equation*}
\sum_{j=1}^{n}H_{j} = (n+1)H_{n} - n.
\end{equation*}

Eq. (2.4) can be written as
\begin{align*}
&(n+2)\sum_{j=1}^{n}\frac{H_{n+1-j}}{j}-\sum_{j=1}^{n}H_{n+1-j}-(n+1)H_{n}+n \\
& =(n+2)\sum_{j=1}^{n}\frac{H_{n+1-j}}{j}-\bigl((n+1)H_{n}-n \bigr)-(n+1)H_{n}+n \\
& =(n+2)\sum_{j=1}^{n}\frac{H_{n+1-j}}{j}-2(n+1)(H_{n+1}-1).
\end{align*}
It can be easily verified that
\begin{align*}
\sum_{j=1}^{n}\frac{H_{n+1-j}}{j} & = \sum_{j=1}^{n}\frac{H_{n-j}}{j} + 
\sum_{j=1}^{n}\dfrac{1}{j(n+1-j)} \\
& = \sum_{j=1}^{n-1}\frac{H_{n-j}}{j} + 2\dfrac{H_{n}}{n+1}. \, \, \tag{2.5}
\end{align*}
Making repeated use of Eq. (2.5), we obtain the identity
\begin{eqnarray*}
\sum_{j=1}^{n}\frac{H_{n+1-j}}{j}=2\sum_{k=1}^{n}\frac{H_{k}}{k+1}.
\end{eqnarray*}
We have then
\begin{align*}
2\sum_{k=1}^{n}\frac{H_{k}}{k+1}& =2\sum_{k=2}^{n+1}\frac{H_{k-1}}{k}
=2\sum_{k=1}^{n+1}\frac{H_{k-1}}{k}=2\sum_{k=1}^{n+1}\frac{H_{k}}{k}
-2\sum_{k=1}^{n+1}\frac{1}{k^{2}} \\
& =2\sum_{k=1}^{n+1}\sum_{j=1}^{k}\frac{1}{jk}-2H^{(2)}_{n+1}
=2\sum_{j=1}^{n+1}\sum_{k=j}^{n+1}\frac{1}{jk}-2H^{(2)}_{n+1} \\
& =2\sum_{k=1}^{n+1}\sum_{j=k}^{n+1}\frac{1}{kj}-2H^{(2)}_{n+1}.
\end{align*}

The order of summation was interchanged. We can sum on all $j$ and for $k=j$, 
we must count this term twice.
We obtain
\begin{align*}
2\sum_{k=1}^{n}\frac{H_{k}}{k+1}=\sum_{k=1}^{n+1}\biggl
(\sum_{j=1}^{n+1}\frac{1}{kj}+\frac{1}{k^{2}} \biggr)-2H^{(2)}_{n+1}
=H^{2}_{n+1}-H^{(2)}_{n+1}. \, \, \tag{2.6}
\end{align*}
Finally
\begin{equation*}
\sum_{i=1}^{n}H_{i}H_{n+1-i}=
(n+2)(H^{2}_{n+1}-H^{(2)}_{n+1})-2(n+1)(H_{n+1}-1). \qedhere
\end{equation*}
\end{proof}

The following Corollary is a direct consequence of Eq. (2.5) and (2.6).
\begin{Corollary} 
For $n \in \mathbf N$, it holds
\begin{align*}
H_{n+1}^{2}-H_{n+1}^{(2)}=2\sum_{j=1}^{n}\frac{H_{j}}{j+1}.
\end{align*}
\end{Corollary}
\begin{proof}
\begin{align*}
& H^{2}_{n+1}-H^{(2)}_{n+1} =H^{2}_{n}-H^{(2)}_{n}+2\frac{H_{n}}{n+1}  \mbox{~~~and~by~iteration,} \\
& H^{2}_{n+1}-H^{(2)}_{n+1}=2\sum_{j=1}^{n}\frac{H_{j}}{j+1}. \qedhere
\end{align*}
\end{proof}

We will use the above Lemma and Corollary in our analysis. We have that
\begin{eqnarray*}
\sum_{i=1}^{n}H_{i}H_{n+1-i}=\sum_{i=1}^{n}\Biggl(H_{i}\biggl(H_{n-i}+
\frac{1}{n+1-i} \biggr) \Biggr)=\sum_{i=1}^{n}H_{i}H_{n-i}+\sum_{i=1}^{n}\frac{H_{i}}{n+1-i}.
\end{eqnarray*}
The second sum substituting $j=n+1-i$ becomes
\begin{eqnarray*}
\sum_{i=1}^{n}\frac{H_{i}}{n+1-i}=\sum_{j=1}^{n}\frac{H_{n+1-j}}{j}.
\end{eqnarray*}
As we have seen it is equal to
\begin{eqnarray*}
\sum_{j=1}^{n}\frac{H_{n+1-j}}{j}=2\sum_{j=1}^{n}\frac{H_{j}}{j+1}
=H^{2}_{n+1}-H^{(2)}_{n+1}.
\end{eqnarray*}

Hence, by Lemma 2.3.9,
\begin{align*}
\sum_{j=1}^{n}H_{j}H_{n-j}&=(n+2)(H^{2}_{n+1}-H^{(2)}_{n+1})-2(n+1)
(H_{n+1}-1)-(H^{2}_{n+1}-H^{(2)}_{n+1}) \\
&=(n+1)\bigl((H^{2}_{n+1}-H^{(2)}_{n+1})-2(H_{n+1}-1) \bigr).
\end{align*}
Using the above result and Eq. (2.3), we have
\begin{eqnarray*}
\sum_{j=1}^{n}jH_{j-1}H_{n+1-j}=\dbinom{n+2}{2} \biggl
((H^{2}_{n}-H^{(2)}_{n})+\dfrac{2n(1-H_{n})}{n+1} \biggr).
\end{eqnarray*}
Having worked out all the expressions involved in the following relation
\begin{align*}
(n+1)B_{n+1}-nB_{n}& =4\biggl(\sum_{j=1}^{n}jH_{j-1}+\sum_{j=1}^{n}jH_{j-1}H_{n-j+1} \biggr) \\
& \quad +2B_{n}-2n(n+1)H_{n+1}+\frac{1}{2}n(n+11).
\end{align*}
This becomes
\begin{align*}
&(n+1)B_{n+1}-nB_{n} \\
&=4 \Biggl(\frac{n(n+1)H_{n+1}}{2}-\frac{n(n+5)}{4}+\dbinom{n+2}{2}\Bigl(
(H^{2}_{n+1}-H^{(2)}_{n+1})-2(H_{n+1}-1) \Bigr) \Biggr) \\
&\qquad {} +2B_{n}-2n(n+1)H_{n+1}+\frac{1}{2}n(n+11) \\
&=2(n+1)(n+2)(H^{2}_{n+1}-H^{(2)}_{n+1})-4(n+1)(n+2)(H_{n+1}-1) 
-\frac{n(n-1)}{2}+2B_{n}.
\end{align*}

Dividing both sides by $(n+1)(n+2)$ and unwinding the recurrence,
\begin{eqnarray*}
\frac{B_{n}}{n+1}=\frac{B_{0}}{1}+2\sum_{i=1}^{n}(H^{2}_{i}-H^{(2)}_{i})-
4\sum_{i=1}^{n}(H_{i}-1)-\sum_{i=1}^{n}\frac{(i-1)(i-2)}{2i(i+1)}.
\end{eqnarray*}
Hence 
\begin{align*}
\frac{B_{n}}{n+1}&=2(n+1)(H^{2}_{n}-H^{(2)}_{n})+4n-4nH_{n}-4\Bigl((n+1)H_{n}-2n \Bigr) \\
& \quad {} - \sum_{i=1}^{n} \biggl(\frac{i+2}{2i}-\frac{3}{i+1} \biggr) \\
&=2(n+1)(H^{2}_{n}-H^{(2)}_{n})-H_{n}(8n+2)+\frac{23n}{2}-3+\frac{3}{n+1} \\
&=2(n+1)(H^{2}_{n}-H^{(2)}_{n})-H_{n}(8n+2)+\frac{23n^{2}+17n}{2(n+1)}.
\end{align*}

Finally, multiplying by $(n+1)$ we obtain
\begin{align*}
B_{n}=2(n+1)^{2}(H^{2}_{n}-H^{(2)}_{n})-H_{n}(n+1)(8n+2)+\frac{23n^{2}+17n}{2}.
\end{align*}
Consequently, by Lemma 2.3.3 the variance of the number of 
comparisons of randomised Quicksort is
\begin{eqnarray*}
7n^{2}-4(n+1)^{2}H_{n}^{(2)}-2(n+1)H_{n}+13n,
\end{eqnarray*}
completing the proof of Theorem 2.3.1. 

\section{``Divide and Conquer'' recurrences}

We have computed the mean and variance of the number of comparisons made by
Quicksort that mainly contribute to its time complexity. Because of the
simple structure of the algorithm (dividing into smaller subproblems) we
can in fact approach many other related problems in the same spirit. Let
$F(n)$ denote the expected value of some random variable associated with
randomised Quicksort and $T(n)$ be the average value of the ``toll function'',
which is the needed cost 
to divide the problem into two simpler subproblems. 
Then $F(n)$ is equal to the contribution $T(n)$,
plus the measures required 
sort the resulting subarrays of $(i-1)$ and $(n-i)$ elements, where the pivot
$i$ can be any key of the array with equal probability. 

Thus, the recurrence relation is
\begin{align*}
F(n) & = T(n) + \frac{1}{n} \sum_{i=1}^{n} (F(i-1) + F(n-i) ) \\
& = T(n) + \frac{2}{n} \sum_{i=1}^{n} F(i-1).
\end{align*}
This is the general type of recurrences arising in the analysis of Quicksort, 
which can be manipulated using the difference method or by generating functions. 
Since an empty array or an one having a unique key is trivially solved,
the initial values of the recurrence is
\begin{equation*}
F(0) = F(1) = 0.
\end{equation*}
The first method leads to the elimination of the sum, 
by subtracting $(n-1)F(n-1)$ from $nF(n)$ -- see \cite{knuth 3}. 
The recurrence becomes
\begin{align*}
nF(n) - (n-1)F(n-1) =nT(n) - (n-1)T(n-1) + 2F(n-1)
\end{align*}
and dividing by $n(n+1)$ we have
\begin{align*}
\dfrac{F(n)}{n+1} = \dfrac{nT(n) - (n-1)T(n-1)}{n(n+1)} + \dfrac{F(n-1)}{n}.
\end{align*}

This recurrence can be immediately solved by ``unfolding'' its terms. 
The general solution is
\begin{align*}
F(n) & = (n+1) \left (\sum_{j=3}^{n}\dfrac{jT(j) - (j-1)T(j-1)}{j(j+1)} + \dfrac{F(2)}{3} \right) \\
&= (n+1) \left (\sum_{j=3}^{n}\dfrac{jT(j) - (j-1)T(j-1)}{j(j+1)} + \dfrac{T(2)}{3} \right).
\end{align*}
When the sorting of subarrays of $m$ keys or less is done by insertion sort, 
the solution of the recurrence is
\begin{align*}
F(n) & = (n+1) \biggl ( \sum_{j=m+2}^{n}\dfrac{jT(j) - (j-1)T(j-1)}{j(j+1)} + 
\dfrac{F(m+1)}{m+2} \biggr ) \\
& =(n+1) \biggl ( \sum_{j=m+2}^{n}\dfrac{jT(j) - (j-1)T(j-1)}{j(j+1)} + 
\dfrac{T(m+1)}{m+2} \biggr ),
\end{align*}
since $n-1 > m$.

Another classic approach, which is more transparent and elegant, is the application of generating functions. 
The recurrence is transformed to a differential equation, which is then solved. 
The function is written in terms of series and the extracted coefficient is the solution.
Multiplying by $nx^{n}$ and then summing with respect to $n$, in order 
to obtain the generating function $G(x) = \displaystyle \sum_{n=0}^{\infty}F(n)x^{n}$, we have
\begin{eqnarray*}
\sum_{n=0}^{\infty}nF(n)x^{n} = \sum_{n=0}^{\infty}nT(n)x^{n} + 2\sum_{n=0}^{\infty}\sum_{i=1}^{n}F(i-1)x^{n}.
\end{eqnarray*} 
The double sum is equal to
\begin{align*}
\sum_{n=0}^{\infty}\sum_{i=1}^{n}F(i-1)x^{n} & = G(x)\sum_{n=1}^{\infty}x^{n} = \dfrac{xG(x)}{1-x}
\end{align*}
and the differential equation is
\begin{eqnarray*}
xG^{\prime}(x) = \sum_{n=0}^{\infty}nT(n)x^{n}+ \dfrac{2xG(x)}{1-x}.
\end{eqnarray*}
Cancelling out $x$ and multiplying by $(1-x)^{2}$, 
\begin{eqnarray*}
&& G^{\prime}(x)(1-x)^{2} =(1-x)^{2} \sum_{n=1}^{\infty}nT(n)x^{n-1}+ 2(1-x)G(x) \\
&& \left( G(x)(1-x)^{2} \right )^{\prime} = (1-x)^{2} \sum_{n=1}^{\infty}nT(n)x^{n-1} \\
&& \Longrightarrow  G(x)(1-x)^{2} = \int (1-x)^{2} \sum_{n=1}^{\infty}nT(n)x^{n-1}\, \mathrm{d}x + C \\
&& \Longrightarrow G(x) = \dfrac{ \displaystyle \int (1-x)^{2} \displaystyle \sum_{n=1}^{\infty}nT(n)x^{n-1}\,
\mathrm{d}x + C}{(1-x)^{2}},
\end{eqnarray*}
where $C$ is constant, which can be found using the initial condition $G(0)= 0$. 
The solution then is being written as power series and the coefficient sequence found 
is the expected sought cost.

Now, one can obtain any expected cost of the algorithm, just by using these results. 
The ``toll function'' will be different for each case. 
Plugging in the average value, the finding becomes a matter of simple operations. 
This type of analysis unifies the recurrences of Quicksort into a single one and 
provides an intuitive insight of the algorithm.

\section{Higher moments}

We have effectively calculated the first and second moments of $C_{n}$ in
Quicksort. Existing literature does not seem to address much questions about
skewness and kurtosis, which are sometimes held to be interesting features
of a random variable. Here, we present an 
inconclusive discussion about the sign of the skewness.

Using the probability generating function, we can
obtain higher moments of the
algorithm's complexity. A Lemma follows
\begin{Lemma}
Let a random variable $X$ having probability generating function: 
$$f_{X}(z)=\sum_{n=0}^{\infty}\mathbb{P}(X=n)z^{n}.$$
For the $k$-th order derivative it holds that
\begin{align*}
\frac{d^{k}f_{X}(z)}{dz^{k}}\bigg|_{z=1}
= \mathbb{E} \bigl( X \cdot (X-1) \cdot \ldots \cdot(X-k+1) \bigr).
\end{align*}
\end{Lemma}
\begin{proof}
Simply by computing successively the $k$-th order derivative of $f_{X}(z)$, we obtain
\begin{align*}
\frac{d^{k}f_{X}(z)}{dz^{k}}=\sum_{n=0}^{\infty}n \cdot (n-1)\cdot \ldots 
\cdot(n-k+1)\mathbb{P}(X=n)\cdot z^{n-k}
\end{align*}
Setting $z=1$, the proof follows directly. Note that the argument is 
similar to continuous random variables. 
\end{proof}

Using \textsc{Maple}, we obtained a recursive form for the general $k$-th 
order derivative of the generating function. 
\begin{Lemma}
Let $$f_{n}(z)=\frac{z^{n-1}}{n}\sum_{j=1}^{n}f_{j-1}(z)f_{n-j}(z)$$
be the generating function of Quicksort's complexity in sorting $n$ keys. 
The $k \in \mathbf N$ order derivative
is given by
\begin{align*}
\frac{\mathrm{d}^{k}f_{n}(z)}{\mathrm{d}z^{k}}=
\frac{1}{n} \cdot \left (\sum_{i=0}^{k} 
\dbinom{k}{i} \cdot \frac{\Gamma(n)}{\Gamma(n-i)} 
\cdot z^{n-i-1}\frac{\mathrm{d}^{k-i}}{\mathrm{d}z^{k-i}}\biggl 
(\sum_{j=1}^{n}f_{j-1}(z)\cdot f_{n-j}(z)\biggr)\right),
\end{align*}
where the $\Gamma$ function is defined for complex 
variable $z \neq 0, -1, -2, \ldots$ as
\begin{equation*}
\Gamma(z):=\int_{0}^{\infty}p^{z-1}e^{-p}\, \mathrm{d}p
\end{equation*}
and when $z$ is a positive integer, then $\Gamma(z)=(z-1)!$.
\end{Lemma}
\begin{proof}
For $k=0$, the result follows trivially. 
Assume that the statement of the Lemma holds for $k=m$. 
The $(m+1)$-th order derivative is
\begin{align*}
\frac{\mathrm{d}^{m+1}f_{n}(z)}{\mathrm{d}z^{m+1}}& = \frac{1}{n} \sum_{i=0}^{m} 
\dbinom{m}{i} \frac{\Gamma(n)}{\Gamma(n-i)} 
\Biggl ((n-i-1)z^{n-i-2} \frac{\mathrm{d}^{m-i}}{\mathrm{d}z^{m-i}} 
\biggl (\sum_{j=1}^{n}f_{j-1}(z)\cdot f_{n-j}(z)\biggr) \\
& \quad {} + z^{n-i-1}\frac{\mathrm{d}^{m-i+1}}{\mathrm{d}z^{m-i+1}}\biggl 
(\sum_{j=1}^{n}f_{j-1}(z)\cdot f_{n-j}(z)\biggr)\Biggr) \\
& \kern -1em = \frac{1}{n} \sum_{i=1}^{m+1} 
\dbinom{m}{i-1} \frac{\Gamma(n)}{\Gamma(n-i+1)} 
(n-i)z^{n-i-1} \frac{\mathrm{d}^{m-i+1}}{\mathrm{d}z^{m-i+1}} 
\biggl (\sum_{j=1}^{n}f_{j-1}(z)\cdot f_{n-j}(z)\biggr) \\
& \quad {} + \frac{1}{n} \sum_{i=0}^{m} \dbinom{m}{i} \frac{\Gamma(n)}{\Gamma(n-i)} 
z^{n-i-1}\frac{\mathrm{d}^{m-i+1}}{\mathrm{d}z^{m-i+1}}\biggl 
(\sum_{j=1}^{n}f_{j-1}(z)\cdot f_{n-j}(z)\biggr).
\end{align*}

Note that
\begin{equation*}
\frac{\Gamma(n)}{\Gamma(n-i+1)} (n-i)= \frac{\Gamma(n)}{\Gamma(n-i)}.
\end{equation*}
Therefore,
\begin{align*}
\frac{\mathrm{d}^{m+1}f_{n}(z)}{\mathrm{d}z^{m+1}}& = \frac{1}{n} 
\sum_{i=1}^{m} \Biggl(
\dbinom{m}{i-1}+\dbinom{m}{i} \Biggr) \frac{\Gamma(n)}{\Gamma(n-i)} 
z^{n-i-1} \\
& \qquad{} \times \frac{\mathrm{d}^{m-i+1}}{\mathrm{d}z^{m-i+1}}\biggl 
(\sum_{j=1}^{n}f_{j-1}(z)\cdot f_{n-j}(z)\biggr).
\end{align*}
The well-known identity, (e.g. see \cite{con})
\begin{equation*}
\dbinom{m}{i-1}+\dbinom{m}{i}=\dbinom{m+1}{i},
\end{equation*}
completes the proof.
\end{proof}

We should point out that Lemma $2.5.2$ is an immediate consequence of Leibniz's product rule.
Next, we shall ask a Question about the sign of the skewness 
of the time complexity of the algorithm,
as it is moderately difficult to solve the
recurrence involved, in order to compute the third moment. We
already have seen that the possibility of worst-case performance of the
algorithm is rather small and in the majority of cases the running time is
close to the average time complexity which is $O\bigl(n \log_{2}(n)\bigr)$.
Intuitively, this suggests that the complexity is negatively skewed. 
We present the following Question:
\begin{question}
Is the skewness $\mathbb{S}(C_{n})$ of the random number of key 
comparisons of Quicksort for the sorting of $n \geq 3$ keys negative?
\end{question} 

Note that the cases $n=1, 2$ are deterministic, since we always make $0$ and $1$ comparisons
for the sorting. This Question may have an affirmative answer, which can be possibly proven
by an induction argument on the number of keys. However, great deal of
attention must be exercised to the fact that the random number of comparisons 
required to sort the segment of keys 
less than the pivot and the segment of 
keys that are greater than the pivot are conditionally independent,
subject to the choice of pivot.

\section{Asymptotic analysis}

After having examined the number of comparisons of Quicksort, in terms of average and
worst case scenarios, and its variance, it is desirable also to study the
concentration of the random variable $C_{n}$ about its mean. One might
hope, by analogy with other probabilistic situations, that for large values
of $n$ the number of comparisons is highly likely to be very close to the
mean.

The analysis will be confined to the number of comparisons, because
this is the most relevant measure for this thesis. 
Since our results will be asymptotic in nature, we need to have some relevant
ideas about convergence of sequences of random variables. The following definitions
come from \cite{billingsley}, \cite{feller}.
\begin{Definition} {\ \\}
(i) A sequence of random variables $\{X_{1}, X_{2}, \ldots\}$ is said to
converge in distribution (or converge weakly) to a random variable $X$ if and only if
$$\lim_{n \to \infty} F_{n}(x)=F(x)$$
at every point $x$ where $F$ is continuous. Here $F_{n}(x)$ and $F(x)$ are
respectively the cumulative distribution functions of random variables
$X_{n}$ and $X$.
We shall denote this type of convergence by
$\displaystyle X_{n}\stackrel{\mathcal{D}}{\longrightarrow} X$.
\\
(ii) A sequence of random variables $\{X_{1}, X_{2}, \ldots\}$ is said to
converge in probability to a random variable $X$ if and only if $\forall \varepsilon > 0$
holds $$\lim_{n \to \infty}\mathbb{P}(\vert X_{n}-X \vert\geq\varepsilon)=0.$$
We will denote this type of convergence by
$\displaystyle X_{n}\stackrel{\mathcal{P}}{\longrightarrow} X$. 
\\
(iii) A sequence $\{X_{1}, X_{2}, \ldots\}$ of random variables is said to converge in
$L^{p}$-norm to a random variable $X$ if and only if
$$\lim_{n \to \infty} \mathbb{E}(\vert X_{n}-X\vert^{p}) = 0.$$
\end{Definition}
Note that convergence in $L^{p}$-norm, for $p\geq 1$, implies convergence
in probability, and it is also easy to see that convergence in probability
implies convergence in distribution: see e.g. \cite{billingsley}. 
Both converse statements are false in general.

We also present the definition
of martingale, which shall be employed in a later stage of our analysis.
\begin{Definition}
Let $\{Z_{1},Z_{2}, \ldots\}$ be a sequence of random
variables. We say $Z_{n}$ is a martingale 
if and only if \newline
(i) $\mathbb{E}(\vert Z_{n}\vert)<\infty$, $\forall  n \in \mathbf N$.
\\
(ii) $\mathbb{E}(Z_{n}\vert Z_{n-1},Z_{n-2},\ldots ,Z_{1})=Z_{n-1}$.
\end{Definition}

\subsection{Binary trees}

In this subsection, a central notion to the analysis of Quicksort, 
and in general to the analysis of algorithms is discussed. 
We begin with a definition.
\begin{Definition}
A graph is defined as an ordered pair of two sets $(V, E)$. The set $V$
corresponds to the set of vertices or nodes. The set $E$ is the set of edges,
which are pairs of distinct vertices.
\end{Definition}
One kind of graph we concentrate on are trees. A definition
of a tree (which suits us) is as follows \cite{knuth 1}:
\begin{Definition}
Tree is a finite set $\Delta$ of nodes, such that: \\
(i) There is a unique node called the root of the tree. 
\\
(ii) The remaining nodes are partitioned into $k \in \mathbf N$ disjoint sets 
$\Delta_{1}, \Delta_{2}, \ldots, \Delta_{k}$ and each of these sets is a tree. 
Those trees are called the subtrees of the root.
\end{Definition}

A particularly common tree is a binary tree. It is defined as a tree with 
the property that every node
has at most $2$ subtrees, i.e. each node -- excluding the root -- is adjacent to
one parent, so to speak, and up to two offspring, namely left and right child
nodes.
Note here that nodes which do not have any child nodes are called external. 
Otherwise, they are called internal. The size of a binary tree is the number
of its nodes. The depth of a node is simply the number of edges from that
node to the root in the (unique) shortest path between them and
the height of a binary tree is the length from the deepest node to the root. 

An extended binary tree is a binary tree with the property that every internal 
node has exactly two offspring \cite{knuth 1}. Let $D_{j}$ be the
depth of insertion of a key in a random binary tree of size $j-1$. Since the root node
is first inserted to an empty tree, it holds that $D_{1}:=0$. The next inserted key
is compared with the key at the root and if it is smaller, is placed to the left;
otherwise is attached as a right subtree, thus $D_{2}:=1$. 

Thus, we define the internal path length of
an extended binary tree having $n$ internal nodes, to be equal to the sum over all internal
nodes, of their distances to the root. Let us denote this quantity by $X_{n}$.
We then have that
\begin{equation*}
X_{n}=\sum_{j=1}^{n}D_{j}.
\end{equation*}
Similarly, the external path length is defined as the total number of edges in
all the shortest paths from external nodes to the root node.
The next Lemma gives a relationship between those two quantities. 
\begin{Lemma}[Knuth \cite{knuth 1}]
For an extended binary tree with $n$ internal nodes it holds that
\begin{equation*}
Y_{n}=X_{n}+2n,
\end{equation*}
where $Y_{n}$ denotes the external path length of the tree.
\end{Lemma} 
\begin{proof}
Suppose that we remove the two offspring of an internal node $v$, with its offspring
being external nodes of the tree. We suppose that $v$ is at distance
$h$ from the root. Then the external path length is reduced by $2(h+1)$, as
its two offspring are removed. At the same time, the external path length is increased by $h$,
because the vertex
$v$ has just become an external node. Thus, the net change is
equal to $-2(h+1)+h=-(h+2)$. For internal path length the change is a reduction
by $-h$ as $v$ is no
longer internal. Thus, overall, the change in $Y_{n}-X_{n}$ is equal to
$-(h+2)-(-h)=-2$. The Lemma follows by induction.
\end{proof}

Binary trees play a fundamental and crucial
role in computing and in the analysis of algorithms. They are widely used as
data structures, because fast insertion, deletion and searching for a given
record can be achieved. 
In Quicksort, letting nodes represent keys, the algorithm's operation can be
depicted as a binary tree. The root node stores the initial pivot element. 
Since the algorithm at each recursion splits the initial array into 
two subarrays and so on, we have an ordered binary tree. 
The left child of the root stores the pivot chosen to sort all keys 
less than the value of root and the right child node stores the pivot for 
sorting the elements greater than the root. 

The process continues 
until the algorithm divides the array into trivial subarrays of order 
$0$ or $1$, which do not need any more sorting. 
These elements are stored as external or leaf nodes in this binary structure. 
It easily follows that for any given node storing a key $k$, its
left subtree store keys less than $k$ and similarly its right
subtree contain keys greater than $k$. 

In variants of Quicksort, where many pivots are utilised to partitioning
process, the generalisation 
of binary trees provides a framework for the analysis,
though we do not develop this in detail here. In
the next subsection, the limiting distribution of the number of
comparisons will be analysed, in terms of trees.

\subsection{Limiting behaviour} 

We have previously seen that the
operation of the algorithm can effectively be represented as a binary tree.
Internal nodes store pivots selected at each recursion step of the algorithm:
so the root vertex, for example, stores the first pivot with which
all other elements are compared.
We are interested in the total number of comparisons made. To understand
this, we note that every vertex is compared with the first pivot
which is at level $0$ of the tree. The array is divided into two parts --
those above the pivot, and those below it. (Either of these subarrays
may be empty). If a subarray is not empty, a pivot is found in it and is attached
to the root as a child on the left, for the elements
smaller than the first pivot or on the right for the elements larger than the
pivot. 

The process then continues recursively and each element at level $k$
in the tree is compared with each of the $k$ elements above it. 
Thus the total number of comparisons
made is the sum of the depths of all nodes in the tree. This is
equivalent to the internal path length of the extended binary tree.
Thus,
\begin{equation*}
Y_{n}=C_{n}+2n,
\end{equation*}
where $Y_{n}$ is the external path length of the tree. To see this, we
simply use Lemma 2.6.5. This fact can be also found in \cite{hosam}.

Generally, assume that we have to sort an array of $n$ distinct items with 
pivots uniformly chosen. 
All $n!$ orderings of keys are equally likely. This is equivalent to carrying
out $n$ successive insertions 
\cite{knuth 3}. Initially, the root node is inserted. The second key to be
inserted is compared with the key at the root. If it is less than that key, 
it is attached as its left subtree. Otherwise, it is inserted 
as the right subtree. This process continues recursively 
by a series of comparisons of keys, until all $n$ keys have been inserted. 
Traversing the binary search tree in order, i.e. visiting the 
nodes of the left subtree, the root and the 
nodes of the right subtree, keys are printed 
in ascending order. Thus, Quicksort can be 
explicitly analysed in this way. 

Recall that each internal node 
corresponds to the 
pivot at a given recursion of the Quicksort process (e.g. the depth of a
given node). Thus the first pivot corresponds to the root 
node, its descendants or child nodes are 
pivots chosen from the two subarrays to be sorted, etc. Eventually,
after $n$ insertions, we have built a binary search tree from top to bottom.  
We can use this approach to understand the following Theorem of R\'{e}gnier \cite{reg}. 
\begin{theorem}[R\'{e}gnier \cite{reg}]
Let random variables $Y_{n}$ and $X_{n}$ denote respectively the external
and internal path length of binary search tree, built by $n$ successive
insertions of keys. Then the random variables
\begin{eqnarray*}
Z_{n}=\frac{Y_{n}-2(n+1)(H_{n+1}-1)}{n+1}=\frac{X_{n}-2(n+1)H_{n}+4n}{n+1}
\end{eqnarray*}
form a martingale with null expectations. For their variances it holds
\begin{eqnarray*}
\operatorname {Var}(Z_{n})=7-\frac{2\pi^{2}}{3}- 
\dfrac{2\log_{e}(n)}{n}+O \left(\dfrac{1}{n}\right).
\end{eqnarray*}
\end{theorem}
\begin{proof} 
By induction on $n$, the base case being trivial. Suppose we have an
$(n-2)$-vertex tree and consider the insertion of the $(n-1)$-th key in the
random binary tree. Its
depth of insertion is $D_{n-1}$, so that
a formerly external node becomes an internal and we see that its
two (new) descendants are both at depth $D_{n-1}+1$. Recall that $D_{n}$ is the random variable
counting the depth of insertion of a key in a random binary tree of size $n-1$, thus $D_{1}:=0$
and $D_{2}:=1$.
The following equation concerning conditional expectations holds: 
\begin{equation*}
n \mathbb{E}(D_{n}\vert D_{1}, \ldots, D_{n-1})=
(n-1) \mathbb{E}(D_{n-1}\vert D_{1}, \ldots, D_{n-2})-D_{n-1}+2(D_{n-1}+1).
\end{equation*}
This recurrence yields
\begin{equation*}
n \mathbb{E}(D_{n}\vert D_{1}, \ldots, D_{n-1})-(n-1) \mathbb{E}(D_{n-1}\vert D_{1}, \ldots, D_{n-2})
=D_{n-1}+2.
\end{equation*}

Summing and using that the left-hand side is a telescopic sum,
\begin{align*}
n\mathbb{E}(D_{n}\vert D_{1}, \ldots, D_{n-1})=\sum_{i=1}^{n-1}(D_{i}+2)=Y_{n-1}.
\end{align*}
The last equation is justified by Lemma $2.6.5$. Also, we have
\begin{align*}
\mathbb{E}(Y_{n}\vert D_{1}, \ldots, D_{n-1})=\mathbb{E}(Y_{n-1}+2+D_{n}\vert D_{1}, \ldots, D_{n-1})=
\frac{n+1}{n}Y_{n-1}+2.
\end{align*}
Thus, taking expectations and using $\mathbb{E}(\mathbb{E}(U\vert V))=\mathbb{E}(U)$,
\begin{align*}
\mathbb{E}(Y_{n})=\dfrac{n+1}{n}\mathbb{E}(Y_{n-1})+2
\Longleftrightarrow \dfrac{\mathbb{E}(Y_{n})}{n+1}=
\frac{\mathbb{E}(Y_{n-1})}{n}+\frac{2}{n+1}.
\end{align*}
This recurrence has solution
\begin{align*}
\mathbb{E}(Y_{n})=2(n+1)(H_{n+1}-1).
\end{align*}

For $Z_{n}$, we deduce that
\begin{align*}
\mathbb{E}(Z_{n}\vert D_{1}, \ldots, D_{n}) &= \mathbb{E} \left(\frac{Y_{n}- 
\mathbb{E}(Y_{n})}{n+1}\vert D_{1}, \ldots, D_{n} \right) \\
& =\dfrac{2}{n+1}
+\dfrac{Y_{n-1}}{n}-\dfrac{ \mathbb{E}(Y_{n})}{n+1} \\
& =Z_{n-1}.
\end{align*}
Therefore, $Z_{n}$ form a martingale.
Further, note that $Z_{n}$ is a linear transformation of the
internal path length $X_{n}$, so we get that
\begin{align*}
\operatorname {Var}(Z_{n})=\frac{1}{(n+1)^{2}}\cdot \operatorname {Var}(X_{n}).
\end{align*}
Previously, we saw that the number of comparisons is just the internal path 
length of a binary search tree. 
As a reminder the variance of the number of comparisons is equal to 
\begin{eqnarray*}
\operatorname {Var}(C_{n})= 7n^{2}-4(n+1)^{2}H_{n}^{(2)}-2(n+1)H_{n}+13n.
\end{eqnarray*}
Hence,
\begin{equation*}
\operatorname {Var}(Z_{n})=7\left(\frac{n}{n+1}\right)^{2}-4H_{n}^{(2)} -
\frac{2H_{n}}{n+1}+\frac{13n}{(n+1)^{2}}.
\end{equation*}

For the purpose of obtaining the asymptotics of the variance, an important family of functions,
which is called polygamma functions are discussed. The digamma function is 
\begin{equation*}
\psi(z)= \dfrac{\mathrm{d}}{\mathrm{d}z}\log_{e}\Gamma(z)
\end{equation*}
and for complex variable $z \neq 0, -1, -2, \ldots$ can be written as \cite{abr},
\begin{align*}
\psi(z) = -\gamma + \sum_{j=0}^{\infty} \left (\dfrac{1}{j+1}-\dfrac{1}{j+z} \right). \, \, \tag{2.7}
\end{align*}
In general, $\forall k \in \mathbf N$, the set 
\begin{equation*}
\psi^{(k)}(z)=\dfrac{\mathrm{d}^{k+1}}{\mathrm{d}z^{k+1}} \log_{e} \Gamma(z),
\end{equation*}
with $\psi^{(0)}(z)= \psi(z)$, forms the family of polygamma functions. 
Differentiating Eq. (2.7),
\begin{align*}
\psi^{(1)}(z) =  \sum_{j=0}^{\infty}  \dfrac{1}{(j+z)^{2}}. \, \, \tag{2.8}
\end{align*}
By Eq. (2.7), it easily follows that 
\begin{equation*}
H_{n}= \psi(n+1) + \gamma.
\end{equation*}
Further, using the fact that $\zeta(2)=\lim_{n \to \infty} H^{(2)}_{n}=\dfrac{\pi^{2}}{6}$ \cite{abr},
where $\zeta(s)=\sum_{j=1}^{\infty}\dfrac{1}{j^{s}}$ is the Riemann zeta function
for $\mathfrak{Re}(s) > 1$, we obtain
\begin{equation*}
H_{n}^{(2)} +\frac{2H_{n}}{n+1} = \frac{\pi^{2}}{6}-\psi^{(1)}(n+1)+\frac{2(\psi(n+1)+ \gamma)}{n+1}. \, \, \tag{2.9}
\end{equation*}

Eq. (2.9) is asymptotically equivalent to
\begin{equation*}
\frac{\pi^{2}}{6}+\dfrac{2\log_{e}(n)}{n},
\end{equation*}
thus, the asymptotic variance is
\begin{align*}
\operatorname {Var}(Z_{n})& =7-\frac{2\pi^{2}}{3} - \dfrac{2\log_{e}(n)}{n}+ O \left( \dfrac{1}{n} \right) \\
& = 7-\frac{2\pi^{2}}{3}- O \left( \dfrac{\log_{e}(n)}{n} \right). \qedhere
\end{align*}
\end{proof}
\begin{remark}
We note that there is a typo in the expression for the asymptotic variance 
in R\'{e}gnier's paper \cite{reg},
which is given as
\begin{equation*}
7-\frac{2\pi^{2}}{3} + O \left ( \dfrac{1}{n} \right).
\end{equation*}
The correct formula for the asymptotic variance is stated in Fill and Janson \cite{asym}.
\end{remark}

The key point about martingales is that they converge.
\begin{theorem}[Feller \cite{feller}]
Let $Z_{n}$ be a martingale, and suppose further that there is a 
constant $C$ such that $\mathbb{E}(Z_{n}^{2})<C$ for all $n$. 
Then there is a random
variable $Z$ to which $Z_{n}$ converges, with probability $1$.
Further, $\mathbb{E}(Z_{n})=\mathbb{E}(Z)$ for all $n$.
\end{theorem}
We showed that
\begin{equation*}
\displaystyle Z_{n}\stackrel{\mathcal{P}}{\longrightarrow} Z.
\end{equation*}
It is important to emphasise that the random variable $Z$ to which
we get convergence is not normally distributed.  
We saw that the total number of comparisons to sort an array of $n \geq 2$ keys, when the pivot is 
a uniform random variable on $\{1, 2, \ldots, n\}$ is equal to the number of comparisons to 
sort the subarray of $U_{n}-1$ keys below pivot plus the number of comparisons 
to sort the subarray of $n-U_{n}$ elements above pivot plus $n-1$ comparisons 
done to partition the array. Therefore,
\begin{equation*}
X_{n}=X_{U_{n}-1}+X^{*}_{n-U_{n}}+n-1,
\end{equation*}
where the random variables $X_{U_{n}-1}$ and $X^{*}_{n-U_{n}}$ are identically
distributed and independent conditional on $U_{n}$.

Consider the random variables  
\begin{equation*}
Y_{n}=\frac{X_{n}-\mathbb{E}(X_{n})}{n}.
\end{equation*}
The previous equation can be rewritten in the following form
\begin{eqnarray*}
Y_{n}=\frac{X_{U_{n}-1}+X^{*}_{n-U_{n}}+n-1-\mathbb{E}(X_{n})}{n}.
\end{eqnarray*}
By a simple manipulation, it follows that \cite{asym}, \cite{Rosler}
\begin{equation*}
Y_{n}= Y_{U_{n}-1}\cdot \frac{U_{n}-1}{n}+Y^{*}_{n-U_{n}}\cdot 
\frac{n-U_{n}}{n}+C_{n}(U_{n}),
\end{equation*}
where
\begin{equation*}
C_{n}(j)=\frac{n-1}{n}+\frac{1}{n}\left(\mathbb{E}(X_{j-1})+ 
\mathbb{E}(X^{*}_{n-j})-\mathbb{E}(X_{n})\right).
\end{equation*}

The random variable $U_{n}/n$ converges to a uniformly distributed variable 
$\Xi$ on $[0, 1]$. A Lemma follows
\begin{Lemma}
Let $U_{n}$ be a uniformly distributed random variable on $\{1, 2, \ldots , n\}$. Then 
$$\displaystyle\frac{U_{n}}{n}\stackrel{\mathcal{D}}{\longrightarrow} \Xi,$$ 
where $\Xi$ is uniformly distributed on $[0, 1]$.
\end{Lemma} 
\begin{proof}
The moment generating function of $U_{n}$ is given by
\begin{eqnarray*}
M_{U_{n}}(s)=\sum_{k=1}^{n}\mathbb{P}(U_{n}=k)e^{sk}=\frac{1}{n} \cdot \sum_{k=1}^{n}e^{sk}
=\frac{1}{n}\cdot \frac{e^{s(n+1)}-e^{s}}{e^{s}-1}.
\end{eqnarray*}
For the random variable $U_{n}/n$, it is 
\begin{eqnarray*}
M_{U_{n}/n}(s)=M_{U_{n}}(s/n)=\frac{1}{n}\cdot\sum_{k=1}^{n}e^{sk/n}=
\frac{1}{n}\cdot \dfrac{e^{s(n+1)/n}-e^{s/n}}{e^{s/n}-1}.
\end{eqnarray*}

The random variable $\Xi$ has moment generating function
\begin{eqnarray*}
M_{\Xi}(s)=\int_{0}^{1}e^{st}\, \mathrm{d}t=\frac{e^{s}-1}{s}.
\end{eqnarray*}
Now the moment generating function of $U_{n}/n$ is an approximation to the average value of
$e^{sx}$ over the interval $[0,1]$ and so, as $n$ tends to infinity,
we can replace it by its integral
$$\int_{0}^{1}e^{sx}\, \mathrm{d}x=\left [\frac{e^{sx}}{s} \right]_{0}^{1}=\frac{e^{s}-1}{s}$$
and now all that is required has been proved. 
\end{proof}
For the function
\begin{eqnarray*} 
C_{n}(j)=\frac{n-1}{n}+\frac{1}{n}\cdot \left(\mathbb{E}(X_{j-1})+\mathbb{E}(X^{*}_{n-j})
-\mathbb{E}(X_{n})\right)
\end{eqnarray*}
using the previous Lemma and recalling that asymptotically the expected complexity 
of Quicksort is $2n\log_{e}(n)$ it follows that
\begin{align*}
\lim_{n \to \infty}C_{n}(n\cdot U_{n}/n)& =\lim_{n \to \infty}\left(\frac{n-1}{n}+
\frac{1}{n}\cdot \bigl(\mathbb{E}(X_{U_{n}-1})+\mathbb{E}(X^{*}_{n-U_{n}})- 
\mathbb{E}(X_{n}) \bigr )\right) \\                    
& =\lim_{n \to \infty}\Biggl(\frac{n-1}{n}+\frac{1}{n} \cdot 
\biggl(2\left(\frac{n \cdot U_{n}}{n}-1 \right)\log_{e}(U_{n}-1) \\
& \qquad{} + 2\left(n-\frac{n \cdot U_{n}}{n} \right)
\log_{e}(n-U_{n})-2n\log_{e}(n)\biggr) \Biggr) \\
& =1+2\xi\log_{e}\xi+2(1-\xi) \log_{e}(1-\xi)=C(\Xi),  \mbox{~ ~$\forall ~ \xi \in [0, 1].$~}
\end{align*}
Thus $C_{n}(n\cdot U_{n}/n)$ converges to $C(\Xi)$, (see as well in \cite{Rosler}).
Therefore, we obtain
\begin{eqnarray*}
\mathscr{L}(Y)= \mathscr{L}\bigl(Y \cdot \Xi+Y^{*} \cdot (1-\Xi)+C(\Xi) \bigr).
\end{eqnarray*}

But this does not work for the normal. Indeed, if $Y$ (and hence $Y^{*}$)
are normals with mean zero and variance $\sigma^{2}$, a necessary condition
for $Y\Xi+Y^{*}(1-\Xi)+C(\Xi)$ to have mean 0 (which would be needed for the equality
to hold) would be that $C(\Xi)=0$.
This equality happens with probability equal to $0$, as we can easily deduce.
Thus, the distribution which the sequence $Y_{n}$ converges is not Gaussian.

\subsection{Deviations of Quicksort}

It is desirable to derive bounds on the deviation of the random number of pairwise comparisons, 
needed to sort an array of $n$ distinct elements from its expected value, as $n$ gets 
arbitrarily large. We remind ourselves that the pivot is uniformly chosen at random. 

We shall derive bounds on the following probability
\begin{eqnarray*}
\mathbb{P}\left(\left \vert\frac{C_{n}}{\mathbb{E}(C_{n})}-1 \right \vert > \epsilon \right)
\end{eqnarray*}
for $\epsilon>0$, sufficiently small. From Chebyshev's inequality, we obtain a bound for the 
above probability. Chebyshev's inequality is as follows \cite{feller}:
\begin{theorem}[Chebyshev's inequality]
Let $X$ be a random variable, with $\mathbb{E}(X^{2}) < \infty$.
Then, for any real number $a>0$ 
\begin{eqnarray*}
\mathbb{P} \bigl(\vert X \vert \geq a \bigr) \leq \dfrac{\mathbb{E}(X^{2})}{a^{2}}.
\end{eqnarray*}
If $\mathbb{E}(X)=m$ and $\operatorname {Var}(X)=\sigma^{2}$, then
\begin{equation*}
\mathbb{P} \bigl(\vert X-m \vert \geq a \bigr) \leq \dfrac{\sigma^{2}}{a^{2}}.
\end{equation*}
\end{theorem}
The random variable $Y_{n}=\dfrac{C_{n}}{\mathbb{E}(C_{n})}$ 
has mean and variance,
\begin{align*}
\mathbb{E}(Y_{n})& = \mathbb{E} \left(\frac{C_{n}}{\mathbb{E}(C_{n})}\right)=
\frac{1}{\mathbb{E}(C_{n})}\cdot \mathbb{E}(C_{n})=1 \\
\\
\operatorname {Var}(Y_{n})& = \operatorname {Var} \left(\frac{C_{n}}{\mathbb{E}(C_{n})}\right)
=\frac{1}{\mathbb{E}^{2}(C_{n})}\cdot \operatorname {Var}(C_{n}) \\
& =\frac{-4(n+1)^{2}H^{(2)}_{n}-2(n+1)H_{n}+(7n+13)n}{\bigl(2(n+1)H_{n}-4n \bigr)^{2}}.
\end{align*}

For $\epsilon>0$, we have
\begin{eqnarray*}
\mathbb{P}(\vert Y_{n}-1\vert>\epsilon)<\frac{\operatorname {Var}(Y_{n})}{\epsilon^{2}}=
\frac{-4(n+1)^{2}H^{(2)}_{n}-2(n+1)H_{n}+(7n+13)n}{\epsilon^{2}\bigl (2(n+1)H_{n}-4n\bigr )^{2}}.
\end{eqnarray*}
It holds that,
\begin{align*}
\operatorname {Var}(C_{n})&= -4(n+1)^{2}H^{(2)}_{n}-2(n+1)H_{n}+(7n+13)n \\
& \qquad {} \leq 7n^{2}+13n \\
& \qquad {} \leq 20n^{2},
\end{align*}
using that the other terms are negative and $n\geq 1$. Further, using that
\begin{eqnarray*}
\lim_{n \to \infty}H^{(2)}_{n}=\sum_{n=1}^{\infty}\frac{1}{n^{2}}=\frac{\pi^{2}}{6}, 
\end{eqnarray*}
we get
\begin{align*}
\operatorname {Var}(C_{n})&=-4(n+1)^{2}H^{(2)}_{n}-2(n+1)H_{n}+(7n+13)n \\
& \quad {} \geq (7n+13)n-4(n+1)^{2}\pi^{2}/6-2(n+1)\bigl(\log_{e}(n)+\gamma +o(1)\bigr) \\
& \qquad{} = (7-2\pi^{2}/3)n^{2}\bigl(1+o(1) \bigr).
\end{align*}
From above inequalities, we deduce that $\operatorname {Var}(C_{n})=\Theta(n^{2})$. 
Thus,
\begin{eqnarray*}
\mathbb{P}(\vert Y_{n}-1 \vert > \epsilon)<\frac{\operatorname {Var}(Y_{n})}{\epsilon^{2}}
=\frac{\Theta(n^{2})}{\epsilon^{2}\bigl(2(n+1)H_{n}-4n \bigr)^{2}}
=O \left((\epsilon \log_{e}(n))^{-2} \right).
\end{eqnarray*}
R\"{o}sler \cite{Rosler} derived a sharper bound. He showed that this probability is 
$O(n^{-k})$, for fixed $k$.
McDiarmid and Hayward \cite{mc} have further sharpened the bound. 
Their Theorem is
\begin{theorem}[McDiarmid and Hayward {\cite{mc}}, McDiarmid {\cite{diarmid}}]\ \\
Let $\epsilon=\epsilon(n)$ satisfy $\dfrac{1}{\log_{e}(n)}< \epsilon \leq 1$. 
Then as $n \to \infty$,
\begin{eqnarray*}
\mathbb{P}(\vert Y_{n}-1 \vert > \epsilon)= n^{-2\epsilon \bigl (\log_{e}^{(2)}(n)-\log_{e}(1/ \epsilon)
+ O(\log_{e}^{(3)}(n) \bigr )},
\end{eqnarray*}
where $\log_{e}^{(k+1)}(n):=\log_{e}\bigl(\log_{e}^{(k)}(n)\bigr)$ and 
$\log_{e}^{(1)}n=\log_{e}n$.
\end{theorem}

In the recent paper of McDiarmid \cite{diarmid}, Theorem $2.6.11$ is revisited
using concentration arguments. By this result, it can be easily deduced, that Quicksort 
with good pivot choices performs well, with
negligible perturbations from its expected number of comparisons.
As we previously saw, Quicksort can be depicted as an ordered binary tree.
The root node or node $1$ corresponds 
to the input array of length $L_{1}=n$ to be sorted.

A pivot with rank $U_{n}=\{1, 2, \ldots, n\}$ is selected uniformly
at random and the initial array is 
divided into two subarrays, one with elements less than the pivot 
and a second, with elements greater than the pivot. Then, the 
node at the left corresponds to the subarray of keys that are
less than the pivot, with length 
$L_{2}=U_{n}-1$ and the node at the right corresponds to the 
subarray of keys that are greater than the pivot, with length $L_{3}=n-U_{n}$. 

Recursively, Quicksort 
runs on these two subarrays and split them in four subarrays, 
until we get trivial subarrays and the initial array is sorted. 
For $j=1,2 \ldots$ let $L_{j}$ be the length of the array to be sorted 
at $j$ node and $M^{n}_{k}$ be the maximum
cardinality of the $2^{k}$ subarrays, after $k$ recursions of Quicksort. 
It is \cite{mc},
\begin{equation*}
M^{n}_{k} = \mathrm {\max} \{L_{2^{k}+i}: i=0, 1,\ldots, 2^{k}-1 \}.
\end{equation*}

The following Lemma gives an upper
bound for the probability that the maximum length of $2^{k}$ subarrays
$M^{n}_{k}$, will exceed $\alpha$ times the initial set's length
$n$. We easily see that the upper bound is rather small quantity. 
Thus, we deduce that the length of the array is 
rapidly decreasing, as Quicksort runs. 
\begin{Lemma}[McDiarmid and Hayward \cite{mc}, McDiarmid \cite{diarmid}]\ \\
For any $0<\alpha<1$ and any integer $k \geq \log_{e}\left(\dfrac{1}{a} \right)$ it holds 
\begin{equation*}
\mathbb{P}\left(M^{n}_{k} \geq \alpha n \right)\leq \alpha \left(\frac{2e\cdot \log_{e}(1/\alpha)}{k} 
\right)^{k}.
\end{equation*}
\end{Lemma}

So far, the analysis explicitly assumed the presence of distinct numbers. 
However, in many sorting problems, one can come across with duplicates.
As we will see next, we have to consider the presence of equal numbers and how this 
affects the algorithm's performance.

\section{Quicksorting arrays with repeated elements}

In this section, we consider the
presence of equal keys. In some cases, there may be multiple
occurrences of some of the keys -- so that there is a multiset of
keys to be sorted. If a key
$a_{i}$ appears $s_{i}$ times, we call $s_{i}$ the multiplicity of
$a_{i}$.

Thus, we have a multiset 
of the array $\{s_{1}\cdot a_{1}, s_{2}\cdot a_{2}, \ldots, s_{n}\cdot a_{n} \}$, 
with $s_{1}+s_{2} + \ldots + s_{n} = N$ and $n$ 
is the number of distinct keys. Without loss of generality, we assume that we
have to sort a random permutation of the array $\{s_{1}\cdot 1, s_{2}\cdot 2, \ldots, s_{n}\cdot n\}$.
Obviously, when $s_{1} =s_{2}= \ldots = s_{n} =1$, then the array contains 
$n=N$ distinct keys. 

Consider an array consisting of keys with ``large'' multiplicities. 
The usual Quicksort algorithm is quite likely to perform badly, since keys equal
to the pivot are not being exchanged. 
Recall that Hoare's partitioning routine \cite{hoare} described
in the previous Chapter, utilises two pointers 
that scan for keys greater than and less
than the pivot, so at the end of partition, 
keys equal to pivot may be on either or both subarrays,
leading to unbalanced partitioning which we saw 
is usually undesirable.
The algorithm can be further
tweaked for efficient sorting of duplicates, 
using a ternary partition scheme \cite{ben}. 
Keys less than pivot are on left, keys equal to pivot on middle 
and on right, keys greater than the pivot. 

Hoare's partitioning scheme can be easily modified, in order 
for the pointers to stop on 
equal keys with the pivot. In other words, the lower pointer searches
for a key greater than or equal to the pivot and the 
upper one for a key smaller than or equal to the pivot.
Sedgewick \cite{se} considers a partitioning routine, where
only one of the pointers stops on keys either greater or smaller than
the pivot.

The recurrence for the expected number of key comparisons 
$\mathbb {E}\bigl(C(s_{1}, s_{2}, \ldots, s_{n})\bigr )$ is
\begin{eqnarray*}
\mathbb{E} \bigl (C(s_{1}, s_{2}, \ldots, s_{n}) \bigr) = N-1 + \dfrac{1}{N} 
\sum_{i=1}^{n}s_{i} \biggl ( \mathbb{E} \bigl (C(s_{1}, \ldots, s_{i-1}) \bigr ) +
\mathbb{E} \bigl (C(s_{i+1}, \ldots, s_{n})\bigr) \biggr ).
\end{eqnarray*}
The analysis has been done in \cite{se}. The solution of the recurrence
is 
\begin{equation*}
\mathbb {E}\bigl(C(s_{1}, s_{2}, \ldots, s_{n})\bigr )= 
2\left (1+\frac{1}{n} \right )NH_{n}-3N-n. \, \, \tag{2.10}
\end{equation*}

Note that when $n=N$, i.e. the keys to be sorted are distinct, Eq. (2.10)
yields the expected number of comparisons, 
when the array is partitioned using $n-1$ comparisons. As we saw in the previous 
Chapter, a different partitioning scheme proposed and analysed in \cite{knuth 3},
\cite{sedgewick}, utilises $N+1$ comparisons for partitioning an array of $N$ keys. In this
case, an upper bound for large $n$ to the expected number of 
comparisons is \cite{se}, 
\begin{equation*}
\mathbb {E}\bigl(C(s_{1}, s_{2}, \ldots, s_{n})\bigr )=
2N(H_{N}+1)-2 + O \left(\frac{N^{2}}{n} \right).
\end{equation*}

\chapter{Sorting by sampling}

The preceding analysis has shown that Quicksort is prone to poor performance,
when the chosen pivots happen to be (close to) the smallest or greatest keys
in the array to be sorted. The partitioning yields trivial subarrays,
leading to quadratic running time taken by the algorithm, and increasing the chances of
`crashing', i.e. an application of Quicksort may terminate 
unsuccessfully
through running out of memory. On the other hand, good choices of pivot yield
a much more efficient algorithm. The uniform model suggests that in fact
any key has equal probability to be selected as pivot. In this Chapter,
we discuss and analyse the general idea of how one can increase the
probability that the selected pivots will produce (reasonably) balanced
partitions, by trying to ensure that their ranks are `near' the middle of the array.

A naive idea towards this, would be finding the median of the keys and use
this as pivot. However, it is obvious that the finding of median imposes
additional costs to the algorithm. Therefore, despite the always good choices, 
this method might not be
better than choosing the pivot randomly. Instead of finding the median of the
array to be sorted, it might be more efficient to randomly pick a sample from the array, 
find its median and use this as pivot. In what it follows, we will
analyse this idea and its variants.

\section{Median of (2k+1) partitioning}

Singleton \cite{Singleton} suggested to randomly select three keys from the array to be
sorted and use their median as pivot, leading to a better estimate of the
median of the array and reducing further the chances of
worst case occurrence. 

This modification can be generalised as to choosing a 
sample of $(2k+1)$ keys at every recursive stage, computing their median and
using that median as pivot, partitioning $n>2k+1$ 
keys. Arrays that contain at 
most $(2k+1)$ keys are sorted by a simpler algorithm,
such as insertion sort. The cost of sorting these small arrays is linear with 
respect to $n$.

Letting $C_{n\{2k+1\}}$ denote the number of
comparisons required to sort $n$ keys when the pivot is the median of
a random sample of $(2k+1)$ elements, uniformly selected from the relevant array, 
the recurrence for the average
number of comparisons is given by (see \cite{knuth 3})
\begin{eqnarray*}
\mathbb{E} (C_{{n}\{2k+1\}})= n+1 + \dfrac{2}{\dbinom{n}{2k+1}}\sum_{j=1}^{n}
\dbinom{j-1}{k}\dbinom{n-j}{k}\mathbb{E}(C_{{j-1}\{2k+1\}}).
\end{eqnarray*}
Multiplying both sides by $\dbinom{n}{2k+1}$, 
\begin{eqnarray*}
\dbinom{n}{2k+1}\mathbb{E} (C_{{n}\{2k+1\}})= \dbinom{n}{2k+1}(n+1) + 2 \sum_{j=1}^{n}
\dbinom{j-1}{k}\dbinom{n-j}{k}\mathbb{E}(C_{{j-1}\{2k+1\}}).
\end{eqnarray*}
Multiplying by $z^{n}$ and summing over $n$, in order to obtain the generating function 
for the expected number of comparisons, 
$f(z)=\sum_{n=0}^{\infty}\mathbb{E}(C_{n\{2k+1\}})z^{n}$
\begin{align*}
& \sum_{n=0}^{\infty}\dbinom{n}{2k+1}\mathbb{E} (C_{{n}\{2k+1\}})z^{n} = \\
\quad {} & \sum_{n=0}^{\infty}\dbinom{n}{2k+1}(n+1)z^{n} 
+ 2\sum_{n=0}^{\infty}\sum_{j=1}^{n}\dbinom{j-1}{k}\dbinom{n-j}{k}\mathbb{E}(C_{{j-1}\{2k+1\}})z^{n}. \, \, \tag{3.1}
\end{align*}
It holds 
\begin{align*}
\sum_{n=0}^{\infty}\dbinom{n}{2k+1}\mathbb{E} (C_{{n}\{2k+1\}})z^{n} & = \dfrac{1}{(2k+1)!}
\sum_{n=0}^{\infty}n(n-1) \ldots (n-2k)\mathbb{E} (C_{{n}\{2k+1\}})z^{n} \\
& = \dfrac{z^{2k+1}f^{(2k+1)}(z)}{(2k+1)!},
\end{align*}
where $f^{(2k+1)}(z)$ denotes the $(2k+1)$-th order derivative of $f(z)$. For the 
first sum in the right-hand side of Eq. (3.1)
\begin{align*}
\sum_{n=0}^{\infty}\dbinom{n}{2k+1}(n+1)z^{n} & = \dfrac{z^{2k+1}}{(2k+1)!} 
\biggl (\sum_{n=0}^{\infty} z^{n+1} \biggr )^{(2k+2)} \\
& = \dfrac{z^{2k+1}}{(2k+1)!} \biggl (\dfrac{z}{1-z} \biggr )^{(2k+2)}.
\end{align*}
The $(2k+2)$-th order derivative of $\bigl (z/(1-z)\bigr )$ can be easily 
seen by induction that is equal to
\begin{equation*}
\dfrac{(2k+2)!}{(1-z)^{2k+3}}.
\end{equation*}

Expanding out the double sum of Eq. (3.1),
\begin{align*}
\sum_{n=0}^{\infty}\sum_{j=1}^{n}\dbinom{j-1}{k}\dbinom{n-j}{k}\mathbb{E}(C_{{j-1}\{2k+1\}})z^{n} = 
\binom{0}{k} \binom{0}{k}\mathbb{E}(C_{{0}\{2k+1\}})z  + \\
& \kern-27em \Biggl (\binom{0}{k} \binom{1}{k}\mathbb{E}(C_{{0}\{2k+1\}}) + \binom{1}{k} 
\dbinom{0}{k}\mathbb{E}(C_{{1}\{2k+1\}}) \Biggr )z^{2} + \\
& \kern-27em  \Biggl (\binom{0}{k} \binom{2}{k}\mathbb{E}(C_{{0}\{2k+1\}}) + \binom{1}{k} 
\dbinom{1}{k}\mathbb{E}(C_{{1}\{2k+1\}}) + \binom{2}{k} 
\dbinom{0}{k}\mathbb{E}(C_{{2}\{2k+1\}})\Biggr )z^{3} + \ldots \\
& \hspace{-26 em} = \Biggl (\binom{0}{k}\mathbb{E} (C_{{0}\{2k+1\}})+ \binom{1}{k}\mathbb{E} 
(C_{{1}\{2k+1\}}) z + \ldots \Biggr ) \Biggl (\binom{0}{k}z + \binom{1}{k}z^{2}+ \ldots \Biggr ) \\
& \hspace{-26 em} = \Biggl ( \sum_{n=0}^{\infty} \dbinom{n}{k}\mathbb{E} (C_{{n}\{2k+1\}})  z^{n}\Biggr) \Biggl 
(\sum_{n=0}^{\infty}\dbinom{n}{k} z^{n+1} \Biggr ) \\
&\hspace{-26 em} = \dfrac{z^{k}f^{(k)}(z)}{k!}\cdot  \dfrac{z^{k+1}}{k!} 
\biggl ( \sum_{n=0}^{\infty}z^{n} \biggr )^{(k)} \\
& \hspace{-26 em} = \dfrac{z^{2k+1}f^{(k)}(z)}{k!(1-z)^{k+1}}.
\end{align*}
The recurrence is transformed to the following differential equation
\begin{eqnarray*}
\dfrac{z^{2k+1}f^{(2k+1)}(z)}{(2k+1)!} = \dfrac{(2k+2)z^{2k+1}}{(1-z)^{2k+3}} 
+ \dfrac{2z^{2k+1}f^{(k)}(z)}{k!(1-z)^{k+1}}.
\end{eqnarray*}
Multiplying both sides by $\left (\dfrac{1-z}{z} \right )^{2k+1}$, 
\begin{eqnarray*}
\dfrac{(1-z)^{2k+1}f^{(2k+1)}(z)}{(2k+1)!} = \dfrac{2(k+1)}{(1-z)^{2}} + \dfrac{2(1-z)^{k}f^{(k)}(z)}{k!},
\end{eqnarray*}
which is a Cauchy--Euler differential equation. This type of differential
equations arises naturally to the analysis of
searching--sorting algorithms and urn models \cite{hc}. Substituting $x=1-z$, and
putting $h(x)=f(1-x)$, we get 
\begin{eqnarray*}
\dfrac{(-1)^{2k+1}x^{2k+1}h^{(2k+1)}(x)}{(2k+1)!} = \dfrac{2(k+1)}{x^{2}}
+ \dfrac{2(-1)^{k}x^{k}h^{(k)}(x)}{k!}.
\end{eqnarray*}

We use the differential operator $\Theta$ for the solution of the differential equation.
It is defined by
\begin{eqnarray*}
\Theta \bigl (h(x) \bigr):= xh^{\prime}(x)
\end{eqnarray*}
and by induction 
\begin{eqnarray*}
\dbinom{\Theta}{k}h(x) = \dfrac{x^{k}h^{(k)}(x)}{k!}.
\end{eqnarray*}
Applying the operator, our equation becomes
\begin{eqnarray*}
\mathcal{P}_{2k+1}(\Theta)h(x) = \dfrac{2(k+1)}{x^{2}},
\end{eqnarray*}
where the indicial polynomial is equal to
\begin{align*}
\mathcal{P}_{2k+1}(\Theta) & = (-1)^{2k+1} \dbinom{\Theta}{2k+1}
 -2(-1)^{k} \dbinom{\Theta}{k}.
\end{align*}

We proceed to identify the nature of the roots of the polynomial.
A Lemma follows:
\begin{Lemma}
The indicial polynomial $\mathcal{P}_{2k+1}(\Theta)$ has $2k+1$ simple roots,
with real parts greater than or equal to $-2$. The real roots are $0, 1, \ldots , k-1$;
$-2$; $3k+2$, if $k$ is odd and the $2 \left \lfloor \frac{k}{2} \right \rfloor$ 
complex roots $\rho_{1}, \ldots , 
\rho_{\left \lfloor \frac{k}{2} \right \rfloor }$ with their conjugates 
$\overline{\rho}_{1}, \ldots , \overline{\rho}_{\left \lfloor \frac{k}{2} \right \rfloor}$.
\end{Lemma}
\begin{proof}
Let $\alpha=x + iy$ being a root of the polynomial. Then
\begin{align*}
\dbinom{\alpha}{2k+1} = -2(-1)^{k}\dbinom{\alpha}{k}. \, \, \tag{3.2}
\end{align*}
From Eq. (3.2), we deduce that $k$ real roots of the polynomial $\mathcal{P}_{2k+1}(\Theta)$ 
are $0, 1, \ldots, k-1$. 
For $\alpha\neq 0, 1, \ldots, k-1$, we have that
\begin{align*}
\dfrac{(\alpha-k)}{(k+1)}\cdot \dfrac{(\alpha-k-1)}{(k+2)}\cdot \ldots \cdot \dfrac{(\alpha-2k)}{(2k+1)} 
= -2(-1)^{k}. \, \, \tag{3.3}
\end{align*}

Suppose that $\mathfrak{Re}(\alpha) < -2.$ Then
\begin{align*}
\left \vert \frac{\alpha-k}{k+1} \right \vert & = \frac {\sqrt{(x-k)^{2}+y^{2}}}{k+1} 
>\frac{\sqrt{\bigl (-(k+2) \bigr )^{2}+y^{2}}}{k+1} \geq \dfrac{k+2}{k+1}. \\
\left \vert \frac{\alpha-k-1}{k+2} \right \vert & = \frac {\sqrt{(x-k-1)^{2}+y^{2}}}{k+2}
>\frac{\sqrt{\bigl (-(k+3) \bigr )^{2}+y^{2}}}{k+2} \geq \dfrac{k+3}{k+2}.  \\
& \setbox0\hbox{=}\mathrel{\makebox[\wd0]{\vdots}}  \\
\left \vert \frac{\alpha-2k}{2k+1} \right \vert & = \frac {\sqrt{(x-2k)^{2}+y^{2}}}{2k+1}
> \frac{\sqrt{\bigl (-2(k+1) \bigr )^{2}+y^{2}}}{2k+1} \geq \dfrac{2k+2}{2k+1},
\end{align*}
so the product of the moduli in Eq. (3.3) is greater than or equal to the 
telescoping product $(2k+2)/(k+1)=2$, arriving in contradiction. Thus,
for any root $\alpha$ of the polynomial, it is proved that
$\mathfrak{Re}(\alpha) > -2$. Moreover, this argument shows that $-2$ is the unique 
root with real part equal to $-2$.
To see all roots are simple, assume that there does exist a repeated root
$\alpha$. Differentiating the polynomial,
\begin{eqnarray*}
\mathcal{P}^{\prime}_{2k+1}(\Theta) =  -\dbinom{\Theta}{2k+1}\sum_{j=0}^{2k} 
\dfrac{1}{\Theta - j} -2(-1)^{k} \dbinom{\Theta}{k}\sum_{j=0}^{k-1}\dfrac{1}{\Theta-j}
\end{eqnarray*}
and 
\begin{align*}
\dbinom{\alpha}{2k+1}\sum_{j=0}^{2k} \dfrac{1}{\alpha - j} = -2(-1)^{k} 
\dbinom{\alpha}{k}\sum_{j=0}^{k-1}\dfrac{1}{\alpha-j}. \, \, \tag{3.4}
\end{align*}
Eq. (3.2) and (3.4) imply that for a repeated root must hold
\begin{eqnarray*}
\sum_{j=0}^{2k} \dfrac{1}{\alpha - j}=\sum_{j=0}^{k-1}\dfrac{1}{\alpha-j}
\end{eqnarray*}
or 
\begin{equation*}
\sum_{j=k}^{2k} \dfrac{1}{\alpha - j}=0. \, \, \tag{3.5}
\end{equation*}
Consequently, Eq. (3.5) implies that $\mathfrak{Im} (\alpha) = 0$ and 
$\alpha \in (k, k+1) \cup (k+1, k+2) \cup \ldots \cup (2k-1, 2k)$.
Considering Eq. (3.2), the left-hand side has modulus less than $1$, since
$k<\alpha < 2k+1$. However, the right-hand side of Eq. (3.2) has modulus
greater than $2$, leading to contradiction. 
\end{proof}

Since the roots are simple, we can factor our polynomial in the form
\begin{eqnarray*}
(\Theta - r_{1})(\Theta - r_{2}) \ldots (\Theta - r_{2k})(\Theta + 2)h(x) = \frac{2k+2}{x^{2}}.
\end{eqnarray*}
The solution to the differential equation $(\Theta - a) h(x) = x^{b}$ is \cite{sedgewick},
\begin{eqnarray*}
h(x) = 
\begin{cases}
\dfrac{x^{b}}{b-a} + cx^{a}, \mbox{~if~} a \neq b \\
\\
x^{b}\log_{e} (x) + cx^{a}, \mbox{~if~} a = b.
\end{cases}
\end{eqnarray*}
Therefore, applying $(2k+1)$ times the solution, we obtain
\begin{eqnarray*}
h(x) = \dfrac{2k+2}{(-2 - r_{1})(-2 - r_{2}) \ldots (-2 - r_{2k})} 
\dfrac{\log_{e}(x)}{x^{2}} + \sum_{j=1}^{2k+1}c_{j}x^{r_{j}},
\end{eqnarray*}
where $c_{j}$ are the constants of integration. Note that
\begin{eqnarray*}
\mathcal {P}_{2k+1}(\Theta) = (\Theta +2)\mathcal {S}_{2k}(\Theta),
\end{eqnarray*}
thus, the expression in the denominator $\mathcal {S}_{2k}(-2)$ is
\begin{align*}
\mathcal {S}_{2k}(-2) & = \mathcal{P}^{\prime}_{2k+1}(-2) \\
& = -\dbinom{-2}{2k+1}\sum_{j=0}^{2k} \dfrac{1}{-2 - j} -2(-1)^{k} 
\dbinom{-2}{k}\sum_{j=0}^{k-1}\dfrac{1}{-2-j} \\
& = -(2k+2)(H_{2k+2}-H_{k+1}).
\end{align*}
Reverting to the previous notation,
\begin{eqnarray*}
f(z) = -\dfrac{1}{H_{2k+2} - H_{k+1}} \dfrac{\log_{e}(1-z)}{(1-z)^{2}} + 
\sum_{j=1}^{2k+1}c_{j}(1-z)^{r_{j}}.
\end{eqnarray*}

Using the identity \cite{con},
\begin{eqnarray*}
\frac{1}{(1-z)^{m+1}}\log_{e}\biggl(\frac{1}{1-z}\biggr)=\sum_{n=0}^{\infty}
\Bigl(H_{n+m}-H_{m}\Bigr)\dbinom{n+m}{m}z^{n}
\end{eqnarray*}
and the binomial Theorem, the solution of the differential equation can be written 
in terms of series,
\begin{eqnarray*}
f(z) = \dfrac{1}{H_{2k+2} - H_{k+1}} \sum_{n=0}^{\infty}\bigl ( (n+1)H_{n} - n \bigr ) z^{n} 
+ \sum_{n=0}^{\infty} \sum_{j=1}^{2k+1}c_{j}(-1)^{n}\dbinom{r_{j}}{n}z^{n}.
\end{eqnarray*}

Extracting the coefficients, the expected number of comparisons of `median of $(2k+1)$' Quicksort is
\begin{align*}
\mathbb{E} (C_{{n}\{2k+1\}}) = \dfrac{1}{H_{2k+2} - H_{k+1}}\bigl ( (n+1)H_{n} - n \bigr ) 
 + \sum_{j=1}^{2k+1} (-1)^{n} \mathfrak{Re} \left ( c_{j}\dbinom{r_{j}}{n} \right).
\end{align*}
The real roots of the polynomial $0, 1, \ldots, (k-1)$ do not contribute to the expected number of 
comparisons, since $n > 2k+1$ and when $k$ is odd, the root $(3k+2)$ adds a negligible constant 
contribution. Further, note that the root $r_{2k+1}=-2$, contributes $c_{2k+1}(n+1)$, with
$c_{2k+1} \in \mathbf R$. Therefore
\begin{align*}
\mathbb{E} (C_{{n}\{2k+1\}})& = \dfrac{1}{H_{2k+2} - H_{k+1}} \bigl ( (n+1)H_{n} - n \bigr ) + 
c_{2k+1}(n+1) \\
& \qquad {} + 2\sum_{j=1}^{\left \lfloor \frac{k}{2} \right \rfloor}(-1)^{n} 
\mathfrak{Re} \left ( c_{j}\dbinom{\rho_{j}}{n} \right) + O(1).
\end{align*}

The asymptotics of the real parts in the last sum is given in \cite{Hen}. 
It is proved that
\begin{eqnarray*}
2(-1)^{n}\mathfrak{Re} \biggl (c_{j} \dbinom{\rho_{j}}{n} \biggr ) = O(n^{- ( \mathfrak{Re}(\rho_{j}) +1 )})
\end{eqnarray*}
and asymptotically the expected number of key comparisons is
\begin{align*}
\mathbb{E} (C_{{n}\{2k+1\}})& = \dfrac{1}{H_{2k+2} - H_{k+1}}n\log_{e}(n) + 
\biggl (c_{2k+1} + \dfrac{1}{H_{2k+2} - H_{k+1}} \bigl(\gamma-1 \bigr) \biggr )n \\
& \qquad {} + o(n),
\end{align*}
since all the other roots have real parts greater than $-2$. 
The solution contains the case of ordinary Quicksort, where the pivot is randomly selected 
and of `median of $3$' Quicksort. The coefficient of the leading term in the latter case is 
\begin{equation*}
\dfrac{1}{H_{4} - H_{2}} = \dfrac{12}{7}, 
\end{equation*}
thus the expected number of key comparisons is
\begin{eqnarray*}
\mathbb{E} (C_{{n}\{3\}}) = \dfrac{12}{7}(n+1)H_{n} + O(n).
\end{eqnarray*}

van Emden \cite{emd} obtained the asymptotic cost of the expected 
number of comparisons, using entropy arguments. 
The leading term is equal to 
\begin{equation*}
an\log_{2}(n),
\end{equation*}
where
\begin{equation*}
a=\dfrac{1}{\mathbb {E}(H)} = -\dfrac{1}{\displaystyle 2 \int_{0}^{1}xg(x)\log_{2}(x)\, \mathrm{d} x}.
\end{equation*}
Here $g(x)$ denotes the probability density of the median of a random sample of 
$(2k+1)$ elements and $\mathbb {E}(H)$ is the expected information yielded from a comparison. 
The nice and simple asymptotic form for the mean number of comparisons is
\begin{eqnarray*}
\mathbb{E}(C_{{n}\{2k+1\}} )\sim \dfrac{\log_{e}(2)}{H_{2k+2}-H_{k+1}}n \log_{2}(n).
\end{eqnarray*}
Note that this result, yields the expected number of comparisons of standard Quicksort, 
for $k=0$. In this case,
\begin{equation*}
a = 2\log_{e}2,
\end{equation*}
thus
\begin{equation*}
\mathbb{E}(C_{n} )\sim 1.386 n \log_{2}(n).
\end{equation*}

The notion of entropy is of great importance to the analysis 
presented in this thesis. 
Entropy is a measure of uncertainty regarding 
the events of a random variable.
In other words, a higher uncertainty about the outcome of a 
random variable pertains to increased entropy.
In the trivial case, where the probability of occurrence of 
an event is $1$, the entropy is equal to $0$,
as there is no uncertainty.
A formal definition follows \cite{sha}.
\begin{Definition}
The Shannon's entropy $H$ of a discrete random variable $X$ taking
the values $x_{1}, x_{2}, \ldots, x_{n}$ with probabilities 
$p(x_{1}), p(x_{2}), \ldots, p(x_{n})$ is defined by, 
\begin{eqnarray*}
H(X) = -\sum_{i=1}^{n}p(x_{i}) \log_{b}\bigl (p(x_{i})\bigr ).
\end{eqnarray*}
\end{Definition}

The base $b$ of the 
logarithm will be normally equal to two; in this case we measure bits of entropy. 
We should mention that the notation $H(X)$ does not merely denote a function of $X$; 
entropy is a function of the probability distribution.
Generally, in sorting algorithms that utilise comparisons for this task, 
entropy quantifies the amount of information gained from the sorting. 
Consider an unsorted array, with all the $n!$ permutations equally likely. 
A comparison gives
$1$ bit of information, thus at least $\log_{2}(n!) = \Omega \bigl (n\log_{2}(n) \bigr )$ 
comparisons
are needed for a complete sort -- see in \cite{bell}. 
This quantity is called information--theoretic lower bound.
The range of entropy is given in the following Lemma.
\begin{Lemma}
\begin{eqnarray*}
0 \leq H(X) \leq \log_{b}(n).
\end{eqnarray*}
\end{Lemma} 
\begin{proof}
Since $\log_{b}p(x_{i}) \leq 0$, the left inequality follows immediately. 
Noting that 
the logarithm is concave function and applying Jensen's 
inequality \cite{jens}, which for
a random variable $X$ and a concave function $f$, states that
\begin{equation*}
f \bigl ( \mathbb{E}(X) \bigr ) \geq \mathbb{E} \bigl (f(X) \bigr),
\end{equation*}
we have
\begin{align*}
\sum_{i=1}^{n} p(x_{i}) \log_{b}\left ( \dfrac{1}{p(x_{i})} \right ) & \leq \log_{b} 
\left (\sum_{i=1}^{n} p(x_{1}) \dfrac{1}{p(x_{i})} \right ) \\
&\quad {} = \log_{b}(n). \qedhere 
\end{align*}
\end{proof}

The next four definitions can be found in \cite{cov}, 
\cite{eli} and \cite{sha}.
\begin{Definition}
The joint entropy $H(X, Y)$ of two discrete random variables $X$ and $Y$ is defined as
\begin{eqnarray*}
H(X, Y) = - \sum_{x \in X} \sum_{y \in Y}p(x, y) \log_{b}p(x, y),
\end{eqnarray*}
where $p(x, y)$ is the probability that $X$ takes the value 
$x$ and $Y$ the value $y$. 
\end{Definition}
\begin{Definition}
The conditional entropy $H(X \vert Y)$
of two discrete random variables $X$ and $Y$ is defined as
\begin{eqnarray*}
H(X \vert Y) = - \sum_{x \in X} \sum_{y \in Y} p(x, y) \log_{b} \dfrac{p(x, y)}{p(y)}.
\end{eqnarray*}
\end{Definition}
\begin{Definition}
The information content of a random variable $X$, 
with probability distribution $\mathbb{P}(X)$ is 
\begin{eqnarray*}
I(X) = -\log_{b}\mathbb{P}(X).
\end{eqnarray*}
\end{Definition}
From definition $3.1.6$, one can easily deduce, that
\begin{eqnarray*}
H(X) = \mathbb{E} \bigl (I(X) \bigr ).
\end{eqnarray*}
In other words, entropy is the expected value of the information. We proceed to 
the definition of mutual information.
\begin{Definition}
The mutual information of two discrete random variables $X$ and $Y$ is defined as:
\begin{eqnarray*}
I( X \wedge Y)= \sum_{x \in X}\sum_{y \in Y}p(x, y)\log_{b} \dfrac{p(x, y)}{p(x)p(y)} 
= H(X) - H(X \vert  Y)
\end{eqnarray*}
and quantifies the amount of information provided about $X$ by $Y$.
\end{Definition}
These definitions will be used in a later part of the thesis.

We have to point out that up until now, our analysis did not take into account the 
added overhead of finding the median at each stage. In the simple case of three 
elements, the overhead is not significant, but for larger samples
this might have adversary effects to the efficiency of Quicksort. 

For the selection of the median, we use  
Hoare's Find algorithm \cite{hoar} or Quickselect. This simple and intuitive algorithm 
searches for an element of a given rank $m$ in an array of $n$ keys. 
As in Quicksort, one partitions the array 
around a randomly chosen pivot, which at the end of the 
partition process is moved to 
its final position, $j$. If $m=j$, the pivot is the sought 
element and the search is completed. 
Otherwise, if $m < j$, Quickselect is recursively 
invoked to the left subarray of $j-1$ keys. 
Conversely, if $m > j$, 
we search in the right subarray for the element of rank $(m-j)$. 

Quickselect is ideal in situations where we want to identify order statistics, 
without the need to do a complete sort. The average number of 
comparisons $\mathbb{E}(C_{n; m})$ required for the
retrieval of the $m$-th order statistic in an array of $n$ keys, 
is given by \cite{knuth 3}
\begin{equation*}
\mathbb{E}(C_{n; m}) = 2 \bigl (n+3+(n+1)H_{n}-(m+2)H_{m}-(n-m+3)H_{n+1-m} \bigr ).
\end{equation*}
For the sample of $(2k+1)$ keys, the rank of the median is $k+1$. 
Therefore,
\begin{eqnarray*}
\mathbb{E}(C_{2k+1; k+1}) = 2 \bigl (2k+4+(2k+2)H_{2k+1}-2(k+3)H_{k+1} \bigr ).
\end{eqnarray*}
The cost for finding the median of a random sample of $3$ keys 
is $8/3$ comparisons at each stage. 
However, as the sample of keys increases, in order to obtain a 
more accurate estimate of the median, the added overhead imposes a bottleneck to 
the efficiency of the algorithm.

The computation of other measures of this variant, such as the expected number of 
exchanges and passes can be performed by solving analogous recurrences, 
as the one for the number of comparisons. 
The average number of passes is recursively given by
\begin{eqnarray*}
\mathbb{E}(P_{{n}\{2k+1\}})= 1 + \dfrac{2}{\dbinom{n}{2k+1}}\sum_{j=1}^{n}\dbinom{j-1}{k}
\dbinom{n-j}{k}\mathbb{E}(P_{{j-1}\{2k+1\}}),
\end{eqnarray*}
where the ``toll function'' is now one recursive call to the algorithm, 
after the chosen pivot is 
the median of $(2k+1)$ keys, which yields two subarrays. 
This recurrence can be turned 
to a differential equation with solution 
\begin{eqnarray*}
\mathbb{E}(P_{{n}\{2k+1\}}) = \dfrac{n+1}{2(H_{2k+2}-H_{k+1})} - 1,
\end{eqnarray*}
noting that for $k=0$, the average number of passes is $n$. In the scheme, 
where arrays containing $m$ or fewer keys are sorted 
by insertion sort, the average costs are reduced. 
We refer to \cite{Hen} where this variant is analysed.

\section{Remedian Quicksort}

In the previous section, we analysed the modification of Quicksort, where the
pivot is selected as the median of a sample of $(2k+1)$ elements. This variant offers
better protection against the occurrence of trivial partitions. 
However, there are some cases,
where the running time of this partitioning scheme can go quadratic. Consider
the application of `median of $3$' Quicksort in an array of $n$ 
numbers, where the keys at 
positions $1$, $n$ and $\left \lfloor \frac{n+1}{2} \right \rfloor$ 
are selected as elements of the sample.
In case that two keys of this sample happen to be the smallest 
(or greatest) elements of the array, 
the chosen pivot will be $2$ (or $n-1$), leading to trivial partitioning, 
making Quicksort everything else, except quick! 
In \cite{sedgewick}, a permutation of the array $\{1, \ldots, 15 \}$ is given, 
which leads Quicksort to worst-case performance and in \cite{erk}, an algorithm is 
presented which forms the worst-case permutation. 

In order to remedy this, a bigger sample of $(2k+1)^{\beta}$ 
keys is randomly selected, its remedian
is found and used as partitioning element of the array to be sorted.
The remedian of the sample is recursively defined to be the median of $(2k+1)$ remedians 
of $(2k+1)^{\beta - 1}$ 
elements, where the remedian of $(2k+1)$ elements is the 
median, and is shown to be a robust estimator
of the median \cite{Bas}, \cite{tuk}.

A particular case of the remedian Quicksort widely used in sorting applications
is Tukey's `ninther', 
where the selected pivot is the median of three medians of three samples, each
containing three elements \cite{hc}, \cite{tuk}. 
In practical implementations, this variant exhibits faster running time \cite{ben} 
with little added overhead. Specifically, the 
computation of the remedian of $9$ elements takes on average $4 \times \frac{8}{3}$ 
comparisons at each call -- four times more than finding 
the median of $3$ randomly chosen keys. 

Let $C_{{n}\{(2k+1)^{\beta}\}}$ denote the number of comparisons required for the 
complete sorting of an array of $n$ distinct keys, where the chosen pivot at each call 
is the remedian of a random sample of $(2k+1)^{\beta}$ elements. The recurrence 
relation is much more complicated than the previous ones and the probability 
$p^{\beta}_{j}$ that the remedian of $(2k+1)^{\beta}$ elements 
is the $(j+1)$-th element is
\begin{align*}
p^{\beta}_{j}=(2k+1)!\sum_{\alpha_{1} + \ldots + \alpha_{2k+1}=j} 
p( \alpha_{1}, \ldots, \alpha_{2k+1}) \frac{\binom{j}{\alpha_{1}, \ldots, \alpha_{2k+1}}
\binom{(2k+1)^{\beta}-j-1}{(2k+1)^{\beta-1}-1-\alpha_{1}, \ldots, (2k+1)^{\beta-1}-
\alpha_{2k+1}}}{\binom{(2k+1)^{\beta}}{(2k+1)^{\beta - 1}, \ldots, (2k+1)^{\beta -1}}},
\end{align*}
where
\begin{equation*}
\dbinom{j}{\alpha_{1}, \ldots, \alpha_{2k+1}} = \dfrac{j!}{\alpha_{1}! \ldots \alpha_{2k+1}!} 
\end{equation*}
is the multinomial coefficient and $p( \alpha_{1}, \ldots, \alpha_{2k+1})$ is defined by
\begin{eqnarray*}
p( \alpha_{1}, \ldots, \alpha_{2k+1}) = p^{(\beta -1)}_{\alpha_{1}}(p^{(\beta -1)}_{0}+ \ldots + 
p^{(\beta -1)}_{\alpha_{2}-1}) \cdot \ldots \cdot (p^{(\beta -1)}_{\alpha_{2k+1}}+ \ldots + 
p^{(\beta -1)}_{(2k+1)^{\beta}-1}),
\end{eqnarray*}
with $p^{(0)}=1$.
See as well \cite{hc}, where the `splitting' probabilities of $3^{d}$ remedian are presented. 

Bentley's and McIlroy's experiments on the `remedian of $3^{2}$' Quicksort \cite{ben} showed
that the average number of key comparisons is 
$1.094n\log_{2}(n) - 0.74n$, very close to the information--theoretic 
lower bound of 
$n\log_{2}(n)-1.44n$. This Quicksort utilises the `ninther' 
partitioning for large arrays, 
then the `median of $3$' is used and as the algorithm proceeds, 
the partitioning strategy 
changes to the standard uniform pivot selection. 
The paper written by Durand \cite{dur} confirmed 
these experimental results, where the average number of comparisons is being given by
\begin{eqnarray*}
\mathbb {E} (C_{{n}\{3^{2}\}}) = 1.5697n \log_{e}(n)-1.0363n+1.5697\log_{e}(n)-7.3484 
+ O \left (\dfrac{1}{n} \right ).
\end{eqnarray*}

From a theoretical point of view, this variant yields 
savings on the expected time needed for the sorting, 
with little additional cost of computing the remedian. 
However, the recurrences are quite involved, 
as the remedian has an inherent recursive definition. 
A different approach would be to randomly 
choose a larger sample and use its elements as pivots 
through complete sorting. An obvious 
advantage of this method, is that the cost of 
computing the median or remedian of a 
sample at each call of the algorithm is 
avoided and instead all the pivots for the 
subsequent calls belong in one sample.

\section{Samplesort}

Having examined the strategy of selecting the median of a sample as pivot 
and the more complicated `remedian' variant, 
we proceed to the analysis of Samplesort algorithm
invented by Frazer and McKellar \cite{Fra}.
Instead to randomly select a sample of keys at each stage, 
computing the median and using it 
as pivot, a larger sample of $2^{k}-1$ keys 
is selected and extracted out of the array. 
It is sorted and its keys are being used 
as partitioning elements for the sorting of the array.

First the median of the sample is used as pivot, 
then the lower quartile
to the lower subarray and the upper one to the subarray 
of the elements that are greater than the median. 
When the sample is exhausted, the resulting subarrays 
can be recursively sorted
by the same procedure or by standard Quicksort. 

For convenience, let $2t+1 = 2^{k}-1$ and 
let us denote the total number of comparisons of Samplesort applied to $n$ keys by 
$C^{\{2^{k}-1\}}_{n}$. 
Then, $C^{\{2^{k}-1\}}_{n}$ is equal to the number of comparisons 
to sort the sample of $2^{k}-1$
elements, plus
the number of comparisons to insert its elements to the 
remainder of the array, plus the number of
comparisons required to sort the resulting $2^{k}$ 
subarrays, using ordinary Quicksort \cite{Fra}.
For the sorting of the sample, we use Quicksort 
and the average number of comparisons is
\begin{equation*}
2(2t+2)H_{2t+1}-4(2t+1). \, \, \tag{3.6}
\end{equation*}

Assume that the sorted sample is $x_{1} < x_{2} < \ldots < x_{t+1} < \ldots < x_{2t+1}$. 
First, the median $x_{t+1}$ is inserted to its final position 
in the array of $n-2t$ keys, by pairwise
comparisons of the $n-2t-1$ keys to $x_{t+1}$. 
The cost of partitioning is
$n-2t-1$ comparisons. Then, the first quartile 
is inserted
to the subarray of the elements less than the median 
$x_{t+1}$ and the third quartile to the
subarray of the elements greater than $x_{t+1}$. 
The cost of partitioning these two subarrays
is $n-2t-1$ comparisons, since the sum of their 
lengths is $n-2t-1$ keys. The process is continued until
all the elements of the sample are used as 
pivots and this will take $2t+1$
partitioning stages. 
Thus, an approximation to the average number of comparisons for the 
insertion of the sample is \cite{Fra}, \cite{hosam},
\begin{equation*}
(n-2t-1)\log_{2} (2t+1). \, \, \tag{3.7}
\end{equation*}

After all elements have inserted, there are $2(t+1)$ subarrays to be 
sorted by Quicksort. The expected number of comparisons is \cite{Fra}
\begin{equation*}
2(n+1)(H_{n+1}-H_{2(t+1)})-2(n-2t-1). \, \, \tag{3.8}
\end{equation*}
Putting together Eq. (3.6), (3.7) and (3.8), we have that the expected number of 
comparisons of Samplesort is
\begin{align*}
& 2(n+1)(H_{n+1}-H_{2(t+1)})-2(n-2t-1) + (n-2t-1)\log_{2} (2t+1) \\ 
& \qquad {} + 2(2t+2)H_{2t+1}-4(2t+1). \, \, \tag{3.9}
\end{align*}
For large values of $n$, the expected number of 
comparisons is given by the following 
Corollary.
\begin{Corollary}[Frazer and McKellar \cite{Fra}]
The asymptotic expected number of comparisons taken by Samplesort for the sorting of an 
array of $n$ keys, using a randomly chosen sample of $l$ keys, is
\begin{eqnarray*}
\mathbb{E}(C^{\{l\}}_{n})=1.386n \log_{2}(n)-0.386\bigl(n-l \bigr)\log_{2}(l)-2n-0.846l.
\end{eqnarray*}
\end{Corollary}

It is worthwhile to note that the programming of 
Samplesort is simple and straightforward, 
thus making it a suitable candidate to sorting 
applications. It is proven in \cite{Fra}, 
that the procedure is asymptotically optimal, 
i.e. as $n \to \infty$, the expected 
number of comparisons approaches the information--theoretic bound. 
The process of randomly drawing the sample out of the array 
to be samplesorted should be carefully 
selected, to the effect that the elements of 
the sample will produce `balanced' 
partitions. The main feature of Samplesort 
is that partitioning preserves the order 
of the sample, thus its keys are exchanged with 
the ones that they have to, 
so as to the exhaustion of the sample, its 
elements to be spread far apart.

In the direction of choosing a more `centered' sample, one can proceed by 
randomly selecting three samples -- each containing three keys -- and 
computing their medians, at an extra cost of $8$ comparisons. For the sorting of
larger arrays, the
number of samples will obviously be greater. The medians will be used as 
elements of the sample and the remaining elements can
be randomly chosen from the array. It should be noted, 
that this is an initial idea, lacking the mathematical analysis.

Albacea \cite{al} derived a modification of Samplesort. 
This variant starts with $1$ key, 
that is used as pivot for the partition of $2$ keys, 
so to have a sorted array of $3$ keys. 
These keys are used as pivots for the partitioning 
of $4$ keys, so as to obtain a sorted 
array of $7$ keys and so on, until the whole array 
is sorted. It is proven \cite{al} that 
the estimated expected number of key comparisons is
\begin{equation*}
n \lceil \log_{2}(n+1) \rceil - 2^{\lceil \log_{2}(n+1) \rceil} - 
n + \lceil \log_{2}(n+1) \rceil +1.
\end{equation*}
The derivation of the expected number of comparisons remains an open problem.

\chapter{Sorting by multiple pivots}

In this Chapter, different partitioning routines are analysed. These 
schemes utilise many pivots for the partitioning of the 
array and naturally arise as a generalisation of the algorithm. 
The pivots are chosen uniformly at random and the array is partitioned into
more than two subarrays. There is an additional overhead of comparing 
the pivots before the partitioning, which adds very small 
contributions to the running time. The aim of this modification is twofold: 
first to study if the possibility of worst-case scenario of the algorithm 
can be reduced further and secondly, to provide a theoretical basis for the
analysis of the generalisation of the algorithm.

We show that the average case analyses of these variants can be 
fully described by a general recurrence model, which
is transformed to a differential equation, whose solution
provides the expected cost of these variants. Further, we demonstrate that the 
integration constants involved in the solution, can be efficiently 
computed using Vandermonde matrices.

\section{Quicksorting on two pivots}

Along the following lines, we present a variant of Quicksort, 
where $2$ pivots are used for the partitioning of the array.  
Let a random permutation of the keys $\{1, 2, \ldots, n\}$ to be sorted, 
with all the $n!$ permutations
equally likely and let their locations in the array be numbered from left to right
by $\{1, 2, \ldots, n\}$. The keys at locations $1$ and $n$
are chosen as pivots and since all the $n!$ permutations are equally likely
to be the input, then  
all the $\binom{n}{2}$ pairs are equiprobable to be selected as pivots.
At the beginning, the pivots 
are compared each other and are swapped, if they are not in order.
If elements $i < j$ are selected as pivots, the array
is partitioned into three subarrays: one with $(i-1)$ keys smaller than $i$,
a subarray of $(j-i-1)$ keys between two pivots and the part of $(n-j)$ elements
greater than $j$. 

The algorithm then is recursively applied to each of these subarrays. 
The number of comparisons during the first stage is
\begin{align*}
A_{n, 2} & = 1 + \bigl ( (i-1)+2(j-i-1)+2(n-j) \bigr ) \\
& = 2n-i-2,
\end{align*}
for $i=1, \ldots, n-1$, and $j= i+1, \ldots, n$. Note that in the specific partitioning
scheme, each element is compared once to $i$ and elements
greater than $i$ are compared to $j$ as well. 
The average number of comparisons for the partitioning of $n$ distinct keys is
\begin{align*}
\mathbb{E}(A_{n, 2})=\dfrac{1}{\dbinom{n}{2}} \sum_{i=1}^{n-1}\sum_{j=i+1}^{n} 
\biggl ( 2n - i - 2 \biggr ) 
& = \dfrac{2}{n(n-1)}\left(\dfrac{5}{6}n^3 - 2n^2+\dfrac{7}{6}n \right ) \\
& = \dfrac{5n-7}{3}.
\end{align*}

Letting $C_{n, 2}$ denote the number of comparisons of dual pivot Quicksort applied to an 
array of $n$ items, the recurrence for the expected number of comparisons is
\begin{align*}
\mathbb{E}(C_{n, 2}) & =  \dfrac{5n-7}{3} +\frac{2}{n(n-1)} \\
& \quad {}\times
\left ( \sum_{i=1}^{n-1}\sum_{j=i+1}^{n} \mathbb{E}(C_{i-1, 2})
+ \sum_{i=1}^{n-1}\sum_{j=i+1}^{n} \mathbb{E}(C_{j-i-1, 2})
+ \sum_{i=1}^{n-1}\sum_{j=i+1}^{n} \mathbb{E}(C_{n-j, 2}) \right ).
\end{align*}
Note that the three double sums above are equal. 
Therefore, the recurrence becomes
\begin{align*}
\mathbb{E}(C_{n, 2})& = \dfrac{5n-7}{3}+\frac{6}{n(n-1)}\sum_{i=1}^{n-1}(n-i) \mathbb{E}(C_{i-1, 2}). 
\end{align*}
Letting $a_{n} = \mathbb{E}(C_{n, 2})$, we have 
\begin{eqnarray*}
a_{n}=\dfrac{5n-7}{3}+\frac{6}{n(n-1)}\sum_{i=1}^{n-1}(n-i)a_{i-1}, \mbox{$~n \geq 2.~$}
\end{eqnarray*}
It holds that $a_{0} = a_{1} = 0$. 
Multiplying both sides by $\dbinom{n}{2}$, we obtain
\begin{align*}
\dbinom{n}{2} a_{n} & = \dbinom{n}{2} \left ( \frac{5n-7}{3} +
\frac{6}{n(n-1)}\sum_{i=1}^{n-1}(n-i)a_{i-1} \right ) \\
& = \frac{n(n-1)(5n-7)}{6} + 3\sum_{i=1}^{n-1}(n-i)a_{i-1}.
\end{align*}
This recurrence will be solved by the difference method. We have
\begin{align*}
\Delta F(n)& :=F(n+1)-F(n) \mbox{\quad and~for~higher~orders~} \\
\Delta^{k}F(n)& :=\Delta^{k-1}F(n+1)-\Delta^{k-1}F(n).
\end{align*}
Applying the difference operator 
\begin{align*}
& \Delta \dbinom{n}{2} a_{n} = \dbinom{n+1}{2}a_{n+1} - \dbinom{n}{2}a_{n} = 
\dfrac{5n^{2}-3n}{2} + 3\sum_{i=0}^{n-1}a_{i}  \\
& \Delta^{2} \dbinom{n}{2} a_{n} = \Delta \dbinom{n+1}{2}a_{n+1} - 
\Delta \dbinom{n}{2}a_{n} = 5n+1+3a_{n}.
\end{align*}
By definition, 
\begin{align*}
\Delta^{2} \dbinom{n}{2} a_{n} & = \Delta \dbinom{n+1}{2}a_{n+1}- \Delta \dbinom{n}{2}a_{n} \\ 
& = \dbinom{n+2}{2}a_{n+2}-2\dbinom{n+1}{2}a_{n+1} + \dbinom{n}{2}a_{n}
\end{align*}
and the recurrence becomes
\begin{align*}
&(n+1)(n+2)a_{n+2}-2n(n+1)a_{n+1}+n(n-1)a_{n}=2(5n+1+3a_{n}) \\
& \quad {}\Longrightarrow (n+1) \bigl ((n+2)a_{n+2}-(n-2)a_{n+1} \bigr )-
(n+2) \bigl ((n+1)a_{n+1}-(n-3)a_{n} \bigr ) \\
& \qquad {}= 2(5n+1).
\end{align*}

Dividing by $(n+1)(n+2)$, we obtain the telescoping recurrence
\begin{eqnarray*}
\frac{(n+2)a_{n+2}-(n-2)a_{n+1}}{n+2}=\frac{(n+1)a_{n+1}-(n-3)a_{n}}{n+1}+
\frac{2(5n+1)}{(n+1)(n+2)},
\end{eqnarray*}
which yields
\begin{eqnarray*}
\frac{(n+2)a_{n+2}-(n-2)a_{n+1}}{n+2}= 2\sum_{j=0}^{n}\frac{5j+1}{(j+1)(j+2)}=
\frac{18}{n+2}+10H_{n+1}-18.
\end{eqnarray*}
The recurrence is equivalent to
\begin{eqnarray*}
na_{n}-(n-4)a_{n-1}=18+10nH_{n-1}-18n.
\end{eqnarray*}
Multiplying by $\dfrac{(n-1)(n-2)(n-3)}{24}$,
this recurrence is transformed to a telescoping one \cite{sedgewick},
\begin{eqnarray*}
\dbinom{n}{4}a_{n}=\dbinom{n-1}{4}a_{n-1}+\frac{18(n-1)(n-2)(n-3)}{24}+
10 \dbinom{n}{4}H_{n-1}-18\dbinom{n}{4}.
\end{eqnarray*}
Unwinding, we have
\begin{align*}
\dbinom{n}{4}a_{n}& =18 \sum_{j=1}^{n}\frac{(j-1)(j-2)(j-3)}{24}+10 
\sum_{j=1}^{n}\dbinom{j}{4}H_{j-1} \\
& \quad{}-18\sum_{j=1}^{n}\dbinom{j}{4}. \, \, \tag{4.1}
\end{align*}

The second sum of Eq. (4.1) is  
\begin{align*}
\sum_{j=1}^{n} \dbinom{j}{4}H_{j-1} & = \sum_{j=1}^{n} \Biggl ( \dbinom{j}{4} \left 
( H_{j} - \dfrac{1}{j} \right ) \Biggr ) = \sum_{j=1}^{n} \dbinom{j}{4}H_{j} -
 \sum_{j=1}^{n} \dbinom{j}{4} \dfrac{1}{j} \\
& = \dbinom{n+1}{5} \biggl (H_{n+1} - \dfrac{1}{5} \biggr ) - 
\dfrac{1}{24} \sum_{j=1}^{n}(j-1)(j-2)(j-3) \\
& = \dbinom{n+1}{5} \biggl (H_{n+1} - \dfrac{1}{5} \biggr ) - 
\dbinom{n}{4} \dfrac{1}{4},
\end{align*}
thus
\begin{eqnarray*}
&& \dbinom{n}{4}a_{n} = \dfrac{9}{2} \dbinom{n}{4} +10 \Biggl ( \dbinom{n+1}{5} 
\left (H_{n+1} - \dfrac{1}{5} \right ) - \dfrac{1}{4} \dbinom{n}{4} \Biggr ) 
-18 \dbinom{n+1}{5}. \\
&& \Longrightarrow a_{n} = \dfrac{9}{2} + 10 \Biggl ( \dfrac{n+1}{5} \left 
(H_{n+1} - \dfrac{1}{5} \right ) - \dfrac{1}{4} \Biggr ) - \dfrac{18(n+1)}{5}.
\end{eqnarray*}
The expected number of comparisons, when two pivots are chosen is
\begin{align*}
a_{n} = 2(n+1)H_{n}- 4n.  
\end{align*}
This is exactly the same as the expected number of comparisons 
for ordinary Quicksort. 

Next, the expected number of key exchanges will be computed. The exchanges 
during partitioning are performed as follows. 
Set two pointers $l \leftarrow 2$, $u \leftarrow n-1$ 
and store temporarily the
pivots in another array of size two, so the cells at locations $1$ and $n$ are empty,
leaving two ``holes''. 

After the pivots are sorted by one comparison,
the key at position $2$ is compared to the first pivot (i.e. the smaller
of the two pivots); if it is less
than the pivot, it is put into the left hole, which now is moved one
position to the right and $l$ is increased by one. If it (i.e. the key at
position $2$) is greater than the first pivot, it is compared to the second
pivot (i.e. the greater of the two pivots). If it is less
than the second pivot, $l$ is increased by one, otherwise $l$ stops.

Now, the $u$ pointer starts its downward scan. If an examined key is
greater than both
pivots, it is put into the right hole, which is moved one position to the 
left and $u$ is decreased by one. If a key is less than the second pivot 
and greater than the first, then
$u$ is decreased by one. In case that a key is less than the first pivot, 
then $u$ stops its scan and the key that is greater to the second pivot,
where $l$ has stopped is put to the right hole, which is moved one position
to the left and the key where $u$ has stopped is put to the left hole, which
is moved one position to the right. Then, $l$ is increased by one,
$u$ is decreased by one and $l$ resumes its scan.

When pointers are crossed, the first pivot
is put to the left hole, the second pivot is put to the right hole and partition is
completed, since keys less than the first pivot are on its left, keys between two pivots
on the middle and keys greater than the second pivot are on its right subarray. 
Note that the auxiliary space required for the storing of pivots
is $O(1)$, since at the end of the partition routine, the pivots are moved back to the array
and two other new pivots can be stored, as the algorithm operates on a given
subarray.
We refer to \cite{sedgewick} for further details
of this scheme.

The average number of swaps during the first stage is
\begin{eqnarray*}
\dfrac{1}{\dbinom{n}{2}}\sum_{i=1}^{n-1}\sum_{j=i+1}^{n}(i-1)
= \dfrac{1}{\dbinom{n}{2}}\sum_{i=1}^{n-1}(n-i)(i-1)
= \dfrac{1}{\dbinom{n}{2}} \left ( \sum_{i=1}^{n-1}(n-i)i -
\sum_{i=1}^{n-1}(n-i) \right ),
\end{eqnarray*}
since $(i-1)$ keys are less than pivot $i$. Thus, the average contribution is 
\begin{eqnarray*}
\dfrac{2}{n(n-1)} \left ( \dfrac{n^{3}-n}{6} - \dfrac{n(n-1)}{2} \right ) = \dfrac{n-2}{3}.
\end{eqnarray*}
For the $(n-j)$ keys greater than $j$ the average value is the same, because the sums are equal. 
Adding the two final ``exchanges'' to get the pivots in place, 
the average number of exchanges during the partitioning routine is 
$\left (\dfrac{2(n+1)}{3} \right )$. Letting $S_{n, 2}$ denote the number of 
exchanges of dual pivot Quicksort, 
the recurrence for the mean number of exchanges in course of the algorithm is
\begin{align*}
\mathbb{E}(S_{n, 2})& =\dfrac{2(n+1)}{3}+ \frac{2}{n(n-1)} \\
&\quad {} \times \left ( \sum_{i=1}^{n-1}\sum_{j=i+1}^{n}\mathbb{E}(S_{i-1, 2})+
\sum_{i=1}^{n-1}\sum_{j=i+1}^{n} \mathbb{E}(S_{j-i-1, 2}) + \sum_{i=1}^{n-1}
\sum_{j=i+1}^{n} \mathbb{E}(S_{n-j, 2}) \right ) \\
& = \dfrac{2(n+1)}{3}+ \frac{6}{n(n-1)} \sum_{i=1}^{n-1}(n-i)
\mathbb{E}(S_{i-1, 2}). 
\end{align*}
Multiplying both sides by $\dbinom{n}{2}$,
\begin{equation*}
\dbinom{n}{2}\mathbb{E}(S_{n, 2}) = \dfrac{n(n-1)(n+1)}{3}+3\sum_{i=1}^{n-1}(n-i)
\mathbb{E}(S_{i-1, 2}). 
\end{equation*}

This recurrence is solved in \cite{sedgewick}: here we present a solution using
generating functions. Letting $b_{n} = \mathbb{E}(S_{n, 2})$ and 
$g(z)=\sum_{n=0}^{\infty}b_{n}z^{n}$ be the generating function of the average number of exchanges,
the recurrence is transformed to the following differential equation:
\begin{align*}
\dfrac{z^{2}}{2} \dfrac{\mathrm{d}^{2} g(z)}{\mathrm{d}z^{2}} 
& =\dfrac{z^{2}}{3} \dfrac{\mathrm{d}^{3}}{\mathrm{d}z^{3}} 
\left (\displaystyle \sum_{n=0}^{\infty}z^{n+1} \right ) +
3 \sum_{n=1}^{\infty}\sum_{i=1}^{n}(n-i)b_{i-1}z^{n}.
\end{align*}
The double sum is equal to
\begin{align*}
\sum_{n=1}^{\infty}\sum_{i=1}^{n}(n-i)b_{i-1}z^{n} & = b_{0}z^{2} + ( 2b_{0}+b_{1} )z^{3} + 
( 3b_{0} + 2b_{1} +b_{2} )z^{4} + \ldots \\
& = z^{2}(b_{0}+b_{1}z+b_{2}z^{2}+ \ldots ) + 2z^{3}(b_{0}+b_{1}z+b_{2}z^{2}+ \ldots ) + \ldots \\
& = (z^{2} + 2z^{3} + 3z^{4} + \ldots )g(z) \\
& = \biggl ( \sum_{n=0}^{\infty} nz^{n+1} \biggr )g(z)
\end{align*}
and our differential equation becomes
\begin{eqnarray*}
\dfrac{z^{2}}{2}\dfrac{\mathrm{d}^{2} g(z)}{\mathrm{d}z^{2}}=\dfrac{2z^{2}}{(1-z)^{4}} + 
3 g(z) \left ( \dfrac{z}{1-z} \right )^{2}.
\end{eqnarray*}

Changing variables $v=1-z$, we have $f^{(k)}(v)=(-1)^{k}g^{(k)}(1-v)$. Thus,
\begin{eqnarray*}
\dfrac{(1-v)^{2}}{2}\dfrac{\mathrm{d}^{2} f(v)}{\mathrm{d}v^{2}}=
\dfrac{2(1-v)^{2}}{v^{4}} + 3 f(v) \left ( \dfrac{1-v}{v} \right )^{2}.
\end{eqnarray*}
The differential equation can be simplified by multiplying both sides by 
$\left (\dfrac{v}{1-v}\right )^{2}$,
\begin{align*}
\dfrac{v^{2}}{2}\dfrac{\mathrm{d}^{2} f(v)}{\mathrm{d}v^{2}}=\dfrac{2}{v^{2}} + 3 f(v). \, \, \tag{4.2}
\end{align*}
An elementary approach to solving this differential equation, is to assume
that the solution is of the form $x^{m}$ \cite{boyce}. Substituting the ``trial solution''
to Eq. (4.2), the characteristic or indicial polynomial is
\begin{equation*}
\mathcal{P}_{2}(m)= m(m-1)-6,
\end{equation*}
with roots $m_{1}=3$ and $m_{2}=-2$. Thus, the solution to the corresponding homogeneous equation is 
$c_{1}v^{3}+c_{2}v^{-2}$, with $c_{1}, c_{2} \in \mathbf R$.
A particular solution of Eq. (4.2), which can be found e.g. using the method in
\cite{sabuw}, is
\begin{eqnarray*}
-\dfrac{4}{5}\log_{e}(v)v^{-2}.
\end{eqnarray*}
By the initial conditions $f(1) = -f^{\prime}(1) = 0$, the solution is
\begin{eqnarray*}
f(v)= \dfrac{4}{25}v^{3} - \dfrac{20 \log_{e}(v)+ 4}{25v^{2}}.
\end{eqnarray*}

In the next section, we will examine the generalised version of this differential equation.
Reverting to variable $z$ and discarding terms for $n \leq 3$,
we see that, expanding out the fraction term as a series, 
\begin{equation*}
g(z) = \sum_{n=0}^{\infty}\left (\dfrac{4}{5} \biggl ( (n+1)H_{n}-n \biggr ) 
- \dfrac{4}{25} (n+1) \right )z^{n}.
\end{equation*}
Finally, the mean number of swaps of dual pivot Quicksort is
\begin{eqnarray*}
b_{n, 2} = \dfrac{4}{5}(n+1)H_{n} - \dfrac{24n+4}{25},
\end{eqnarray*}
which is nearly $2.4$ times greater than the expected number of exchanges of standard Quicksort.

The recurrence for the number of partitioning stages $P_{n, 2}$ is much simpler;
\begin{eqnarray*}
P_{n, 2} = 1  + P_{i-1, 2} + P_{j-i-1, 2}+P_{n-j, 2}.
\end{eqnarray*}
By the same reasoning, as in the derivation of the 
expected number of exchanges, the solution is 
\begin{align*}
\mathbb{E}(P_{n, 2}) = \dfrac{2}{5}(n+1)- \dfrac{1}{2}.   
\end{align*} 

\subsection{The variance of the number of key comparisons}

It is desirable to compute the variance of the number of key 
comparisons of dual pivot Quicksort, as
this measure provides a grip of the deviation of the 
random number of comparisons from its expected value. 
By the recursive relation, we have
\begin{eqnarray*}
\mathbb{P}(C_{n, 2} = t) = \dfrac{1}{\dbinom{n}{2}}\sum_{i=1}^{n-1}\sum_{j=i+1}^{n}\mathbb{P}(A_{n, 2}+
C_{i-1, 2}+C_{j-i-1, 2}+C_{n-j, 2} = t),
\end{eqnarray*}
noting that the resulting subarrays are independently sorted, the above is
\begin{eqnarray*}
\dfrac{1}{\dbinom{n}{2}}\sum_{i=1}^{n-1}\sum_{j=i+1}^{n}\sum_{l, m}\biggl (\mathbb{P}(C_{i-1, 2} = l) 
\mathbb{P}(C_{j-i-1, 2} = m)\mathbb{P}(C_{n-j, 2} =  t-m-l-2n+i+2) \biggr).
\end{eqnarray*}
Letting $\displaystyle f_{n}(z) = \sum_{t = 0}^{\infty}\mathbb{P}(C_{n, 2} = t)z^{t}$ be the ordinary
probability generating function for the number of comparisons needed to sort $n$ keys, we obtain
\begin{align*}
f_{n}(z) = \dfrac{1}{\dbinom{n}{2}}\sum_{i=1}^{n-1}
\sum_{j=i+1}^{n}z^{2n-i-2}f_{i-1}(z)f_{j-i-1}(z)f_{n-j}(z). \, \, \tag{4.3}
\end{align*}

It holds that $f_{n}(1) = 1$ and $f'_{n}(1) = 2(n+1)H_{n} - 4n$. 
The second order derivative of Eq. (4.3) 
evaluated at $z=1$ is recursively given by
\begin{eqnarray*}
&& f''_{n}(1)~ = \dfrac{2}{n(n-1)} \biggl ( \sum_{i=1}^{n-1}\sum_{j=i+1}^{n}(2n-i-2)^{2}- 
\sum_{i=1}^{n-1}\sum_{j=i+1}^{n}(2n-i-2) \\
&& +~2\sum_{i=1}^{n-1}\sum_{j=i+1}^{n}(2n-i-2)\mathbb{E}(C_{i-1, 2}) + 
2\sum_{i=1}^{n-1}\sum_{j=i+1}^{n}(2n-i-2)\mathbb{E}(C_{j-i-1, 2}) \\
&& +~2\sum_{i=1}^{n-1}\sum_{j=i+1}^{n}(2n-i-2)\mathbb{E}(C_{n-j, 2})+
2\sum_{i=1}^{n-1}\sum_{j=i+1}^{n}\mathbb{E}(C_{i-1, 2})\mathbb{E}(C_{j-i-1, 2}) \\
&& + ~2\sum_{i=1}^{n-1}\sum_{j=i+1}^{n}\mathbb{E}(C_{i-1, 2})\mathbb{E}(C_{n-j, 2}) +
2\sum_{i=1}^{n-1}\sum_{j=i+1}^{n}\mathbb{E}(C_{j-i-1, 2})\mathbb{E}(C_{n-j, 2}) \\
&& +~\sum_{i=1}^{n-1}\sum_{j=i+1}^{n}f''_{i-1}(1) +  
\sum_{i=1}^{n-1}\sum_{j=i+1}^{n}
f''_{j-i-1}(1)+\sum_{i=1}^{n-1}\sum_{j=i+1}^{n}f''_{n-j}(1) \biggr ).
\end{eqnarray*}
The fourth and fifth sum turn out to be equal and by simple manipulation of
indices, the sums involving products of expected values are equal.
The double sum of the product of the mean number of comparisons can be
simplified as follows, using Corollary 2.3.5:
\begin{align*}
& \sum_{i=1}^{n-1}\sum_{j=i+1}^{n}\mathbb{E}(C_{i-1, 2})\mathbb{E}(C_{n-j, 2}) = \sum_{i=1}^{n-1}\left( 
\mathbb{E}(C_{i-1, 2}) \biggl (\sum_{j=0}^{n-i-1}\mathbb{E}(C_{j, 2}) \biggr) \right) \\
& = \sum_{i=1}^{n-1} \Biggl (\biggl( (2iH_{i-1}-4(i-1) \biggr ) 
\biggl (2 \dbinom{n-i+1}{2}H_{n-i}+ \dfrac{n-i-5(n-i)^{2}}{2} \biggr) \Biggr ).
\end{align*}
Further, using $\dbinom{n-i+1}{2}=\dbinom{n-i}{2}+(n-i)$, 
\begin{align*}
\sum_{i=1}^{n-1}i\dbinom{n-i+1}{2}H_{i-1}H_{n-i} &= \sum_{i=1}^{n-1}\biggl((i-1) + 
1 \biggr)\dbinom{n-i+1}{2}H_{i-1}H_{n-i} \\
&= \sum_{i=1}^{n-1}(i-1) \dbinom{n-i}{2}H_{i-1}H_{n-i} +\sum_{i=1}^{n-1}
\dbinom{n-i}{2}H_{i-1}H_{n-i} \\
+ & \sum_{i=1}^{n-1}(i-1)(n-i)H_{i-1}H_{n-i} + \sum_{i=1}^{n-1}(n-i)H_{i-1}H_{n-i}.
\end{align*}
The four sums can be evaluated using Corollary $3$ in \cite{spi}. 

After some computations in \textsc{Maple}, that can be found in Appendix A, 
the recurrence is 
\begin{align*}
f''_{n}(1)& = 2(n+1)(n+2)(H^{2}_{n} - H^{(2)}_{n}) - H_{n} \left (\dfrac{17}{3}n^{2} 
+ \dfrac{47}{3}n + 6 \right ) + \dfrac{209}{36}n^{2} \\ 
& \quad{} + \dfrac{731}{36}n + \dfrac{13}{6} 
+ \dfrac{6}{n(n-1)}\sum_{i=1}^{n-1}(n-i)f''_{i-1}(1).
\end{align*}
Subtracting $\dbinom{n}{2}f''_{n}(1)$ from $\dbinom{n+1}{2}f''_{n+1}(1)$, we have
\begin{align*}
\Delta \dbinom{n}{2}f''_{n}(1)& = 4n(n+1)(n+2)(H^{2}_{n} - H^{(2)}_{n}) - 
\dfrac{nH_{n}}{9}(84n^{2} + 198n + 42) \\
& \quad{} + 3 \displaystyle \sum_{i=1}^{n}f''_{i-1} 
+ \dfrac{n}{9}(79n^{2} + 231n + 14 ),
\end{align*}
using the identity \cite{sedgewick}
\begin{equation*}
H^{2}_{n+1} - H^{(2)}_{n+1} = H^{2}_{n} - H^{(2)}_{n} + \dfrac{2H_{n}}{n+1}.
\end{equation*}

Also, it holds that 
\begin{align*}
\Delta^{2} \dbinom{n}{2}f''_{n}(1) & = 12(n+1)(n+2)(H^{2}_{n} - H^{(2)}_{n}) -
H_{n}(20n^{2} + 32n - 12)  \\
&\qquad{} + 17n^{2} +37n + 3f''_{n}(1).
\end{align*}
The left-hand side of the previous equation is the same as
\begin{eqnarray*}
\dbinom{n+2}{2}f''_{n+2}(1)-2\dbinom{n+1}{2}f''_{n+1}(1) + \dbinom{n}{2}f''_{n}(1)
\end{eqnarray*}
and the recurrence becomes
\begin{align*}
& (n+1)(n+2)f''_{n+2}(1)-2n(n+1)f''_{n+1}(1)+n(n-1)f''_{n}(1) \\
& \quad{} = 2 \Biggl ( 12(n+1)(n+2)(H^{2}_{n} - H^{(2)}_{n}) - 
H_{n}(20n^{2} + 32n - 12) 
+ 17n^{2} +37n + 3f''_{n}(1) \Biggr).
\end{align*}
Dividing by $(n+1)(n+2)$, we obtain the telescoping recurrence
\begin{align*}
& \frac{(n+2)f''_{n+2}(1)-(n-2)f''_{n+1}(1)}{n+2} \\
& \quad {}= \frac{(n+1)f''_{n+1}(1)-(n-3)f''_{n}(1)}{n+1} \\
& \qquad {} + 2 \Biggl( 12(H^{2}_{n} - H^{(2)}_{n}) -
\dfrac{ H_{n}(20n^{2} + 32n - 12)}{(n+1)(n+2)} + 
\dfrac{17n^{2} +37n}{(n+1)(n+2)}  \Biggr ),
\end{align*}
with solution
\begin{align*}
(n+2)f''_{n+2}(1)-(n-2)f''_{n+1}(1) & = (24n^{2} +100n + 104)
(H^{2}_{n+1}-H^{(2)}_{n+1}) \\
& \quad {} -  H_{n+1}(88n^{2} + 292n + 224) + 122n^{2} 
+ 346n + 224, 
\end{align*}
which is equivalent to
\begin{align*}
nf''_{n}(1) - (n-4)f''_{n-1}(1) &= (24n^{2} + 4n)(H^{2}_{n-1}-H^{(2)}_{n-1}) \\
& \quad{} - H_{n-1}(88n^{2}- 60n - 8)+ 122n^{2} - 142n + 20.
\end{align*}

Again as before, multiplying both sides by $\dfrac{(n-1)(n-2)(n-3)}{24}$, 
the recurrence telescopes with solution 
\begin{eqnarray*}
f''_{n}(1) = 4(n+1)^{2}(H^{2}_{n+1} - H^{(2)}_{n+1}) - 
4 H_{n+1}(n+1)(4n+3) + 23n^{2} + 33n + 12.
\end{eqnarray*}

Using the well known fact that 
$$\operatorname {Var}(C_{n, 2}) = f''_{n}(1) + f'_{n}(1) - \bigl (f'_{n}(1) \bigr )^{2},$$  
the variance of the number of key comparisons of dual pivot Quicksort is
\begin{align*}
7n^{2}- 4(n+1)^{2} H^{(2)}_{n} - 2(n+1)H_{n} + 13n. \, \, \tag{4.4} 
\end{align*}
Note that the variance of dual pivot Quicksort is identical with the 
variance of ordinary Quicksort. In the 
next subsection, we provide the theoretical explanation of this fact.

\subsection{Distribution of the number of key comparisons}

Our results have shown that dual pivot Quicksort has the same expected number
of comparisons, and the same variance, as in the case of `one-pivot'
Quicksort. Thus, it is natural to ask if the two random variables
have the same distribution. We now show this, after an argument sketched by Prof.
Colin McDiarmid \cite{inperson}.

Suppose that an array of $n$ distinct keys $x_{1},x_{2}, \ldots, x_{n}$ 
is to be sorted by Quicksort and let, as usual,
$C_{n}$ be the random number of comparisons required for the sorting. 
Obviously, $C_{1}=0$, $C_{2}=1$ and for $n \geq 3$ we pick uniformly
at random an ordered pair of
distinct indices $(I, J)$ in $[n]=\{1, 2, \ldots, n\}$
and we use $x_{I}$ as the first pivot.
Given that $I=i$, the pivot $x_{i}$
partitions the array of $n$ keys to the subarray
of $(i-1)$ keys less than $x_{i}$ and to the subarray of $(n-i)$ 
keys greater than $x_{i}$ by $(n-1)$ comparisons. 

Given that $I=i$, if $x_{J}<x_{i}$, then $x_{J}$ is a uniformly at random chosen pivot
from the subarray of $(i-1)$ elements less than $x_{i}$.
In this case, for $I=i$ and $J=j$, the
subarray of $(i-1)$ keys is partitioned to the subarray of $(j-1)$ keys less
than $x_{j}$ and to the subarray of $(i-j-1)$ keys greater than $x_{j}$
by $(i-2)$ comparisons.
Note that $(i-2)$ keys are compared to both pivots in two partitioning
stages. Therefore, the following recurrence holds:
\begin{align*}
\mathbb{P}(C_{n}=t) & = \dfrac{1}{\dbinom{n}{2}} 
\sum_{j=1}^{n-1}\sum_{i=j+1}^{n} \mathbb{P}\bigl((n-1)+(i-2) 
+C^{(1)}_{j-1}+C^{(2)}_{i-j-1}
+C^{(3)}_{n-i}=t \bigr) \\
& = \dfrac{2}{n(n-1)} 
\sum_{j=1}^{n-1}\sum_{i=j+1}^{n} \mathbb{P}\bigl((n+i-3) 
+C^{(1)}_{j-1}+C^{(2)}_{i-j-1}
+C^{(3)}_{n-i}=t \bigr),
\end{align*}
where $C^{(1)}_{n}$, $C^{(2)}_{n}$ and $C^{(3)}_{n}$ are 
independent copies of $C_{n}$ -- that is, are random variables with the
same distribution as $C_{n}$, independent of it and each other.

If $x_{J}>x_{i}$, then $x_{J}$ is a uniformly at random selected pivot from
the subarray of $(n-i)$ keys greater than $x_{i}$.
Given that $I=i$ and $J=j$, the
subarray of $(n-i)$ keys is partitioned to the subarray of $(j-i-1)$ keys less
than $x_{j}$ and to the subarray of $(n-j)$ keys greater than $x_{j}$ by
$(n-i-1)$ comparisons.
The recurrence relation is
\begin{align*}
\mathbb{P}(C_{n}=t) & = \dfrac{1}{\dbinom{n}{2}} 
\sum_{i=1}^{n-1}\sum_{j=i+1}^{n} \mathbb{P}\bigl((n-1)+(n-i-1) 
+C^{(1)}_{i-1}+C^{(2)}_{j-i-1}
+C^{(3)}_{n-j}=t \bigr) \\
& = \dfrac{2}{n(n-1)} 
\sum_{i=1}^{n-1}\sum_{j=i+1}^{n} \mathbb{P}\bigl((2n-i-2) 
+C^{(1)}_{i-1}+C^{(2)}_{j-i-1}
+C^{(3)}_{n-j}=t \bigr),
\end{align*}
where as in the previous recurrence, $C^{(1)}_{n}$, $C^{(2)}_{n}$ 
and $C^{(3)}_{n}$ are independent copies of $C_{n}$.
Observe that 
\begin{eqnarray*}
\sum_{j=1}^{n-1}\sum_{i=j+1}^{n} (n+i-3)=\sum_{j=1}^{n-1}(n-j)(2n-j-2)
=\sum_{i=1}^{n-1}\sum_{j=i+1}^{n}(2n-i-2),
\end{eqnarray*}
thus for any two pivots selected uniformly at random, the recurrences
are the same.

Recall that for the random number of comparisons $C_{n, 2}$ of 
dual pivot Quicksort, it holds that $C_{1, 2}=0$, $C_{2, 2}=1$ 
and for $n \geq 3$ we choose uniformly at random an ordered
pair of distinct indices $(I, J)$ in $[n]=\{1, 2, \ldots, n \}$.
The pivots $x_{I}$ and $x_{J}$ are sorted by one comparison and
we assume that
its outcome is $x_{I}<x_{J}$. Given that 
$I=i$ and $J=j$, the array is partitioned
to the subarray of $(i-1)$ keys less than $x_{i}$, the subarray of
$(j-i-1)$ keys
between two pivots and the subarray of $(n-j)$ keys greater than $x_{j}$. Since
keys greater than $x_{i}$ are compared with the other
pivot as well, the recurrence for the random number of comparisons is
\begin{align*}
\mathbb{P}(C_{n, 2}=t) & = \dfrac{1}{\dbinom{n}{2}} 
\sum_{i=1}^{n-1}\sum_{j=i+1}^{n} \mathbb{P}\bigl(1+(i-1)+2(j-i-1)+2(n-j) \\
& \qquad{} +C^{(1)}_{i-1, 2}+C^{(2)}_{j-i-1, 2}
+C^{(3)}_{n-j, 2}=t \bigr) \\
& = \dfrac{2}{n(n-1)} 
\sum_{i=1}^{n-1}\sum_{j=i+1}^{n} \mathbb{P}\bigl((2n-i-2) 
+C^{(1)}_{i-1, 2}+C^{(2)}_{j-i-1, 2}
+C^{(3)}_{n-j, 2}=t \bigr),
\end{align*}
where $C^{(1)}_{n, 2}$, $C^{(2)}_{n, 2}$ and $C^{(3)}_{n, 2}$ are independent
copies of $C_{n, 2}$. Note that when $x_{I}>x_{J}$, the recurrence is the same. 
Thus, since dual pivot Quicksort
and ordinary Quicksort satisfy the same recurrence and have the same initial
conditions for $n=1, 2$, we deduce that the random variables $C_{n, 2}$ and 
$C_{n}$ are identically distributed.

\section{Multikey partitioning}

A natural extension of having two pivots would be to have some other number
$k$ of pivots. Here, we study the idea of randomly picking $k$ pivots $i_{1}, i_{2},
\ldots, i_{k}$ and partitioning the array simultaneously according to these.

Again, let a random permutation of the array $\{1, 2, \ldots , n \}$ be given
to be sorted using this variant, with all the $n!$ permutations
equally likely to be the input. The $k$ rightmost keys are chosen
as pivots, are compared to
each other and exchanged, if they are out of order.
The sorting of the
pivots can be efficiently implemented by insertion sort. 
Since all $n!$ permutations of the keys
are equally likely to be the input, this amounts to the fact 
that any $k$-subset of keys
has equal probability to be selected. 

The remaining
$(n-k)$ keys are compared to the pivots and the array is partitioned to
$(k+1)$ subarrays. The partitioning can be performed as follows.
We compare the leftmost key to a randomly chosen pivot; if it is smaller than
this pivot, it is compared with another smaller pivot (if one exists). 
Otherwise it is compared with a larger pivot (to the right) 
and after a series of
comparisons, is inserted to its place between any two pivots, or to 
the left of the smallest pivot or to the right of the greatest pivot. 
We continue in the same fashion, until all keys are examined. 

In \cite{Hennequin}, each
of the $(n-k)$ keys is compared to the pivots by binary search,
so a key is compared first to the median of the sorted array of the pivots. If
it is less, is compared with the first quartile, otherwise is compared with the third
quartile and after a series of comparisons is inserted
to its position. In worst case, it takes
$O\bigl(\log_{2}(k) \bigr)$ comparisons for the 
insertion of a key.
Then, multipivot Quicksort is recursively applied to each of the
resulting segments 
that contains at least $(k+1)$ keys and
arrays with less than $(k+1)$ keys are sorted by insertion sort
in $O(n)$ time. 

This is equivalent to
$(k+1)$-ary search trees, which is a generalisation of binary trees. Indeed,
if $n \geq k+1$, the $k$ pivots are stored in the root node of the tree in
increasing order and the remaining $(n-k)$ keys are placed in the resulting
$(k+1)$ subtrees of the root. In case that $n=0$, the tree is empty and if
$n \leq k$, the tree has a single node, which stores the keys in order.
Under the assumption of uniformity, this is a $(k+1)$-random tree of $n$
nodes -- see the article of Chern {\em et al.} \cite{hc} and Mahmoud's book
\cite{hosam} for the correspondence between trees and variants of Quicksort. 

Let $f(n, k)$ denote the expected cost of the
algorithm applied to an array of $n$ keys. We deliberately 
allow some flexibility in the form of cost; 
a typical example might be the
number of comparisons. 
The expected cost of this variant is recursively given by
\begin{align*}
f(n,k)  & = T(n, k) \\
& \quad{} + \dfrac{1}{\dbinom{n}{k}} \underbrace{\sum_{i'_{1}} \sum_{i'_{2}} 
\ldots \sum_{i'_{k}}}_{i'_{1} <i'_{2} < 
\ldots < i'_{k}} \Bigl ( f(i'_{1}-1, k) + f(i'_{2} - i'_{1} -1, k) 
+ \ldots + f(n - i'_{k}, k) \Bigr),
\end{align*}
where $i'_{1} < i'_{2} < \ldots < i'_{k}$ are the pivots in increasing order,
$T(n, k) = \overline{a}(k)n+\overline{b}(k)$ is the average value of a
``toll function'' $\tau(n, k)$ during the first recursive call and $f(i'_{1}-1, k)$
denotes the average cost for sorting the subarray of $(i'_{1}-1)$ elements
less than $i'_{1}$ by multipivot Quicksort on $k$ pivots.

Though this looks a complex $k$-index summation, the recursion can be
simplified, by noting that the pivots are randomly selected and the sums are
equal,
\begin{align*}
f(n,k) & = T(n, k) \\
& \quad{} + \dfrac{1}{\dbinom{n}{k}} \underbrace{\sum_{i'_{1}} \sum_{i'_{2}} \ldots 
\sum_{i'_{k}}}_{i'_{1} <i'_{2}< \ldots < i'_{k}} \Bigl ( f(i'_{1}-1, k) + 
f(i'_{2} - i'_{1} -1, k ) + \ldots + f(n - i'_{k}, k ) \Bigr) \\
& = T(n, k) +  \dfrac{1}{\dbinom{n}{k}} \sum_{i'_{1}=1}^{n - k+1} 
\sum_{i'_{2}= i'_{1}+1}^{n-k+2} \ldots \sum_{i'_{k} = i'_{k-1} + 1}^{n} 
\Bigl ( f(i'_{1}-1, k) + \ldots + f(n - i'_{k}, k ) \Bigr ) \\
& = T(n, k) + \dfrac{(k+1)!}{n(n-1) \ldots (n-k+1)}
\sum_{i'_{1} = 1}^{n-k+1} \dbinom{n - i_{1}}{k-1}f(i'_{1}-1, k).
\end{align*}
Multiplying both sides by $\dbinom{n}{k}$, the recurrence relation becomes
\begin{align*}
\dbinom{n}{k}f(n,  k) = \dbinom{n}{k}T(n, k) + (k+1)\sum_{i'_{1} = 1}^{n-k+1} 
\dbinom{n - i'_{1}}{k-1}f(i'_{1}-1, k).
\end{align*}
For notational convenience, let $f(n, k) = a_{n}$ and consider the generating function 
$h(x)= \displaystyle \sum_{n=0}^{\infty}a_{n}x^{n}$;
\begin{eqnarray*}
\sum_{n=0}^{\infty}\dbinom{n}{k}a_{n}x^{n} = \sum_{n=0}^{\infty} \dbinom{n}{k}T(n, k)x^{n} 
+ (k+1)\sum_{n=0}^{\infty} \left (\sum_{i'_{1} = 1}^{n} 
\dbinom{n - i'_{1}}{k-1}a_{i'_{1}-1} \right )x^{n}.
\end{eqnarray*}
The recurrence is transformed to a $k$-th order differential equation
\begin{align*}
\dfrac{h^{(k)}(x)x^{k}}{k!} & = \sum_{n=0}^{\infty} \dbinom{n}{k}T(n, k)x^{n} + 
h(x)(k+1)\sum_{n=0}^{\infty}\dbinom{n-1}{k-1}x^{n} \\
& = \sum_{n=0}^{\infty} \dbinom{n}{k}\bigl (\overline{a}(k)n + \overline{b}(k) 
\bigr )x^{n} + (k+1)h(x) \left (\dfrac{x}{1-x} \right)^{k} \\
& =  \dfrac{x^{k}\bigl (\overline{a}(k)(x+k)+ 
\overline{b}(k)(1-x) \bigr )}{(1-x)^{k+2}} + (k+1)h(x)\left (\dfrac{x}{1-x} \right)^{k},
\end{align*}
since it can be easily seen by induction that the $k$-th order derivative of 
\begin{equation*}
\sum_{n=0}^{\infty} \bigl (\overline{a}(k)n + \overline{b}(k) \bigr )x^{n} = 
\dfrac{\overline{a}(k)x + \overline{b}(k)(1-x)}{(1-x)^{2}}
\end{equation*}
is
\begin{equation*}
\dfrac{k!\bigl (\overline{a}(k)(x+k)+ \overline{b}(k)(1-x) \bigr )}{(1-x)^{k+2}}.
\end{equation*}

Multiplying by $\left (\dfrac{x}{1-x} \right)^{-k}$, the differential equation is simplified to
\begin{eqnarray*}
\dfrac{h^{(k)}(x)(1-x)^{k}}{k!} = \dfrac{\overline{a}(k)(x+k)+ \overline{b}(k)(1-x) }
{(1-x)^{2}} + (k+1)h(x).
\end{eqnarray*}
This differential equation is an equidimensional Cauchy--Euler equation, 
as the one encountered in the previous Chapter.
Changing variables $x=1-z$, it is $h(x) = g(1-x)$. Applying the differential
operator $\Theta$, where $\Theta g(z)=zg^{\prime}(z)$, the differential equation becomes 
\begin{eqnarray*}
\bigl ((-1)^{k}\Theta(\Theta -1) \ldots (\Theta -k + 1)-(k+1)! \bigr )g(z) = 
\dfrac{k!\bigl (\overline{a}(k)(1-z+k)+ \overline{b}(k)z \bigr )}{z^{2}}
\end{eqnarray*}
and the indicial polynomial $\mathcal{P}_{k}(\Theta)$ is equal to
\begin{align*}
\mathcal{P}_{k}(\Theta)=(-1)^{k}\Theta^{\underline{k}}-(k+1)!.
\end{align*}
Using the notation from \cite{con}, $\Theta^{\underline{k}}=\Theta(\Theta-1)\ldots(\Theta-k+1)$ 
with $k \geq 0$,
denotes the falling factorial. Again, we need a Lemma regarding the roots of the
indicial polynomial.
\begin{Lemma}
The indicial polynomial $\mathcal{P}_{k}(\Theta)$ has $k$ simple roots
with real parts in the interval $[-2, k+1]$. The real roots are $-2$; $(k+1)$, if 
$k$ is even and the $2 \left \lfloor \frac{k-1}{2} \right \rfloor$ complex roots 
$\alpha_{1}, \ldots, \alpha_{\left \lfloor \frac{k-1}{2} \right \rfloor}$ with their conjugates
$\overline{\alpha}_{1}, \ldots, \overline{\alpha}_{\left \lfloor \frac{k-1}{2} \right \rfloor}.$
\end{Lemma} 
\begin{proof}
Let $\alpha = x+iy$ be a root of the polynomial. It holds
\begin{align*}
& \alpha(\alpha-1)\ldots (\alpha-k+1)=(-1)^{k}(k+1)! \\
& \Longrightarrow \alpha(\alpha-1)\ldots (\alpha-k+1)=(-2)(-3)\ldots \bigl (-(k+1) \bigr ) \\
& \Longrightarrow \frac{\alpha}{-2}\frac{\alpha-1}{-3}\ldots \frac{\alpha-k+1}{-(k+1)} = 1. \,\, \tag{4.5}
\end{align*}
Suppose that $\mathfrak{Re}(\alpha) < -2$. Then
\begin{align*}
\left \vert \frac{\alpha}{-2} \right \vert & = \frac 
{\sqrt{x^{2}+y^{2}}}{2}>\frac{\sqrt{(-2)^{2}+y^{2}}}{2} \geq 1 \\
\left \vert \frac{\alpha-1}{-3} \right \vert & = 
\frac {\sqrt{(x-1)^{2}+y^{2}}}{3}>\frac{\sqrt{(-3)^{2}+y^{2}}}{3} \geq 1 \\
& \setbox0\hbox{=}\mathrel{\makebox[\wd0]{\vdots}}  \\
\left \vert \frac{\alpha-(k-1)}{-(k+1)} \right \vert & = 
\frac {\sqrt{\bigl (x-(k-1) \bigr)^{2}+y^{2}}}{k+1} > 
\frac{\sqrt{\bigl (-(k+1) \bigr )^{2}+y^{2}}}{k+1} \geq 1.
\end{align*}
Considering the moduli in Eq. (4.5), we see that the left-hand side
is a product of numbers which are all greater than $1$, and
so the overall product is greater than $1$, but the right-hand side is equal
to $1$, leading to contradiction. Therefore every root has real part greater
than or equal to $-2$. Further, looking over the same argument, we see that
the only way we can have the real part being equal to $-2$ is if the imaginary
part is equal to zero. 

By the Fundamental Theorem of
Algebra a polynomial of degree $n$ has $n$ complex roots with 
multiplicities. Note that $-2$ is always a simple root since,
\begin{align*}
\mathcal{P}_{k}(-2) & = (-1)^{k}(-2)(-3) \ldots \bigl (-(k+1) \bigr ) - (k+1)! \\
& = (-1)^{2k}(k+1)! - (k+1)! = 0, \\ 
\\
\mathcal{P}^{\prime}_{k}(-2) & = (k+1)!\sum_{j=0}^{k-1} \dfrac{1}{-2-j} 
= -(k+1)!(H_{k+1}-1)  < 0.
\end{align*}
Suppose that $\alpha$ is a repeated root. Since $\mathfrak{Re}(\alpha) \geq -2$, 
we can write $\alpha=(x-2)+iy$ with $x \geq 0$. By the above comments we have that  
\begin{eqnarray*}
\sum_{j=0}^{k-1}\frac{1}{\alpha-j}=0 \\
\Longrightarrow \sum_{j=0}^{k-1}\frac{1}{(x-2-j)+iy}=0 \\
\Longrightarrow \sum_{j=0}^{k-1}\frac{(x-2-j)-iy}{(x-2-j)^{2}+y^{2}}=0
\end{eqnarray*}
In particular, $\mathfrak{Im}(\alpha) = 0$. But that
imaginary part is equal to $\sum_{j=0}^{k-1}\frac{y}{(x-2-j)^{2}+y^{2}}$ which
is clearly only equal to zero if $y=0$ i.e. the root is real. 
We will thus have obtained our contradiction if we can show that there are no
real roots other than $-2$ and (for $k$ even) $k+1$. 

To do this, suppose that we did have a real root $\alpha>-2$. We note first that $\alpha>0$: because
if not, then since $-2$ is also a root, there is a root of $\mathcal{P}^{\prime}$ between $-2$ and $0$
by Rolle's Theorem. But since $\mathcal{P}_{k}^{\prime}(\beta)=0$ implies that $\sum_{j=0}^{k-1}
\frac{1}{\beta-j}=0$ and when $\beta<0$ this number is clearly negative. 
Thus any real root $\alpha$ is positive. It is also $\leq k+1$ as if it were greater than 
$k+1$ we would have $\alpha(\alpha-1)\ldots (\alpha-k+1)>(k+1)!$. Further,
note that $k+1$ is a root if and only if $k$ is even and there cannot be
a root in $\alpha\in [k,k+1)$ as the product $\alpha(\alpha-1)\ldots (\alpha-k+1)$
would be $<(k+1)!$. 

Suppose then that $\alpha\in (j,j+1)$ for some $0\leq j\leq k-1$. Then the product of 
the non-negative numbers in the sequence $\alpha, \alpha-1,\ldots, \alpha-k+1$ is at most  
$\alpha(\alpha-1)\ldots (\alpha-j)\leq (j+1)\ldots 2\cdot 1=(j+1)!$. 
Thus, to get $\alpha$ being a root, we have to have that 
$$(-1)^{k}(\alpha-j-1)\ldots (\alpha-k+1)\geq \frac{(k+1)!}{(j+1)!}= (j+2)(j+3) \ldots (k+1).$$
However the largest in modulus of $\alpha-j-1,\ldots , \alpha-k+1$ is $\alpha-k+1>j+1-k$ and
so their product is less than $(k-1-j)!$. Consequently our inequalities together imply 
$$\frac{(k+1)!}{(j+1)!(k-1-j)!}<1\Longrightarrow (k+1) \dbinom{k}{j+1} <1$$
and this is a contradiction, completing the proof. 
\end{proof}

The differential equation can be written as
\begin{eqnarray*}
\mathcal{S}_{k-1}(\Theta)(\Theta + 2)g(z) = \dfrac{k!\bigl (\overline{a}(k)(1-z+k)+ 
\overline{b}(k)z \bigr )}{z^{2}}.
\end{eqnarray*}
Letting $r_{k}=-2$ and the remaining $(k-1)$ simple roots be
$r_{1}, r_{2}, \ldots, r_{k-1}$, we have
\begin{eqnarray*}
(\Theta - r_{1}) \ldots (\Theta - r_{k-1})(\Theta + 2)g(z) = \dfrac{k!\bigl 
(\overline{a}(k)(1-z+k)+\overline{b}(k)z \bigr )}{z^{2}}.
\end{eqnarray*}
For the solution of our differential equation, let two functions $g_{1}(z)+g_{2}(z)=g(z)$. Then
\begin{eqnarray*}
(\Theta - r_{1}) \ldots (\Theta - r_{k-1})(\Theta + 2) \bigl (g_{1}(z)+g_{2}(z) \bigr ) = 
\dfrac{\overline{a}(k)(k+1)!}{z^{2}} + \dfrac{ (\overline{b}(k)- \overline{a}(k))k!}{z}
\end{eqnarray*}
and by the property of linearity of differential operator
\begin{align*}
(\Theta - r_{1}) \ldots (\Theta - r_{k-1})(\Theta + 2)g_{1}(z)  & = 
\dfrac{\overline{a}(k)(k+1)!}{z^{2}} \\
(\Theta - r_{1}) \ldots (\Theta - r_{k-1})(\Theta + 2)g_{2}(z) & =  
\dfrac{\bigl (\overline{b}(k)-\overline{a}(k) \bigr )k!}{z}.
\end{align*}
In the same manner as in the analysis of `median of $(2k+1)$' 
Quicksort, applying $k$ times the solution, we obtain
\begin{align*}
g_{1}(z) & = \dfrac{\overline{a}(k)(k+1)!}{(-2-r_{1})(-2-r_{2}) 
\ldots (-2-r_{k-1})}\dfrac{\log_{e}(z)}{z^{2}}
+\sum_{i=1}^{k}c_{i}z^{r_{i}} \\
g_{2}(z) & = \dfrac{k!}{(-1-r_{1})(-1-r_{2}) \ldots 1}
\dfrac{(\overline{b}(k)- \overline{a}(k))}{z} 
+ \sum_{i=1}^{k}d_{i}z^{r_{i}},
\end{align*}
where $c_{i}$ and $d_{i}$ are constants of integration.
In order to evaluate $\mathcal{S}_{k-1}(-2)$, note that
\begin{equation*}
\mathcal{S}_{k-1}(-2) = \mathcal{P}_{k}^{\prime}(-2),
\end{equation*}
thus
\begin{equation*}
\mathcal{S}_{k-1}(-2) = -(k+1)!(H_{k+1} - 1).
\end{equation*}
Moreover,
\begin{equation*}
\mathcal{P}_{k}(-1)= -kk!.
\end{equation*}

Combining both solutions, 
\begin{align*}
g(z) = -\dfrac{\overline{a}(k)}{H_{k+1} - 1}\dfrac{\log_{e} (z)}{z^{2}} + \dfrac{1}{k}
\dfrac{(\overline{a}(k)-\overline{b}(k))}{z} + \sum_{i=1}^{k}s_{i}z^{r_{i}}, \, \, \tag{4.6}
\end{align*}
where $s_{i}= c_{i} + d_{i}$. The constants of integration can be found 
solving the following system of equations
\begin{eqnarray*}
g(1)=g^{\prime}(1)= \ldots = g^{(k-1)}(1) = 0.
\end{eqnarray*}
In terms of series;
\begin{align*}
h(x) & = \dfrac{\overline{a}(k)}{H_{k+1} - 1} \sum_{n=0}^{\infty} \bigl ( (n+1)H_{n}-n) \bigr ) x^{n} + 
\sum_{n=0}^{\infty}\sum_{i=1}^{k}s_{i}(-1)^{n}\dbinom{r_{i}}{n}x^{n} \\
& \qquad {}+  \dfrac{\overline{a}(k)-\overline{b}(k)}{k} \sum_{n=0}^{\infty}x^{n}. \, \, \tag{4.7}
\end{align*}
The third sum of Eq. (4.7) adds to the solution a constant negligible contribution. 
Also, the root $(k+1)$, when $k$ is even, contributes a constant and the root $r_{k} = -2$, 
adds $s_{k}(n+1)$, with $s_{k} \in \mathbf {R}.$ Extracting the coefficients,
the expected cost of multipivot Quicksort is
\begin{align*}
a_{n} & = \dfrac{\overline{a}(k)}{H_{k+1} - 1}\bigl ((n+1)H_{n} - n \bigr)+ s_{k}(n+1)   
+ 2\sum_{i=1}^{\lfloor \frac{k-1}{2} \rfloor }(-1)^{n}\mathfrak{Re} 
\left ( s_{i}\dbinom{\alpha_{i}}{n} \right )
 + O(1).
\end{align*}

The asymptotics of the last sum can be found by the well-known Stirling's 
formula, that states \cite{con}
\begin{equation*}
n!\sim \sqrt{2\pi n}\left(\frac{n}{e}\right)^{n}.
\end{equation*}
Expressing the binomial coefficient in terms of $\Gamma$ functions, we have
\begin{eqnarray*}
(-1)^{n}\dbinom{\alpha_{i}}{n} = \dbinom{-\alpha_{i} + n -1}{n} = 
\dfrac{\Gamma(n-\alpha_{i})}{n!\Gamma(-\alpha_{i})}.
\end{eqnarray*}
The relation 6.1.26 in \cite{abr} reads for $x, y \in \mathbf {R},$ 
\begin{equation*}
\vert \Gamma (x + iy) \vert \leq \vert \Gamma(x) \vert,
\end{equation*}
thus
\begin{equation*}
\vert \Gamma (n - \alpha_{i}) \vert \leq \vert \Gamma \bigl (n - \mathfrak{Re}(\alpha_{i}) \bigr ) \vert.
\end{equation*}
Using Stirling's formula,
\begin{align*}
\dfrac{ \Gamma(n - \mathfrak{Re}(\alpha_{i}))}{n!} & \sim \dfrac{\sqrt{ 2\pi \biggl (n - \bigl (\mathfrak{Re}
(\alpha_{i}) + 1 \bigr ) \biggr)}\left ( \dfrac{n - \bigl (\mathfrak{Re}(\alpha_{i}) + 1 \bigr )}{e} \right )^
{n-  \bigl (\mathfrak{Re}(\alpha_{i}) +1 \bigr )}}{\sqrt{ 2 \pi n} \biggl ( \dfrac{n}{e} \biggr )^{n}} \\
& \sim \cfrac{ \left ( \dfrac{n -\bigl ( \mathfrak{Re}(\alpha_{i})+1 \bigr )}{e} \right )^{n-
\bigl ( \mathfrak{Re}(\alpha_{i}) +1\bigr )}}{\biggl ( \dfrac{n}{e} \biggr )^{n}} \\
& \hspace{7 mm} = \left ( \dfrac{n-\bigl ( \mathfrak{Re}(\alpha_{i})+1 \bigr )}
{e}  \right )^{- \bigl ( \mathfrak{Re}(\alpha_{i}) +1\bigr )} 
\biggl (1 -\dfrac{\mathfrak{Re}(\alpha_{i}) + 1}{n} \biggr )^{n} \\
& \sim  Cn^{-  ( \mathfrak{Re}(\alpha_{i}) +1)},
\end{align*}
where $C=e^{-(\mathfrak{Re} (\alpha_{i})+1 )}$ is an unimportant constant.

Therefore the term is bounded by
\begin{eqnarray*}
2(-1)^{n}\mathfrak{Re} \biggl (s_{i} \dbinom{\alpha_{i}}{n} \biggr ) = O(n^{- 
( \mathfrak{Re}(\alpha_{i}) +1 )})
\end{eqnarray*}
and asymptotically, the expected cost is 
\begin{eqnarray*}
\dfrac{\overline{a}(k)}{H_{k+1}-1} n \log_{e} (n) + \left ( s_{k} + 
\dfrac{\overline{a}(k)}{H_{k+1}-1}(\gamma -1) \right )n + o(n),
\end{eqnarray*}
since all the other roots have real parts greater than $-2$.

Knowing the coefficients, any mean cost of the generalisation of 
the algorithm can be directly computed, 
using this solution, which assumes a simple form. These coefficients are 
related to the number of pivots used during the partitioning scheme. 
In \cite{Hennequin}, the average number of comparisons 
of the first stage is given by:
\begin{equation*}
\overline{a}(s)n + O(1),
\end{equation*}
where $s$ denotes the number of partitions, when $s-1$ pivots are used or 
equivalently the maximum number of descendants of a node of an $s$--ary tree. 
The coefficient $\overline{a}(s)n$ is equal to
\begin{eqnarray*}
\overline{a}(s) = \lceil \log_{2}(s) \rceil + \dfrac{s-2^{\lceil \log_{2}(s) \rceil }}{s}.
\end{eqnarray*}
Thus, the average number of comparisons of Quicksort on $k$ pivots is
\begin{eqnarray*}
\left (\dfrac{\lceil \log_{2}(k+1) \rceil + 1 - \frac{2^{\lceil \log_{2}(k+1) \rceil }}
{k+1}}{H_{k+1} - 1} \right )(n+1)H_{n} +  O(n).
\end{eqnarray*}

\subsection{Derivation of integration constants using Vandermonde matrices}

The constants of integration can be found using 
Vandermonde matrices. Differentiating
$m$ times Eq. (4.6),
\begin{align*}
g^{(m)}(z) & =\dfrac{\overline{a}(k)}{H_{k+1}-1}\dfrac{(-1)^{m+1}m!\bigl 
((m+1)\log_{e}(z)-((m+1)H_{m}-m) \bigr)}{z^{m+2}} \\
& \qquad{} +(-1)^{m}m!\dfrac{\bigl (\overline{a}(k)-\overline{b}(k) \bigr)}{kz^{m+1}}+
\sum_{i=1}^{k}s_{i}r_{i}^{\underline{m}}z^{r_{i}-m} \\
& = \dfrac{(-1)^{m}m!}{z^{m+1}} \biggl 
(-\dfrac{\overline{a}(k)\bigl ((m+1)\log_{e}(z)-((m+1)H_{m}-m)
 \bigr)}{z(H_{k+1}-1)} \\
& \qquad{} + \dfrac{\bigl (\overline{a}(k)-\overline{b}(k) \bigr)}
{k} \biggr )+\sum_{i=1}^{k}s_{i}r_{i}^{\underline{m}}z^{r_{i}-m}.
\end{align*}
The result can be easily proven by induction or by Leibniz's product rule. 
Using the initial conditions, namely that
$g$ and its first $(k-1)$ derivatives are $0$ when evaluated at $z=1$, 
we obtain
\begin{equation*}
\sum_{i=1}^{k}s_{i}r_{i}^{\underline{m}}=(-1)^{m+1}m! \biggl 
(\dfrac{\overline{a}(k)\bigl ((m+1)H_{m}-m)
\bigr)}{H_{k+1}-1} + \dfrac{\overline{a}(k)-\overline{b}(k)}
{k} \biggr ),
\end{equation*}
for $m=0, 1, \ldots, (k-1)$. In matrix form, the linear system is
\begin{center}
\begin{eqnarray*}
&& \begin{bmatrix}
1 & 1 & \ldots & 1 \\
r_{1} & r_{2} & \ldots & -2  \\
\vdots & \vdots & \ddots & \vdots  \\
r^{\underline{k-1}}_{1} & r^{\underline{k-1}}_{2} & \ldots & (-2)^{\underline{k-1}} \\
\end{bmatrix}
\begin{bmatrix}
s_{1} \\
s_{2} \\
\vdots \\
s_{k}
\end{bmatrix} = \\
&& \begin{bmatrix}
& -\dfrac{1}{k}\bigl (\overline{a}(k)-\overline{b}(k) \bigr ) \\
& \dfrac{\overline{a}(k)}{H_{k+1} - 1} + \dfrac{1}{k}\bigl (\overline{a}(k)-\overline{b}(k) \bigr ) \\
& \vdots \\
& (-1)^{k}(k-1)!\biggl 
(\dfrac{\overline{a}(k)\bigl (kH_{k-1}-(k-1)
\bigr)}{H_{k+1}-1} + \dfrac{\overline{a}(k)-\overline{b}(k)}
{k} \biggr )
\end{bmatrix}
\end{eqnarray*}
\end{center}

Here we use the well-known identity
$x^{n}=\sum_{k=0}^{n} \genfrac\{\}{0pt}{}{n}{k}x^{\underline{k}}$ \cite{abr}, where
$\genfrac\{\}{0pt}{}{n}{k}$ are the Stirling numbers of the second kind, to transform
the coefficient matrix into a Vandermonde matrix.  
The determinant of this Vandermonde matrix is equal to
\begin{equation*}
\prod_{1 \leq i < j \leq n} (r_{j}-r_{i}) \neq 0,
\end{equation*}
as the roots are all simple. Considering the expected number of passes
of multipivot Quicksort, it holds that $a(k) = 0$ 
and $b(k) = 1$, for $k=1, 2, \ldots$ ~. 
The system is,
\begin{eqnarray*}
&& \begin{bmatrix}
1 & 1 & \ldots & 1 \\
r_{1} & r_{2} & \ldots & -2  \\
\vdots & \vdots & \ddots & \vdots  \\
r^{\underline{k-1}}_{1} & r^{\underline{k-1}}_{2} & \ldots & (-2)^{\underline{k-1}} \\
\end{bmatrix}
\begin{bmatrix}
s_{1} \\
s_{2} \\
\vdots \\
s_{k}
\end{bmatrix} =
\begin{bmatrix}
& \dfrac{1}{k} \\
& -\dfrac{1}{k} \\
& \vdots \\
& (-1)^{k-1} \dfrac{(k-1)!}{k} 
\end{bmatrix}
\end{eqnarray*}
Turning the coefficient matrix into a Vandermonde 
one and using the identity $\sum_{j=0}^{n}(-1)^{j}j! \genfrac\{\}{0pt}{}{n}{j}=(-1)^{n}$, 
(see subsection 24.1.4 in \cite{abr}),
we obtain 
\begin{eqnarray*}
\centering
&& \begin{bmatrix}
1 & 1 & \ldots & 1 \\
r_{1} & r_{2} & \ldots & -2  \\
\vdots & \vdots & \ddots & \vdots  \\
r^{k-1}_{1} & r^{k-1}_{2} & \ldots & (-2)^{k-1} \\
\end{bmatrix}
\begin{bmatrix}
s_{1} \\
s_{2} \\
\vdots \\
s_{k}
\end{bmatrix} = 
\begin{bmatrix}
& \dfrac{1}{k} \\
& -\dfrac{1}{k} \\
& \vdots \\
& (-1)^{k-1} \dfrac{1}{k} 
\end{bmatrix}
\end{eqnarray*}

In \cite{Hou} and \cite{turn} the inverse of Vandermonde matrix is given, in terms 
of product of an upper and lower triangular matrices. Letting ${\bf A}^{-1}$ 
denote the inverse, it is equal to
\begin{eqnarray*}
{\bf A}^{-1} = 
\begin{bmatrix}
1 & \frac{1}{r_{1}-r_{2}} & \frac{1}{(r_{1}-r_{2})(r_{1}-r_{3})} & \ldots \\
0 &  \frac{1}{r_{2}-r_{1}} & \frac{1}{(r_{2}-r_{1})(r_{2}-r_{3})} & \ldots \\
0 & 0 & \frac{1}{(r_{3}-r_{1})(r_{3}-r_{2})} & \ldots \\
0 & 0 & 0 & \ldots \\
\vdots & \vdots & \vdots & \ldots  
\end{bmatrix}
\begin{bmatrix}
1 & 0 & 0 & \ldots  \\
-r_{1} & 1 & 0 & \ldots  \\
r_{1}r_{2} & -(r_{1}+r_{2}) & 1 & \ldots \\
-r_{1}r_{2}r_{3} & r_{1}r_{2}+r_{1}r_{3}+r_{2}r_{3} & -(r_{1}+r_{2}+r_{3}) & \ldots \\
\vdots & \vdots & \vdots & \ldots  
\end{bmatrix}
\end{eqnarray*}
It is clear that the lower triangular matrix, post--multiplied by
the vector \newline
$\bigl (1/k,-1/k,\ldots (-1)^{k-1}/k \bigr )^{\bf T}$,
will give us
$$\biggl(1/k,-(r_{1}+1)/k,(r_{1}+1)(r_{2}+1)/k, -\prod_{i=1}^{3}(r_{i}+1)/k,
\ldots , 
(-1)^{k-1}\prod_{i=1}^{k-1}(r_{i}+1)/k \biggr )^{\bf T}.$$
Thus, the solution is
\begin{eqnarray*}
s_{i} =(-1)^{k-1}\dfrac{\displaystyle \prod_{\substack{j \neq i \\ 1 \leq j \leq  k}}
(r_{j}+1)}
{k \displaystyle \prod_{\substack{j \neq i \\ 1 \leq j \leq  k} }(r_{i} - r_{j})}
\end{eqnarray*}
and the expected number of partitioning
stages of multipivot Quicksort on $k$ pivots is
\begin{eqnarray*}
(-1)^{k-1}\dfrac{\displaystyle \prod_{j=1}^{k-1}(r_{j}+1)}{k 
\displaystyle \prod_{j=1}^{k-1}(-2 - r_{j})}(n+1) 
+ o(n).
\end{eqnarray*}

Note that 
\begin{equation*}
\dfrac{\displaystyle \prod_{j=1}^{k-1}(r_{j}+1)}{
\displaystyle \prod_{j=1}^{k-1}(-2 - r_{j})} = (-1)^{k-1} \dfrac{kk!}{(k+1)!(H_{k+1}-1)},
\end{equation*}
therefore the mean number of partitioning stages is
\begin{equation*}
\dfrac{n+1}{(k+1)(H_{k+1}-1)} + o(n).
\end{equation*}
We remark that a generalised version of this result can be found in \cite{Hennequin}. 
In \cite{iliopo}, it was shown that the constants 
of integration can be computed for 
arbitrary values of the coefficients $a(k)$ 
and $b(k)$ by the same method, as in the
derivation of the integration constants 
in the simple case of $a(k)=0$ and $b(k)=1$.

At the end of this section, it should be noted that the worst-case probability is not 
eliminated, but is less likely to occur. In an unfortunate situation, where 
the $k$ smallest or greatest keys are selected as pivots, partitioning 
will yield trivial subarrays and one containing the remaining elements. 
If the chosen pivots happen to be close to the quantiles of the 
array, this yields an optimal partitioning of the array. 
In the next section, we examine ways of a more efficient selection of pivots.

\section{Multipivot--median partitioning}

The preceding analysis of multipivot Quicksort, where $k$ pivots 
are uniformly selected at random has showed that the worst-case 
scenario is less likely from the standard `one--pivot' model. 
Is any other way, where we can reduce further the probability of such scenario?
We have seen that choosing the median from a random sample of 
the array to be sorted, yields savings to the running time 
of the algorithm.
Since we have examined the analysis of multiple pivots, 
then we can select these pivots as the quantiles of a bigger random sample. 

Thus, one can randomly choose a larger sample of $k(t+1)-1$ keys, 
find the $(t+1)$-st, $2(t+1)$-th, \ldots, $(k-1)(t+1)$-th smallest keys 
and use these $(k-1)$ statistics as pivots. Note that for $k=2$, 
this variant contains the median of $2t+1$ Quicksort as a special case and 
for $t=0$, we have the multipivot algorithm, whose mathematical analysis 
was presented in the previous section. This `generalised Quicksort' 
was introduced by Hennequin \cite{Hennequin}. Let $T(n_{\{k, t\}})$ 
be the average of a ``toll function'' 
during the first pass and $f(n_{\{k, t\}})$ the total expected 
cost of this variant, when applied to an 
array of $n$ keys. The following recurrence (which is not presented so simply
in Hennequin) holds:
\begin{align*}
f(n_{\{k, t\}}) & = T(n_{\{k, t\}})+ \dfrac{1}{\dbinom{n}{k(t+1)-1}} \\
& \kern-3em {} \times \underbrace{\sum_{i_{1}} \sum_{i_{2}} \ldots \sum_{i_{k-1}}}_{i_{1} < i_{2} < \ldots < i_{k-1}} 
\Biggl ( \binom{i_{1}-1}{t} \binom{i_{2}-i_{1}-1}{t}\ldots \binom{i_{k-1}-i_{k-2}-1}{t}\binom{n-i_{k-1}}{t} \\
& \kern-2em {} \cdot \bigl (f((i_{1}-1)_{\{k, t\}}) + f((i_{2}-i_{1}-1)_{\{k, t\}}) + \ldots + 
f((i_{k-1}-i_{k-2}-1)_{\{k, t\}}) \\
& \kern-1em {}+ f((n-i_{k-1})_{\{k, t\}})\bigr ) \Biggr ),
\end{align*}
since the pivots $i_{1}, \ldots, i_{k-1}$ are selected to be the $(t+1)$-st, \ldots, 
$(k-1)(t+1)$-th smallest keys of the sample and
each of the $k$ resulting subarrays 
$(i_{1}-1), (i_{2}-i_{1}-1), \ldots, (n - i_{k-1})$ contain $t$ elements of the sample.

As before, the general recurrence of average cost is translated to a differential 
equation with indicial polynomial \cite{Hennequin},
\begin{eqnarray*}
\mathcal{P}_{k(t+1)-1}(\Theta) = (-1)^{k(t+1)-1} \dbinom{\Theta}{k(t+1)-1} 
- k(-1)^{t} \dbinom{\Theta}{t}.
\end{eqnarray*}
A Lemma follows concerning the whereabouts of the roots of this polynomial:
\begin{Lemma}
The indicial polynomial $\mathcal{P}_{k(t+1)-1}(\Theta)$ has $k(t+1)-1$ simple roots, 
with real parts greater than or equal to $-2$.
The real roots are the integers $0, 1, \ldots, (t-1)$, $-2$; $k(t+1)+t$, 
when $t$ is odd and
$k$ is even or when $t$ is even and $k$ is odd
and the $2 \left \lfloor \frac{(k-1)t+k-2}{2} \right \rfloor$ 
complex roots $\lambda_{1}, \ldots, 
\lambda_{\left \lfloor \frac{(k-1)t+k-2}{2} \right \rfloor}$
with their conjugates $\overline{\lambda}_{1}, \ldots, 
\overline{\lambda}_{\left \lfloor \frac{(k-1)t+k-2}{2} \right \rfloor}$.
\end{Lemma}
\begin{proof}
It can be easily deduced that the integers $0, 1, \ldots, (t-1)$ and $-2$ are roots of
the polynomial. Now, when $t$ is odd and $k$ is even, then $k(t+1)-1$ is odd
and so a root $\alpha$ of the polynomial will satisfy 
\begin{align*}
\dbinom{\alpha}{k(t+1)-1} = k\dbinom{\alpha}{t} \, \, \tag{4.8}
\end{align*}
and by simple manipulations we can now verify that $k(t+1)+t$ is also a root.
Similarly, if $t$ is even and $k$ is odd, we have that $k(t+1)-1$ is even and
Eq. (4.8) is valid, again making $k(t+1)+t$ a root. 

It will be proved by contradiction that all roots have real parts greater than or equal to $-2$. 
Note that the argument is similar with the proofs of Lemmas 3.1.1 and 4.2.1.
For any root $r=x+iy$, with $x, y \in \mathbf R$ holds
\begin{align*}
(-1)^{k(t+1)-1}\dbinom{r}{k(t+1)-1}= k(-1)^{t}\dbinom{r}{t}. \, \, \tag{4.9}
\end{align*}
For $r \neq 0, 1, \ldots, (t-1)$, Eq. (4.9) can be written as
\begin{align*}
(-1)^{k(t+1)-1}\dfrac{(r-t)(r-t-1) \ldots \bigr (r -k(t+1)+2 \bigr)}
{(t+1) \ldots \bigl (k(t+1)-1\bigr)}=(-1)^{t}k. \, \, \tag{4.10}
\end{align*}
Assume that $\mathfrak{Re}(r)<-2$, then
\begin{align*}
\left \vert \frac{r-t}{t+1} \right \vert & = \frac {\sqrt{(x-t)^{2}+y^{2}}}{t+1}>
\frac{\sqrt{\bigl(-(t+2)\bigr)^{2}+y^{2}}}{t+1} \geq \dfrac{t+2}{t+1} \\
\left \vert \frac{r-t-1}{t+2} \right \vert & = \frac {\sqrt{(x-t-1)^{2}+y^{2}}}{t+2}
>\frac{\sqrt{\bigl(-(t+3)\bigr)^{2}+y^{2}}}{t+2} \geq 
\dfrac{t+3}{t+2} \\
& \setbox0\hbox{=}\mathrel{\makebox[\wd0]{\vdots}}  \\
\left \vert \frac{r-\bigl(k(t+1)-2\bigr)}{k(t+1)-1} \right \vert & = 
\frac {\sqrt{\bigl (x-k(t+1)+2\bigr)^{2}
+y^{2}}}{k(t+1)-1} > 
\frac{\sqrt{\bigl (-k(t+1)\bigr )^{2}+y^{2}}}{k(t+1)-1} \\
& \qquad {} \geq \frac{k(t+1)}{k(t+1)-1}.
\end{align*}
Considering the product of moduli, we see that the left-hand side of Eq.
(4.10) is greater than or equal to (in modulus) 
the telescoping product $k(t+1)/(t+1)=k$, which gives a contradiction. 
Further, this argument shows that
$-2$ is the unique root with the least real part.

We now show the roots are simple. Assuming, for a contradiction, that $r$ is
a repeated root, we get:
\begin{align*}
(-1)^{k(t+1)-1} \dbinom{r}{k(t+1)-1} \sum_{j=0}^{k(t+1)-2}\dfrac{1}{r-j}
=k(-1)^{t} \dbinom{r}{t}\sum_{j=0}^{t-1}\dfrac{1}{r-j}. \, \, \tag{4.11}
\end{align*}
Eq. (4.9) and (4.11) imply that
\begin{equation*}
\sum_{j=0}^{k(t+1)-2}\dfrac{1}{r-j}=\sum_{j=0}^{t-1}\dfrac{1}{r-j}
\end{equation*}
or 
\begin{equation*}
\sum_{j=t}^{k(t+1)-2}\dfrac{1}{r-j}=0. \, \, \tag{4.12}
\end{equation*}
From Eq. (4.12), we deduce that $\mathfrak{Im}(r)=0$ and $r \in (t, t+1)\cup (t+1, t+2)\cup 
\ldots \cup \bigl (k(t+1)-3, k(t+1)-2 \bigr )$. 
However, the modulus in the left-hand side of Eq. (4.9)
is smaller than $1$, while the right-hand side is greater 
than $k$, proving that all roots are simple. 
\end{proof}

By the Lemma, the polynomial can be written in terms of simple factors
and the differential equation can be solved using the same way, as 
in other variants of Quicksort, previously analysed.
The average cost of `generalised Quicksort' is
\begin{eqnarray*}
\dfrac{\overline{a}(k, t)}{H_{k(t+1)}-H_{t+1}} \bigl ((n+1)H_{n}-n \bigr ) + O(n),
\end{eqnarray*}
when the ``toll function'' is linear and its average is $\overline{a}(k, t)n + O(1)$.

Our analyses of the average cost of the algorithm and its variants has showed that 
any Quicksort needs on average $Cn \log_{e}(n)+ O(n)$ key comparisons for the complete 
sorting of a file consisting of $n$ distinct keys. The constant $C$ can be
made very close to the information--theoretic bound, as we saw. 
In many sorting applications
the `median of $3$' is being used, with savings on the average 
time and little overhead for the computation of median. For 
large arrays, one can use the `remedian of $3^{2}$' Quicksort.

\chapter{Partial order of keys}

\section{Introduction}

Here, we investigate the analysis of Quicksort under the
assumption of prior information of the order of keys. Specifically, we assume
that there is a partial order on the keys. The rough idea is to see how
much having partial information compatible with the true order allows
us to speed up the process of finding the true order.

Let us illustrate the idea first with a simple example. 
Suppose that there are $d$ levels with $k$ keys at each level, so that
$n=kd$. Anything in
a higher level is known to be above everything in a lower level.
Computing the ratio of the expected complexities, we have
\begin{eqnarray*}
\frac{\mathbb{E}(C_{n})}{\mathbb{E}(C^{*}_{n})}=
\frac{d\bigl(2(k+1)H_{k}-4k\bigr)}{2(n+1)H_{n}-4n},
\end{eqnarray*}
where $C^{*}_{n}$ and $C_{n}$ denote the number of comparisons of Quicksort
with uniform pivot selection and in case of partial order, respectively. 
We consider the following cases:
\begin{enumerate}
\item When d is fixed number, then as $n$ tends to infinity,
\begin{align*}
\lim_{n \to \infty}\frac{n}{d} = \infty.
\end{align*}
Thus,
\begin{align*}
\lim_{n \to \infty}\frac{\mathbb{E}(C_{n})}{\mathbb{E}(C^{*}_{n})}
\sim \lim_{n\to \infty}\frac{d\bigl(2k \log_{e}(k)\bigr)}{2n\log_{e}(n)}
=\lim_{n \to \infty}\frac{\log_{e}(k)}{\log_{e}(n)}=\lim_{n \to \infty} \frac{\log_{e}(n/d)}{\log_{e}(n)}=1.
\end{align*}
\item $d=k=\sqrt{n}$.
\begin{eqnarray*}
\lim_{n \to \infty}\frac{\mathbb{E}(C_{n})}{\mathbb{E}(C^{*}_{n})}=\lim_{n \to \infty}\frac{\sqrt{n}
\bigl(2(\sqrt{n}+1)H_{\sqrt{n}}-4\sqrt{n}\bigr)}{2(n+1)H_{n}-4n}
\sim \lim_{n \to \infty} \frac{2n\log_{e}(\sqrt{n})}{2n \log_{e}(n)}=\frac{1}{2}.
\end{eqnarray*}
We see that Quicksort is on average twice as fast, when sorting a partially
ordered array.
\item $k=\frac{1}{c}$ where $c$ is a constant. Then, 
\begin{align*}
\lim_{n \to \infty}\frac{\mathbb{E}(C_{n})}{\mathbb{E}(C^{*}_{n})}=
\frac{cn\bigl(2(\frac{1}{c}+1)H_{1/c}-\frac{4}{c} \bigr)}
{2(n+1)H_{n}-4n}=\frac{n\bigl((2+2c)H_{1/c}-4\bigr)}{2n \log_{e}(n)}
=\frac{c'}{\log_{e}(n)},
\end{align*}
\end{enumerate}
where $c'= \dfrac{(2+2c)H_{1/c}-4}{2}$.

We should think a little about variability too. It is unsurprising that
having the additional information about levels reduces variability of the
number of comparisons, let us get a preliminary result. If we have the
level structure, then $\operatorname {Var}(C_{n}^{*})$ is the sum  of the variances
of sorting each of the $d$ independent levels. Each of these variances,
since there are $k$ keys in each level, is just $\operatorname {Var}(C_{k})$.
For simplicity, we assume $k \rightarrow \infty$ as $n \rightarrow \infty$
and do asymptotics. We then have (for $m$ either $n$ or $k$)
$$\operatorname {Var}(C_{m})\sim \left(7-\frac{2\pi^{2}}{3}\right) \cdot m^{2}$$
and thus we get
\begin{eqnarray*}
\dfrac{\operatorname {Var}(C_{n})}{\operatorname {Var}(C_{n}^{*})}
\simeq \frac{(7-2\pi^{2}/3)n^{2}}{(7-2\pi^{2}/3)k^{2}d}=d
\end{eqnarray*}
so the variance of the version with the presorting is reduced by a
factor of about $d$. 
These suggest there is interest in studying this situation, we now do so in more detail.

\section{Partially ordered sets}

An approach of having additional information is the partial order of the keys. 
We shall employ this assumption along the following lines. 
First, we present a definition \cite{neg}.
\begin{Definition}
Let a finite set $P$ equipped with a binary relation `$\leq$' which
has the following properties. (Here, $x$, $y$ and $z$ are elements of $P$). \newline
(i) $x \leq x$, $\forall x \in P$. (That is, $\leq $ is reflexive)
\\
(ii) If $x\leq y$ and $y\leq x$, then $x=y$. ($\leq$ is antisymmetric)
\\
(iii) If $x \leq y$ and $y \leq z$, then $x \leq z$. ($\leq$ is transitive)
\\
Then the pair (P, $\leq$) is called Partially Ordered Set.
\end{Definition}
Henceforth, in this thesis we abbreviate `partially ordered set' to `poset'. We also present
two key definitions \cite{neg}.
\begin{Definition}
Let $(P,\leq)$ be a poset. We say that two elements $x$ and $y$ of this poset are
comparable if $x\leq y$ or $y\leq x$. Otherwise they are incomparable.
\end{Definition}
\begin{Definition}
Let $(P,\leq)$ be a poset. \newline
(i) A minimal element of $(P,\leq)$, is an element with the property 
that no other element is smaller than it.
A maximal element of $(P,\leq)$, is an element with the property 
that no other element is greater than it. 
\\
(ii) A chain in $P$ is a set $T$ of elements, where every pair of elements
of $T$ are comparable. The number of elements of $P$ in the longest chain in
$P$ is called the height of $P$ and denoted by $h(P)$. 
\\
(iii) An antichain in $P$ is a set $U$ of elements, no two of which are
comparable. The number of elements of $P$ in the order of the
largest antichain is called the width of $P$, and is denoted by $w(P)$.
\\
(iv) A total order in $P$ is a partial order where every pair of elements
are comparable. 
\end{Definition}
For example, the set of subsets of $X=\{1,2\}$ has an antichain of order
$2$, namely $\{1\}$ and $\{2\}$. A chain of length 3 in it, is $\emptyset
\leq \{1\}\leq \{1,2\}$. This partial order is not a total order
as $\{1\}$ and $\{2\}$ are not comparable. Usually, if we have a partial order on a set,
there will be several ways of extending it to a total order on that
set. Often, we will use $\prec$ rather than $<$ to denote the partial order. 
Further, we present the following definition, that we will come across later.
\begin{Definition}
Let a poset (P, $\leq$). Its comparability graph $G(P)$ is the graph with the poset's 
vertex set, such that the elements are adjacent if and only if they are comparable in
(P, $\leq$). Its incomparability graph $G(\tilde{P})$ is the graph, 
such that the elements are adjacent 
if and only if they are incomparable in (P, $\leq$).
\end{Definition}

Another example of a partial order which usually is not a total order is the
collection of subsets of a fixed set $X$, with the partial order $\leq$
being inclusion, normally denoted as $\subseteq$. It is easy to check that, 
for any $A\subseteq X$, we have
that $A\leq A$ since any set is a subset of itself: if $A\leq B\leq A$
then we indeed have that $A=B$, giving asymmetry: and finally, if
$A\subseteq B\subseteq C$ then of course $A\subseteq C$ and so $\leq$
will be transitive. This is not a total order if $X$ has order at least
2, as $\{x_{1}\}\subseteq X$ and $\{x_{2}\}$ are not comparable for
$x_{1}\neq x_{2}$ members of $X$. However, when we have a partial order on a set $P$ 
there will be at least one
total order on $P$ extending it, and in fact usually there will be several
such:
\begin{Definition}
A linear extension of a partial order $(P,\prec)$ is a total order $<$
on the set $P$ such that whenever $x\prec y$ in the partial order, then
we have $x<y$ in the total order too. The number of linear extensions of a poset $P$ 
is denoted by $e(P)$. 
\end{Definition}

In other words, a linear extension of a partial order is a total order
on the same set which is compatible with the partial order.
This is of course of great relevance to us, as the situation we are in
is that we are given partial information on the true order of the set
of elements and want to know how many more pairwise comparisons we have
to do to work out the true order on it: that is, we are trying to identify
which of the numerous linear extensions of the partial order is the true
order on it, with as few comparisons as possible. 

The number of linear extensions 
of a poset can vary substantially
according to the structure of the poset. For example, trivially, if the
partial order happens already to be a total order there is only 
one extension, namely itself. Equally trivially, if the partial order
contains no comparisons -- i.e. it provides no information whatsoever --
then all $n!$ possible orderings of the $n$ elements of $P$ are 
linear extensions.

Here is a generic lower bound on the number of pairwise comparisons we need
to make in order to find the true order of our data, given a partial ordering
$P$ of it.
\begin{theorem}
Given a partial order $(P,\prec)$ which is partial information about the true 
total order on the underlying set $P$, it takes at least
$ \bigl \lceil\log_{2}\bigl (e(P) \bigr) \bigr \rceil$ pairwise comparisons 
to find the total order.
\end{theorem}
\begin{proof} 
Recall that linear extension or total order of a poset $(P,\prec)$ 
is a total order compatible 
with the partial one. All elements are comparable, forming a unique chain. 
Now let $x$, $y$ be incomparable, distinct members of poset. In $A$ 
total orders, we will have that
$x<y$ and in $B$ linear extensions, $x>y$. It holds that $A+B=e(P)$ so, 
$$\mathrm{\max} (A, B) \geq \dfrac{e(P)}{2}.$$ 
Thus, after making one comparison, there are at least $\frac{e(P)}{2}$ 
candidates. 
Similarly, of those number of linear extensions, choosing again two elements, 
we have to consider at least $\frac{e(P)}{4}$ total orders. 
Thus, after $r$ comparisons,
we will examine at least $\frac{e(P)}{2^{r}}$ linear extensions.
We want to identify the unique total order among the number $e(P)$ 
of all possible linear extensions. 
Therefore, this fraction becomes equal to unity 
when $r=\bigl \lceil\log_{2}\bigl(e(P)\bigr)\bigr\rceil$
comparisons. 
\end{proof}
\begin{remark}
The quantity $\log_{2}\bigl(e(P) \bigr)$ is the information--theoretic 
lower bound, that we first came across 
in section $3.1$.
\end{remark}

Kislitsyn \cite{kisl} and independently Fredman \cite{fred}, showed that often this 
lower bound is close to the truth.
\begin{theorem}[Fredman \cite{fred}, Kislitsyn \cite{kisl}]
Sorting an array of $n$ keys, which obeying a partial order $P$,
can be achieved in worst case by $\log_{2}\big(e(P)\big)+2n$ comparisons.
\end{theorem}
\begin{Corollary}
If $P$ is a poset for which $\log_{2}\bigl(e(P)\bigr)$ grows faster than $n$,
we have that the number of comparisons of finding the true total order,
in worst case, is
$\log_{2}\bigl(e(P)\bigr)\bigl(1+o(1)\bigr)$.
\end{Corollary}
\begin{proof} 
We have that this quantity is bounded below by
$\log_{2}\bigl(e(P)\bigr)$ and above by $\log_{2}\bigl(e(P)\bigr)+2n$ which is
$\log_{2}\bigl(e(P)\bigr)\bigl(1+o(1)\bigr)$ by the assumption in the statement of the corollary.
\end{proof}
Kahn and Kim \cite{kk} gave an algorithm for actually doing the finding
of the true total order, which uses at most $54.45\log_{2}\bigl(e(P)\bigr)$ comparisons.
This has recently been reproved by Cardinal {\em et al.}
\cite{Cardinal} whose proof manages to avoid Kahn and Kim's use of the ellipsoid method, a technique 
which though it is in theory polynomial-time, is difficult to do in practice.
What all this makes clear is that, in considering how much information
we can deduce from a random partial order, we will need to know about
the logarithm of the number of linear extensions the partial order
typically has. 

Often a useful notion in studying posets is the theory of levels, \cite{bri}, \cite{kr}.
\begin{Definition}
If $(P,\leq)$ is a poset, we define
$$L_{1}=\{x\in P:\,\not\exists y\in P,\,y\leq x\,\wedge \,y\neq x\}$$
the set of minimal elements of our poset to be the first level of
our poset. The next level is the set
of minimal elements in $\displaystyle P \setminus L_{1}$: each of these will have (at least
one) element of $L_{1}$ below it. We then continue by induction, defining
$L_{i}$ to be the level of minimal elements of
$P \setminus \left (\displaystyle \cup_{j=1}^{i-1}L_{j} \right )$.
\end{Definition}
Note that every level of a poset is an antichain: for two minimal elements
in a poset cannot be comparable with each other. Moreover, every time you go up in a chain, 
you go up to a higher level.
Thus the height of the poset will be the number of levels. 
\begin{Definition}
The linear sum of two posets $(P_{1},\prec_{1})$ and $(P_{2},\prec_{2})$
is a poset with vertex set the disjoint union of $P_{1}$ and $P_{2}$
and with $x\prec y$ if and only if:
\\
(i) if $x$ and $y$ are in $P_{1}$ and $x\prec_{1} y$:
\\
(ii) if $x$ and $y$ are in $P_{2}$ and $x\prec_{2} y$:
\\
(iii) if $x \in P_{1}$ and $y \in P_{2}$, then automatically we have $x \prec y$.
\end{Definition}

A useful definition for us will be the following \cite{neg}.
\begin{Definition} {\ \\}
(i) Suppose $P$ is a partially ordered set, and $L$ is a particular total
order on the same set which agrees with $P$ on every pair of elements
which are comparable in $P$. (In
other words, $L$ is one of the linear extensions of our partial order).
Then a setup is a pair $x,y$ of elements which are incomparable in $P$ but
are consecutive in $L$.
\\
(ii) The number of setups $S(P,L)$ which must be made comparable to
obtain the linear extension $L$ is denoted $s(P,L)$.
\\
(iii) The setup number of $P$ is the minimum, over all the linear
extensions $L$ of $P$, of $s(P,L)$ and it is denoted $s(P)$. 
\end{Definition}
So the point is that, if information is given in the partial order $P$
and using pairwise comparisons to obtain the rest of the order,
$s(P)$ is a lower bound on the number of comparisons which have to be
done to find the true order.

The setup number is also known as the jump number. The following Lemma is used
for the derivation of a simple lower bound on the setup number of a poset.
\begin{Lemma}[Dilworth \cite{dilworth}]
The minimum number of chains into which a poset $P$ can be partitioned
is the width $w(P)$.
\end{Lemma} 
\begin{proof} 
This is standard and there are several proofs, we refer to
\cite{dilworth}.
\end{proof}

\begin{theorem}
For any poset $P$, $s(P)\geq w(P)-1$. 
\end{theorem}
\begin{proof} 
The best case is to partition the poset into $w(P)$ chains 
$C_{1},C_{2}, \ldots, C_{w(P)}$ and hope that there is some ordering of
these chains (which without loss of generality is the order given)
such that every element in $C_{i}$ is less than the minimum element
of $C_{i+1}$ for each $1\leq i\leq w-1$. Because if this happens,
then the total order is just the direct sum of the $C_{i}$ and
we only have to do $w(P)-1$ comparisons of the maximum element of $C_{i}$
with the minimum element of $C_{i+1}$.
\end{proof}

What we will do for the next while is consider various ways in which
we could have a partial order given to us before we start using Quicksort
to determine the complete order. So we are imagining that a previous
researcher had carried out some of the comparisons and we want to know
how many more comparisons we have to carry out to determine the
total order. We will consider cases where the poset is randomly 
generated. 

\section{Uniform random partial orders}

\begin{Definition}
A uniform random partial order is a partial order selected uniformly
at random from all the partial orders on $S=\{1,2, \ldots, n\}$.
\end{Definition}
This means that all partial orders on $S$ are equally likely to be chosen. 
The basic structural result on such posets is the following, rather
surprising, one. 
\begin{theorem}
[Kleitman and Rothschild \cite{kr}. Alternative proof by Brightwell, Pr\"{o}mel
and Steger \cite{Bps}] Suppose that $\leq$ is a uniform random partial order on
$\{1,2, \ldots, n\}$. Then, {\bf whp.} (henceforth, `{\bf whp.}' stands for `with high probability',
which denotes the fact that as $n \to \infty$, the probability of an event approaches $1$) 
there are three levels: the bottom level has approximately $n/4$ elements in it, 
the middle layer approximately $n/2$
elements and the top layer about $n/4$ elements in it.
\end{theorem}

\begin{remark}
This feature of a uniform random partial order -- that it has height only
$3$ -- is surprising to most mathematicians when they hear it, and perhaps
suggests that ``in nature'' posets do not occur uniformly at random
-- some posets are favoured over others.
\end{remark}
Here is the key information on the number of linear extensions.
\begin{theorem}[Brightwell \cite{bri}]
Given any function $\omega(n)$ tending to infinity with $n$ (think of it
as doing so extremely slowly) the number of linear extensions of a random
partial order chosen uniformly at random is {\bf whp.}, between
$$\frac{(n/2)!\bigl((n/4)! \bigr)^{2}}{\omega(n)}\mbox{~and~}(n/2)!
\bigl((n/4)!\bigr)^{2}\omega(n).$$
\end{theorem} 
\begin{proof} 
We refer to Brightwell \cite{bri}, page 66. 
\end{proof}
\begin{Corollary}
When $P$ is selected uniformly at random, we have {\bf whp.} 
\begin{align*}
\dfrac{1}{\log_{e}(2)}\big(n\log_{e}(n)-O(n)\big) 
& \leq \log_{2}\big(e(P)\big) \\
& \leq \frac{1}{\log_{e}(2)}\big(n\log_{e}(n)+O(n) \big).
\end{align*}
\end{Corollary}
In other words, in this situation, the entropy lower bound on the
number of comparisons required is essentially the right answer.
\begin{proof} 
We have by the previous result, taking $\omega(n)$
to go to infinity very slowly, in particular more slowly than $\log_{e}(n)$, that 
\begin{align*}
& \log_{2}\left(\dfrac{(n/2)!\bigl((n/4)! \bigr)^{2}}{\omega(n)} \right)
\leq \log_{2} \bigl(e(P) \bigr) \leq \log_{2}\biggl((n/2)!\bigl((n/4)! \bigr)^{2}\omega(n) \biggr) \\
& \Longrightarrow \frac{1}{\log_{e}(2)}
\log_{e}\left(\frac{(n/2)! \bigl((n/4)! \bigr)^{2}}{\omega(n)}\right)
\leq \log_{2}\bigl(e(P) \bigr)\leq \frac{1}{\log_{e}(2)} \log_{e}\biggl((n/2)!\bigl((n/4)! 
\bigr)^{2}\omega(n) \biggr). \\
& \Longrightarrow \frac{1}{\log_{e}(2)}
\bigl((n/2)\log_{e}(n/2)+2(n/4)\log_{e}(n/4)-(n/2)-2(n/4)+O(\log_{e}(n) \bigr).\\
& \qquad{} \leq \log_{2}\bigl(e(P) \bigr) \\
& \qquad{} \leq \frac{1}{\log_{e}(2)}
\bigl((n/2)\log_{e}(n/2)+2(n/4)\log_{e}(n/4)-(n/2)-2(n/4)+O(\log_{e}(n) \bigr).\\
&\Longrightarrow \frac{1}{\log_{e}(2)}\bigl(n\log_{e}(n)-O(n) \bigr)
\leq \log_{2}\bigl(e(P) \bigr) \leq \frac{1}{\log_{e}(2)}\bigl(n\log_{e}(n)+O(n) \bigr),
\end{align*}
as required. 
\end{proof}

Using this Corollary, a key result follows regarding algorithm's time complexity:
\begin{Corollary}
Given a uniform partial order on a set of $n$ keys, the time taken
to sort them by pairwise comparisons is approximately
\begin{eqnarray*}
\frac{1}{2\log_{e}(2)}\approx 0.72135
\end{eqnarray*}
times the number of comparisons
required by Quicksort to sort them without the partial information.
\end{Corollary} 
\begin{proof} 
The expected number of comparisons required by Quicksort for the sorting of $n$ keys
is $2n\log_{e}(n)\bigl(1+o(1) \bigr)$, 
and since
the variance of this is asymptotically $(7-2\pi^{2}/3)n^{2}\bigl(1+o(1) \bigr)$,
we have by Chebyshev's inequality, letting $C_{n}$ be the number of comparisons
\begin{eqnarray*}
\mathbb{P}\bigl ( \big \vert (C_{n}-\mathbb{E}(C_{n})\bigr) \big \vert > n \log_{e}\bigl(\log_{e}(n)\bigr) \bigr)\leq
\frac{(7-2\pi^{2}/3)n^{2} \bigl(1+o(1) \bigr)}{n^{2}(\log_{e}(\log_{e}(n)))^{2}} \\
\Longrightarrow \mathbb{P}\bigl(\big \vert C_{n}-2n\log_{e}(n) \bigl(1+o(1) \bigr) \vert 
>n\log_{e}\bigl(\log_{e}(n) \bigr) \bigr) 
\to 0, \mbox{~as~} n \to \infty.
\end{eqnarray*}
Thus the probability of the complementary event, namely that
$C_{n}$ is within $n\log_{e}\log_{e}(n)$ of its mean will tend to $1$. 
Therefore, {\bf whp.} the number of comparisons is in
\begin{eqnarray*}
\Big(2n\log_{e}(n)\big(1+o(1)\big)-n\log_{e}\big(\log_{e}(n)\big),
2n\log_{e}(n)\big(1+o(1)\big)+n\log_{e}\big(\log_{e}(n)\big)\Big).
\end{eqnarray*}
Thus, {\bf whp.} it takes about
\begin{equation*}
2n\log_{e}(n)\bigl(1+o(1) \bigr)~~\mbox{comparisons}.
\end{equation*}
On the other hand, we have just seen that with a uniform partial order of keys,
it takes about 
\begin{eqnarray*}
\frac{1}{\log_{e}(2)}\cdot n\log_{e}(n) \mbox{~comparisons.}
\end{eqnarray*}
This number is indeed $1/ \bigl(2\log_{e}(2) \bigr)$ times the number Quicksort needs
and the numerical value of this fraction is as stated. 
\end{proof}

It is of interest to compare this with the naive lower bound on setup number,
which performs rather poorly here.  
\begin{Corollary}
The setup number of a uniform random poset is at least 
\begin{equation*}
\frac{n}{2}+O(1). 
\end{equation*}
\end{Corollary}
\begin{proof} 
This is immediate from the lower bound $s(P)\geq w(P)-1$
and the fact that, by the Theorem of Kleitman and Rothschild \cite{kr}, we clearly
have that $w(P)$ is greater than $n/2+O(1)$. Thus in this case the
simplest setup number lower bound is not a very good one, 
as we have seen that the true answer is $O\bigl(n\log(n) \bigr)$ in this model.
\end{proof}

This in turn implies that if we were to use Quicksort, even in the optimal
cases, we would have to compare $n/2+O(1)$ pairs of keys. The expected
time to do this would be asymptotically $2(n/2)\log(n/2)\sim n\log(n)$.

\section{Random bipartite orders}

In this section, we inspect the number of linear extensions for bipartite orders. 
Let present a definition \cite{bri}.
\begin{Definition}
Let $X$ and $Y$ be two disjoint sets, each one having cardinality equal to $n$.
A random bipartite order $A_{p}(X, Y)$ is the poset $X\cup Y$ with both $X$
and $Y$ antichains, and for each pair $(x, y) \in X \times Y$ there is
a relation $x \prec y$ with probability $p$ and no relation (i.e. they
are incomparable) with probability $1-p$, independently of all other pairs.
\end{Definition}

We now think about the number of linear extensions. A linear extension of
such an order will have to, amongst other things, put the set $X$ in
order -- there are $n!$ ways to do this -- and there are similarly $n!$ ways
to order $Y$. However there will also be some choices to make elsewhere,
because while we have some relations $x<y$ for $(x,y) \in X \times Y$ in the
partial order, we will also have some incomparable pairs. 
More precisely, the probability that a total order
on $X$, a total order on $Y$ and a decision rule $\alpha$ 
for each pair $(x,y) \in X\times Y$
on whether $x\prec y$ or $y\prec x$, is a total order compatible with the
partial order -- i.e. a linear extension of the partial order -- is
$(1-p)^{\ell(\alpha)}$, where $\ell(\alpha)$ is the number of reversals
in $\alpha$, that is the number of pairs $(x,y) \in  X \times Y$ such that
$x>y$ in the total order. 

Thus the expected number of linear extensions is
$\displaystyle (n!)^{2}\sum_{\alpha}(1-p)^{\ell(\alpha)}$.
The function multiplying $(n!)^{2}$ here is discussed at length in
\cite{bri}: it is defined
$$\eta(p):=\prod_{i=1}^{\infty}\bigl(1-(1-p)^{i}\bigr).$$
Though some more work needs to be done to check this, it turns out that
{\bf whp.} the number of linear extensions is close to this mean value.
We quote the result from Brightwell \cite{bri}.
\begin{theorem}[Brightwell \cite{bri}]
The number of linear extensions $e(P)$ of a random bipartite partial order
$A_{p}(X,Y)$ with $\vert X \vert=\vert Y \vert=n$ and probability $p$ 
such that $\displaystyle \lim_{n \to \infty}\cfrac{p(n)
\cdot n^{1/7}}{\bigl(\log_{2}(n) \bigr)^{4/7}}=\infty$, satisfies {\bf whp.},
$$e(P)=(n!)^{2}\cdot \frac{1}{\eta(p)} \cdot \bigl(1+o(1) \bigr).$$
\end{theorem}
In particular this applies when $0<p<1$ is a constant. More precisely,
we have that {\bf whp.}
$$(n!)^{2}\cdot \eta(p)^{-1} \cdot\left(1-\frac{c\log_{2}^{3}(n)}{n}\right)\leq e(P)
\leq (n!)^{2}\cdot \eta(p)^{-1}\cdot \left(1+\frac{c\log_{2}^{3}(n)}{n}\right).$$ 
Thus, for $p$ constant, letting $C=\log_{2}\bigl(\eta(p)\bigr)$ and noting
that both logarithms approach $1$ as $n \rightarrow \infty$, we have
that {\bf whp.}
\begin{eqnarray*}
2\log_{2}(n!)-C+o(1) \leq \log_{2}\bigl(e(P)\bigr) \leq 2\log_{2}(n!)- C+o(1).
\end{eqnarray*}
(Of course the two $o(1)$ terms are different). Thus the order of
magnitude of $\log_{2}\bigl(e(P)\bigr)$ is {\bf whp.}
\begin{eqnarray*}
2\log_{2}(n!)=\frac{2\cdot \log_{e}(n!)}{\log_{e}(2)}
=\frac{2n\log_{e}(n)(1+o(1))}{\log_{e}(2)}.
\end{eqnarray*}

We need to be careful about comparing this example with Quicksort: we 
must remember that the total number of keys being
sorted in this example is $n+n=2n$. Therefore, the expected time would 
be $4n\log_{e}(2n)\bigl(1+o(1)\bigr)=4n\log_{e}(n)\bigl(1+o(1)\bigr)$. 
Thus the factor
by which we are quicker here is again $2\log_{e}(2)$. 
In other words, we get the same speed-up as for 
the uniform and bipartite cases. We now move 
forward to the analysis of random $k$-dimensional orders in Quicksort.

\section{Random k--dimensional orders}

Here, the application of Quicksort in a random k--dimensional order 
is considered. For this purpose,
a definition follows:
\begin{Definition}
A random $k$-dimensional partial order on the set 
$P=\{1,2, \ldots, n\}$ 
is defined as
follows. We select $k$ total orders on $\{1,2, \ldots, n\}$ uniformly at random
from all $n!$ total orders on that set, say we chose
$\leq_{1},\leq_{2},\ldots, \leq_{k}$. We then define the partial order
$\prec$ by
$$x\preceq y\Leftrightarrow x\leq_{i} y\mbox{~for~all~}1\leq i\leq k.$$
\end{Definition}
Of course it is highly likely that some of the total orders will be
inconsistent with each other, and so we will only get a partial order.
We now have to change perspective: we assume that the partial order
which results from these $k$ total orders is given to us as partial
information about the order on the set $P$, and that we have to use
pairwise comparisons to find the true total order on $P$. Therefore, the
total order we are looking for, and the $k$ total orders we used to
define the partial order, may have little to do with each other.

We aim to estimate time complexity of finding the true order when the
partial information given is a random $k$-dimensional order. 
Again, we use the Theorem $5.2.8$
coupled with the information lower bound.
Thus we need to know about the number of linear extensions of a random
$k$-dimensional order. An important result follows,
\begin{theorem}[Brightwell \cite{bright}]
The number $e(P)$ of linear extensions of a random $k$-dimensional 
partial order $P$ 
with $\vert P\vert=n$ satisfies {\bf whp.}
$$\left (e^{-2}n^{1-1/k} \right)^{n}\leq e(P)\leq \left(2kn^{1-1/k} \right)^{n}.$$
\end{theorem} 
Consequently we have that $\log_{2}(e(P))$ is bounded below by
$n\cdot (1-1/k)\cdot \log_{2}(n)(1+o(1))$ and similarly is bounded above by
$n\cdot\big(\log_{2}(2k)+(1-1/k)\cdot\log_{2}(n)\big)(1+o(1))$. This is of
course larger in order of magnitude (for  fixed $k$, say) than $2n$ so
the information lower bound is tight. Now, we deduce that
\begin{Corollary}
The time complexity of finding the true total order
on a set given a random $k$-dimensional partial order on it, where
$k$ is a constant, is {\bf whp.} asymptotically equivalent to
\begin{eqnarray*}
\frac{n\log_{e}(n)}{\log_{e}(2)}\cdot \left(1-\frac{1}{k} \right).
\end{eqnarray*}
\end{Corollary}
\begin{proof}
For the logarithm of the number of the linear 
extensions as $n$ tends to infinity holds,
\begin{align*}
& n\cdot \left(1-\frac{1}{k} \right)\cdot \log_{2}(n) \bigl(1+o(1) \bigr) \leq \log_{2}\bigl(e(P) \bigr) \\
& \quad{} \leq n\cdot\big(\log_{2}(2k)+\left(1-\frac{1}{k} \right) \cdot \log_{2}(n)\big) \\
& \Longrightarrow \frac{n}{\log_{e}(2)}\cdot\left(1-\frac{1}{k} \right)\cdot \log_{e}(n) \\
& \quad{} \leq \log_{2}\big(e(P)\big) \leq \frac{n}{\log_{e}(2)} \cdot\big(\log_{2}(2k)
+\left(1-\frac{1}{k} \right) \cdot\log_{e}(n)\big). 
\end{align*}
Then, we obtain
\begin{eqnarray*}
\log_{2}\big(e(P)\big)=\frac{n\cdot \log_{e}(n)}{\log_{e}(2)}\cdot 
\left(1-\frac{1}{k}\right)\cdot \big(1+o(1)\big),
\end{eqnarray*}
which completes the proof.
\end{proof}
Therefore, the speed up relating to Quicksort with no prior information is
on average
\begin{eqnarray*}
\frac{n\cdot \log_{e}(n)}{2n\cdot \log_{e}(n)\log_{e}(2)}
\cdot \left(1-\frac{1}{k}\right) \\
=\frac{1}{2 \log_{e}(2)} \cdot \left(1-\frac{1}{k}\right) \\
\approx 0.72135 
\cdot \left(1-\frac{1}{k} \right).
\end{eqnarray*} 
\begin{remark}
Note that the factor 
\begin{equation*}
\frac{1}{2\log_{e}(2)} \approx 0.72135.
\end{equation*}
was previously encountered in uniform and bipartite random orders.
It is worth to point out that for the multiplier 
\begin{equation*}
1-\frac{1}{k},
\end{equation*}
when $k$ is arbitrarily large, the speed up is as of the case of uniform random
orders. Whereas, having few $k$-dimensional orders, Quicksort runs much
faster.  
\end{remark}

\section{Random interval orders} 

In this section, we examine the case where poset forms an interval order. 
A definition follows \cite{fishburn}, \cite{tro}.
\begin{Definition}
A poset $(P,\leq)$ is called an interval order if there exists a 
function $I$ such that each element $x \in P$
is mapped to a closed interval $I(x)=[a_{x}, b_{x}] \subseteq \mathbf R$. 
Then $\forall x, y \in P$, 
it holds that $x \prec y$ if and only if $b_{x} \leq a_{y}$.
\end{Definition}

In other words, there exists a mapping $x\mapsto I(x):=[a_{x}, b_{x}]$ for 
every element of $(P,\leq)$, 
having the property that any two elements of the poset are comparable if and only if 
their corresponding intervals do not intersect. Otherwise, they are incomparable.
Hence, the size of the largest chain of the poset is the maximum number of 
pairwise non-intersecting intervals.
Conversely, the size of the largest
antichain is the maximum number of intersecting intervals. We present the 
definition of random interval order.
\begin{Definition}
A random interval order is one where we generate $2n$ independent numbers
$X_{1}, \ldots, X_{n}, Y_{1}, \ldots, Y_{n}$ from the uniform distribution on 
$[0,1]$ and form $n$ closed 
intervals $I_{j}$, for $1\leq j\leq n$, where $I_{j}=[X_{j},Y_{j}]$ if
$X_{j}<Y_{j}$ and $[Y_{j},X_{j}]$ otherwise. (The event that $X_{j}=Y_{j}$
has probability zero so can be ignored). Then we define a partial order by
saying that $I_{i}\prec I_{j}$ if and only if the maximum element of
$I_{i}$ is less than the minimum element of $I_{j}$. 
\end{Definition}
\begin{remark}
In fact any continuous probability distribution can be chosen 
to the analysis of random interval orders.
\end{remark}

Again, we want to estimate how many linear extensions there are of these.
This time, there does not appear to be an immediate bound for the number
of linear extensions in the literature. However one can obtain the
relevant bound showing that $\log_{2}\bigl(e(P)\bigr)$ is {\bf whp.}
at least $cn\log_{e}(n)$ for some $c>0$, which will of course be enough
to show that the $2n$ term in the Fredman \cite{fred} -- Kislitsyn \cite{kisl} bound
$\log_{2}\bigl(e(P) \bigr)+2n$
is small compared with the term $\log_{2}\bigl(e(P)\bigr)$. The main results
that we need in this direction are the following two Theorems, regarding the size of the largest
antichain and chain of a random interval order respectively. 
For their proofs, we refer to \cite{jsw}.
\begin{theorem}[Justicz {\em et al.} \cite{jsw}]
Let $A_{n}$ denote the size of a largest set of pairwise intersecting
intervals in a family of $n$ random intervals. Then there exists a function
$f(n)=o(n)$, such that {\bf whp.} we have  
$$\frac{n}{2}-f(n)\leq A_{n}\leq \frac{n}{2}+f(n).$$
\end{theorem}
\begin{theorem}[Justicz {\em et al.} \cite{jsw}]
Let $Y_{n}$ denote the maximum number of pairwise disjoint intervals 
in a family of $n$ random intervals. Then 
\begin{equation*}
\lim_{n \to \infty} \dfrac{Y_{n}}{\sqrt{n}}=\dfrac{2}{\sqrt{\pi}}
\end{equation*}
in probability. 
\end{theorem}

The following Corollary
gives a lower bound for the number of comparisons 
required to sort a random interval order:
\begin{Corollary}
The number of comparisons for sorting $n$ keys, given a random interval 
order is {\bf whp.} at least
$cn\log_{e}(n)\bigl(1+o(1) \bigr)$, where one can take 
$c=1/2\log_{e}(2)\approx 0.72135$.
\end{Corollary} 
\begin{proof} 
Theorem 5.6.4 shows that the largest antichain of
the random interval order is {\bf whp.} at least
$r=\lceil (1-\epsilon)n/2 \rceil$ for any $\epsilon>0$. This is because
a family of intersecting intervals forms an antichain. Thus we need to sort
all these $r$ incomparable elements of
the partial order in a total order extending it, and there are at least
$r!$ ways of doing this. Using Stirling's formula, we obtain 
\begin{eqnarray*}
r! \geq \left (\dfrac{(1-\epsilon)n}{2} \right)! \sim \sqrt{(1-\epsilon)n \pi}
\left (\frac{(1-\epsilon)n}{2e}\right )^{(1-\epsilon)n/2} \\
\Longrightarrow \log_{2}\bigl (e(P) \bigr) \geq \log_{2}(r!) \geq \frac{(1-\epsilon)n}{2}
\log_{2}\left (\dfrac{(1-\epsilon)n}{2e} \right) + \frac{1}{2} \log_{2}\bigl((1-\epsilon)n \pi \bigr).
\end{eqnarray*}
Then we have that $\log_{2}\bigl (e(P)\bigr)$ has
order of magnitude $n\log_{2}(n)$: in particular $2n=o\bigl(\log_{2}(e(P)\bigr)$
and so the complexity is $\log_{2}\bigl(e(P)\bigr)\bigl(1+o(1)\bigr)$. 
Further, we have that $\log_{2}\bigl(e(P)\bigr)$ is at least, by the above,
$n(1-\epsilon)\log_{2}(n)\bigl (1+o(1)\bigr)/2$ which is equal to
$n\log_{e}(n)\dfrac{1+o(1)}{2\log_{e}(2)}$
and so we can take the constant $c$ to be at least $1/2\log_{e}(2)$.
\end{proof}

We now present a much stronger result that gives sharp bounds 
on the number of linear extensions
of a random interval order, following the insightful 
suggestions of Prof. Colin McDiarmid \cite{inperson}.
\begin{theorem}
The number of comparisons for sorting $n$ keys, given a random interval 
order is {\bf whp.} for $0<\epsilon<1$, between:
\begin{equation*}
\bigl(1-\epsilon+o(1) \bigr)n\log_{2}(n) \leq \log_{2}\bigl(e(P)\bigr) 
\leq \bigl(1+\epsilon+o(1) \bigr)n\log_{2}(n).
\end{equation*}
\end{theorem}
\begin{proof}
Let $0 < a < b < 1$ and  consider the interval $(a, b)$. Let
$I(i, j)$ be the interval $(\frac{i-1}{2^{j}}, \frac{i+1}{2^{j}})$, 
where $i$ and $j$ are positive integers with $i$ odd and $i < 2^{j}$. 
Further, let $j(a, b)$ be the least $j$ 
such that $\frac{i}{2^{j}} \in (a, b)$
for some positive integer $i$. There is a unique such $i$ since if there were
at least two odd $i$, then there is at least one even $k$ between them
and considering $k/2^{j}$, we can replace it by $(k/2)/2^{j-1}$ 
giving a smaller value of $j$ and contradicting the definition of $j$. 
Since $i$ is unique, we may call it $i(a,b)$. 

We denote the interval $I\bigl(i(a, b), j(a, b)\bigr)$ by $J(a, b)$.
Recall that $i$ is odd with $i<2^{j}$ and observe that if $b-a > 2^{-j}$,
then $j(a, b) \leq j$. For a given such $i$ and $j$,
let $\mathcal{A}(i, j)$
be the set of all intervals $(a, b)$, such that $J(a, b)=I(i, j)$.
The sets $\mathcal{A}(i, j)$ are antichains: for if we had
two intervals $(a_{1}, b_{1})$ and $(a_{2},  b_{2})$ in $\mathcal{A}(i, j)$
with $(a_{1}, b_{1})$ being less than $(a_{2}, b_{2})$ (of course this is
equivalent to $b_{1}<a_{2}$) then saying that
$$J(a_{1},b_{1})=J(a_{2},b_{2})=\left(\frac{i-1}{2^{j}}, \frac{i+1}{2^{j}}\right)$$
would imply that $j(a_{1}, b_{1})=j(a_{2}, b_{2})$ and 
$i(a_{1}, b_{1}) =i(a_{2}, b_{2})$. But given that $j(a_{1}, b_{1})=j(a_{2}, b_{2})$,
there is clearly some $i/2^{j}$ with $i$ odd in $(a_{1}, b_{1})$ which is
less than any such in $(a_{2}, b_{2})$ and the result follows. 
Indeed, if the midpoint of $I(i, j)$ (i.e. $i/2^{j})$ is
less than the midpoint of $I(i^{\prime}, j^{\prime})$ (i.e. $i^{\prime}/2^{j^{\prime}}$)
then no interval in $\mathcal{A}(i^{\prime}, j^{\prime})$ can precede
any interval in $\mathcal{A}(i, j)$ in the interval order.

Let $X$ and $Y$ be independent and uniformly distributed
random variables which denote the endpoints of a random interval in
$(0, 1)$. We have that 
\begin{equation*}
\mathbb{P}(\vert X - Y \vert \leq \epsilon/4)< \epsilon/2,
\end{equation*}
since
\begin{align*}
\mathbb{P}\bigl(\vert X-Y \vert \leq \frac{\epsilon}{4}\bigr) & =
\mathbb{P}\bigl(Y-\frac{\epsilon}{4} \leq X \leq Y + \frac{\epsilon}{4}\bigr) \\
& =\int_{0}^{1}\mathbb{P}\bigl(Y-\frac{\epsilon}{4} \leq X \leq Y 
+ \frac{\epsilon}{4} \big\vert Y=y \bigr)f_{Y}(y)\, \mathrm{d}y \\
& =\int_{0}^{1}\mathbb{P}\bigl(Y-\frac{\epsilon}{4}\leq X \leq Y+
\frac{\epsilon}{4} \big\vert Y=y \bigr)\, \mathrm{d}y.
\end{align*}
The last equation follows because the probability density of $Y$, $f_{Y}(y)$ 
is $1$ on $[0,1]$ and $0$ elsewhere.
Given that $Y=y$, the probability that $X$ is in the interval $(y -\epsilon/4, y+\epsilon/4)$ is 
(as it is uniformly distributed) at most the length of the 
interval $(y+\epsilon/4)-(y-\epsilon/4)=\epsilon/2$. 

Thus the random number $\mathcal{N}$ of intervals with length at most $\epsilon/4$ is
stochastically dominated by a binomial random variable $\mathcal{B}(n, \epsilon/2)$ with
$n$ independent trials and success probability $\epsilon/2$. Therefore,
\begin{equation*}
\mathbb{P}\bigl(\mathcal{N} \geq \epsilon n) \leq \mathbb{P}
( \mathcal{B}(n, \epsilon/2) \geq \epsilon n \bigr)
\end{equation*}
and this probability tends to $0$, as $n \to \infty$.
This is a consequence of
Chernoff's inequality, in the following form: if $X$ is a
binomially distributed variable with $n$ independent trials 
and success probability $p$, then for $\delta>0$ 
$$\mathbb{P}\bigl(X\geq n(p+\delta)\bigr)\leq
\left(\left(\frac{p}{p+\delta}\right)^{p+\delta}
\left(\frac{1-p}{1-p-\delta}\right)^{1-p-\delta}\right)^{n}.$$
Chernoff's inequality appears in various places: we refer to \cite{billingsley} and to \cite{prob}.
Since the number being raised to the power $n$ is $<1$, this will indeed
tend to $0$ (in fact will do so rapidly). The result in our case follows plugging in
$p=\epsilon/2$ and $\delta=\epsilon/2$ and shows that the number of intervals
with length at most $\epsilon/4$ is {\bf whp.} less than $\epsilon n$.

Let $j_{0}$ be a positive integer and $m$ be the number of intervals
$I(i,j)$ with $j\leq j_{0}$. Let $\mathcal{I}$ be a set of intervals
and $\mathcal{I}^{\prime}$ be the set of intervals in $\mathcal{I}$ with length $>2^{-j_{0}}$.
Let $n(i, j)$ be the number of intervals of $\mathcal{I}^{\prime}$ in $\mathcal{A}(i, j)$. 
Then, the number of linear extensions $e(P)$ satisfies, since we have to put
each of the antichains in order and there are $r!$ ways to order an
antichain of $r$ elements, 
\begin{equation*}
e(P) \geq \prod n(i, j)! \geq \prod \biggl(\dfrac{n(i, j)}{e} \biggr)^{n(i, j)}
\end{equation*}
and by convexity
$$\log_{2}\bigl(e(P)\bigr)\geq \sum n(i, j)\log_{2}\frac{n(i, j)}{e}
\geq \vert \mathcal{I}^{\prime}\vert \log_{2}\frac{\vert \mathcal{I}^{\prime} \vert}
{em}.$$
Choosing $j_{0}$ sufficiently large that  $2^{-j_{0}}<\epsilon/4$ ,
the number of intervals of length at most $2^{-j_{0}}$ is less than the number of
intervals of length $<\epsilon/4$ which by Chernoff's inequality 
is {\bf whp.} $< \epsilon n$: thus almost all intervals have length at least $2^{-j_{0}}$
and so $\mathcal{I}^{\prime}$ has order at least $(1-\epsilon)n$. This
will give us that {\bf whp.}
$$\log_{2}\bigl(e(P)\bigr)\geq \bigl(1-\epsilon+o(1)\bigr)n\log_{2}(n).$$
This completes the proof, as this gives the lower bound and the upper bound
is just a consequence of the fact that there are at most $n!$ linear
extensions of a partially ordered set with $n$ keys, and then we
use Stirling's formula again. 
\end{proof}

In the next Chapter
we proceed to the analysis of partial orders where the information--theory
lower bound is not $\omega(n)$. Central to the subsequent analysis
are random graphs. 

\chapter{Linear extensions of random graph orders}

Recall that in the previous Chapter, we examined the number of linear 
extensions of various partial orders,
where the information--theoretic lower bound dominated the linear term.
In this Chapter, we examine the case where
both terms are asymptotically equivalent. We obtain bounds
of the expected height of a random graph and we derive a new bound on
the number of linear extensions of a random graph order.

\section{Random graph orders}

In this section, we consider random graphs. A definition follows \cite{sva}:
\begin{Definition}
The Erd\H{o}s--R\'{e}nyi random graph $G(n,p)$ has labelled vertex set
$\{1, 2, \ldots, n \}$
and for each pair of vertices, the probability of an edge arising between
them is $p$, independently of all the other pairs of vertices. 
\end{Definition}
By the definition, here and throughout this thesis, a random graph $G(n,p)$ 
denotes a simple graph,
without loops or multiple edges. 
An independent set of a graph
is a subset of the vertex set, such that there is no edge connecting any two vertices. 
On the other hand, a clique of a random graph
$G(n, p)$ is a subset of its vertex set, with the property that an edge is arising between 
every two vertices.
Note that $p$ may very well depend on $n$. We will usually be interested
in the behaviour as $n \to \infty$. 
\begin{Definition}
The random graph order $P(n, p)$, with partial order relation $\prec$,
is a partially ordered set with underlying set the vertices of $G(n, p)$ and
we initially say that $i \prec j$ if and only if $i<j$ and the edge $i \sim j$ is present
in the random graph $G(n,p)$. We then take the transitive closure of this
relation to get a partial order.

In the same manner, we write $P(\mathbf {Z},p)$ for the infinite partially ordered
set obtained by taking vertex set $\mathbf {Z}$, the set of integer numbers, 
saying initially $i\prec j$
if and only if $i<j$ in the usual total order on $\mathbf Z$ and the
edge $i \sim j$ is present (which it is with probability $p$ independent of
all other edges), and then taking the transitive closure.  
\end{Definition}

The point is that for a partial order we of course require transitivity
by the axioms for a partially ordered set. But of course this is not
guaranteed in a random graph. We could, for example, have the edges 
$1 \sim 2$ and $2 \sim 3$ in the random graph, but no edge between $1$ and $3$. Then
of course we would have put in the edge $1 \sim 3$ as since $1 \prec 2$ and
$2 \prec 3$ we must have $1 \prec 3$ by transitivity.
Note that there are efficient algorithms for finding the transitive
closure of a relation. What we aim to do next, following the analysis in the last section,
is to consider how many comparisons will be needed to finalise the order of a set 
of keys when a random graph
partial order on the set is given already. For our needs, we present the 
definition of graph entropy, introduced by K\"{o}rner \cite{ja}. 
This definition is from Simonyi's survey on graph entropy \cite{sim}.
\begin{Definition}
Let $G=G(n, p)$ be a random graph. Let $X$ be a random variable taking
its values on the vertices of $G$ and $Y$ taking its values on the stable
(independent) sets of $G$. Suppose further that their joint distribution
is such that $X \in Y$ with probability $1$. Also, the marginal distribution
of $X$ on $V(G)$ is identical to the given distribution $P$. Then, the graph
entropy $H(G, P)$ of the random graph $G$ is 
\begin{align*}
H(G,P)=\min I(X\wedge Y),
\end{align*}
where $I(X\wedge Y)$ is as in Definition 3.1.7.
\end{Definition}

As in the last Chapter, we need to know about the number of linear extensions.
The following Theorem from \cite{abbj} will give us what we need.
\begin{theorem}[Alon {\em et al.} \cite{abbj}]
Let $0<p<1$ be fixed and consider $e(P)$, where the partial order is
from $P(n,p)$. Then we have that there are $\mu(p)>0$ and
$\sigma^{2}(p)>0$ such that
$$\frac{\log_{e} \bigl(e(P)\bigr)- \bigl (\mu(p) \cdot n \bigr) }{\sigma(p) \cdot \sqrt{n}}\stackrel
{\mathcal{D}}{\longrightarrow} \mathbb {N}(0,1).$$
\end{theorem}
\begin{Corollary}
There is a constant $c(p)$ such that {\bf whp.} we have
$$\log_{2}\bigl(e(P)\bigr)=c(p)n+O\bigl(n^{1/2}\omega(n)\bigr),$$
where $\omega(n)$ is any function tending to infinity with $n$ (we usually
think of it as doing so very slowly). 
\end{Corollary}
\begin{proof} 
For random variable $\mathbb{N}(0,1)$, the probability that it is
between $-\omega(n)$ and $\omega(n)$ is $1-o(1)$ as $n \to \infty$.
Thus for large enough $n$
\begin{eqnarray*}
\frac{\log_{e} \bigl(e(P) \bigr)-\mu (p)\cdot n}{\sigma (p)\cdot \sqrt{n}}
\in \bigl(-\omega(n), \omega(n)\bigr).
\end{eqnarray*}
Changing the base of the logarithm, we get:
\begin{align*}
&\frac{\dfrac{\log_{2}(e(P))}{\log_{2}(e)}-\mu(p)\cdot n}{\sigma(p)\cdot \sqrt{n}}
\in \bigl(-\omega(n),\omega(n)\bigr) \\
&\Longrightarrow \log_{2}\bigl(e(P)\bigr) \in
\biggl(\log_{2}(e) \cdot \left (\mu(p) n-\omega(n)\sigma(p)\sqrt{n} \right ), \log_{2}(e) \cdot \left ( 
\mu(p) n+ \omega(n)\sigma(p)\sqrt{n} \right )\biggr).
\end{align*}
and this gives the claim, with $c(p)=\mu(p)\cdot \log_{2}(e)>0$. 
\end{proof}
\begin{Corollary}
For a random partial order $P(n,p)$ with $0<p<1$ constant, 
\begin{equation*}
\log_{2}\bigl(e(P)\bigr)=\log_{2}(e)\mu(p) n \bigl(1+o(1) \bigr).
\end{equation*}
\end{Corollary}
\begin{proof} 
The proof follows directly from the previous Corollary. 
\end{proof}

In other words, the logarithm of the number of linear extensions is
linear in $n$. Thus when we use Theorem 5.2.8, we see that both terms in
it $\log_{2}\bigl(e(P) \bigr)$ and $2n$ are linear. So we do not get the
exact asymptotics as in the previous Chapter. Recall that in that 
Chapter, we always had $\log_{2} \bigl(e(P) \bigr)=\omega(n)$
so outweighed the linear term. (In many cases,
$\log_{2}(e(P))$ was of order of magnitude $n\log_{e}(n)$ so won comfortably).
However here it is not clear what multiple of
$n$ will be the time complexity of finding the total order by pairwise
comparisons. Note that it will be at most some constant
multiple of $n$, so this will be quicker than just using about
$2n\log_{e}(n)$ comparisons in Quicksort -- in other words, the partial
order here does substantially speed up the process of finding the true order.

To deal with this question in more detail, we need to know about the structure
of a random graph order. We concentrate to begin with on the case where $p$ is a constant.
In this case, we shall see that basically the partial order consists of
a linear sum of smaller partial orders.  

\section{Additive parameters and decomposing posets}

\begin{Definition}
A vertex $v$ in any partial order is said to be a post if and only if every
other vertex of the partial order is comparable with it.
\end{Definition}
If $v$ is a post in a partial order $(P,\prec)$, then we can write the partial order
as a linear sum of two subposets $(P_{1},\prec)$ and $(P_{2},\prec)$, namely
$$P_{1}=\{x\in P: x\preceq v\}\,P_{2}=\{x\in P: x\succ v\}.$$
Clearly two elements in the same $P_{i}$ which were comparable before
still are, and any element $x$ in $P_{1}$ is smaller than any element $y$
in $P_{2}$ since $x\prec v\prec y$ by assumption. 
\begin{Definition}
A parameter $f$ of partial orders is said to be additive if and only
if, whenever $P$ is the linear sum of two subposets $P_{1}$ and $P_{2}$
we have $f(P)=f(P_{1})+f(P_{2})$.
\end{Definition}

An example of an additive parameter is the height, as if we have a longest
chain in $P_{1}$ and a longest chain in $P_{2}$, we can concatenate them
to form a chain in the linear sum, so $f(P)\geq f(P_{1})+f(P_{2})$:
and in the other direction, given a longest chain in $P$, we restrict to
the subsets and get chains in $P_{1}$ and $P_{2}$, so $f(P)\leq f(P_{1})
+f(P_{2})$, and so they are equal. Another example is the number of further comparisons 
needed to sort a partially ordered set which is a linear sum of two subposets. 
An important Lemma follows	
\begin{Lemma}
The number of further comparisons $c(P)$ of Quicksort which need to be
applied to a poset $P$ to obtain the true ordering, 
with knowledge of where the posts are and where elements are relative to
the posts, is an additive parameter.
\end{Lemma}
\begin{proof} 
Write $P=P_{1}\oplus P_{2}$. 
If we sort the whole linear sum of $P_{1}$ and $P_{2}$ with
$k$ comparisons, we have sorted $P_{1}$ and $P_{2}$ as well. Thus, taking
$k=c(P)$, we see that $c(P)\geq c(P_{1})+c(P_{2})$. (We have sorted
$P_{1}$ and $P_{2}$ using $k$ comparisons: we might have done it with fewer).
Conversely, if we have sorted $P_{1}$ and $P_{2}$ with a total of $k$ comparisons,
then we have sorted the whole of $P$ because, by definition, everything
in $P_{1}$ is above everything in $P_{2}$, thus $c(P) \leq c(P_{1})+c(P_{2})$.
The result follows. 
\end{proof}

Now here is a key result from Brightwell \cite{bri}.
\begin{theorem}[Brightwell \cite{bri}]
With probability $1$, the set of posts in $P(\mathbf {Z},p)$ for
$0<p<1$ constant, is infinite.
\end{theorem}
Basically, each bit between posts will be small. Indeed Brightwell's survey \cite{bri}
also shows that for sufficiently large $k$, given our $p$ (which remember
is constant) there is a constant dependent of $p$, $c(p)>1$ such that the probability that none of
$\{2k,4k, \ldots, 2k^{2}\}$ are posts is less than or equal to $c^{-k}$. We can now start
showing how to use this idea to break down various invariants of $P(n,p)$
into small units. We need some notation about posts. Let their positions be 
$$\ldots U_{-1},U_{0},U_{1}\ldots$$
where $U_{0}$ is the first post at or to the right of $0$. Then we say that
$P_{j}$ is the poset induced on the interval $(U_{j}, U_{j+1}]$. These posets 
are called the factors of the 
partial order. The next Theorem from \cite{bri} presents an 
important result, regarding convergence in distribution
of additive parameters.

\begin{theorem}[Brightwell \cite{bri}]
Let $p$ be a constant with $0<p<1$. Let $f$ be an additive parameter of
partial orders which is not proportional to $\vert P\vert$. Let
$Y$ and $Z$ be the random variables $f(P_{0})$ and $f(P_{-1})$ respectively.
Further, suppose that the moments $\mathbb{E}(Y^{r})$ and $\mathbb{E}(Z^{r})$ are finite for 
all $r \in \mathbf {N}$. Then
there exist constants $\mu=\mu(p) > 0$ and $\sigma=\sigma(p)>0$,
such that
$\mathbb{E} \bigl(f (P(n, p)) \bigr)/n \to \mu$ and 
$\operatorname{Var} \bigl(f (P(n, p)) \bigr)/n \to \sigma^{2}$. Furthermore
\begin{eqnarray*}
\frac{f \bigl(P(n, p)\bigr)-\mu(p) \cdot n}{\sigma(p)
\cdot \sqrt{n}}\stackrel{\mathcal{D}}{\longrightarrow} \mathbb{N}(0, 1),
\end{eqnarray*}
with convergence of all moments.
\end{theorem}

The following Corollary is a consequence of Theorem $6.2.5$.
\begin{Corollary}
Given $0<p<1$ constant, the height of $P(n,p)$, 
which is an additive parameter is {\bf whp.} 
equal to $\mu(p)\cdot n\cdot \bigl(1+o(1) \bigr )$. 
\end{Corollary} 
\begin{proof} 
This follows easily from the previous Theorem. 
\end{proof}

\section{Average height of random graph orders}

What we really need to do now is to obtain bounds for the average number of linear extensions.
Albert and Frieze \cite{albfr}, derived estimates for the average height of
a random graph order, which is an additive parameter as we previously saw.
The idea is the following. A random graph is sequentially constructed; at
each step a new vertex $j$ is added and the probability of an edge from it
to any previously existing vertex $i$ with $i \leq j$ is $1/2$. They
consider both an underestimate of the height and an overestimate. 
We present a generalisation of this construction
with a constant probability $p \in (0, 1)$ of an edge arising. 
\begin{theorem}
The underestimate $f(p)$ and overestimate $h(p)$ increments of 
the average height are given by:
\begin{align*}
& f(p) = 1 - \dfrac{ \displaystyle \sum_{j=1}^{\infty}\Biggl ( \biggl ( 
\displaystyle \prod_{i=1}^{j-1}\dfrac {p(1-p)^{i}}{ 1-(1-p)^{i+2}} \biggr )
\biggl ((1-p)^{j} \biggr ) \Biggr )}
{\displaystyle \sum_{j=1}^{\infty}\Biggl 
(\displaystyle \prod_{i=1}^{j-1}\dfrac {p(1-p)^{i}}{ 1-(1-p)^{i+2}  } \Biggr ) } \\
& h(p) = \dfrac{1}{\displaystyle \sum_{j=1}^{\infty}(1-p)^{\frac{j(j-1)}{2}}}.
\end{align*}
\end{theorem}
\begin{proof}
Let $l_{k}$ be
the length of the longest chain in a random graph order of size 
$k \in \mathbf N$ and $d_{k}$ be the number of top endpoints of longest
chains. Consider the addition of vertex $(k+1)$.
The event that the new vertex $(k+1)$ is the new, unique, endpoint
of a longest chain is the event that one or more of the edges from $(k+1)$ to
the $d_{k}$ endpoints actually arises. In this event, what will
happen is that the length of the longest chain will increase by $1$ and the number
of endpoints of longest chains will drop to $1$, with probability $1-(1-p)^{d_{k}}$. 
The complementary event is that the number of endpoints will be increased by one.

It is at the next step that we make a pessimistic assumption. 
The pessimistic assumption is that, in these cases where $(k+1)$ does
not become the unique endpoint of a longest chain, the number of endpoints
increases by $1$ with probability only $p$. In fact, though there will
certainly be at least one vertex one level below all the $d_{k}$ upper
endpoints, in most cases there will be more than that -- say $r$ of them --
so if $(k+1)$ is joined to any of them the number of longest chains will
increase, and the probability of this happening will be $1-(1-p)^{r}$. In
this event, the number of endpoints will increase by at least $1$, but
the length of the longest chain will remain unchanged. 

Let the random variables $\eta$ and $\theta$ denote the underestimates 
of the length of
the longest chain and the number of endpoints of the longest chain(s)
respectively. These variables obey the following recurrence: 
\begin{eqnarray*}
(\eta_{k+1}, \theta_{k+1})=
\begin{cases} 
(\eta_{k}+1, 1) & \mbox{ ~with~probability~$1-(1-p)^{\theta_{k}}$} \\
(\eta_{k}, \theta_{k}+1) & \mbox{ ~with~probability~$p(1-p)^{\theta_{k}}$} \\
(\eta_{k}, \theta_{k}) & \mbox{ ~with~ probability~$(1-p)(1-p)^{\theta_{k}}$}
\end{cases}
\end{eqnarray*}
Also, let $\mu$ and $\phi$ be the overestimates 
of the length of
the longest chain and the number of endpoints of the longest chain(s)
respectively. In this case, the recurrence relation is:  
\begin{eqnarray*}
(\mu_{k+1}, \phi_{k+1} )=
\begin{cases}
(\mu_{k}+1, 1) & \mbox{~with~probability~$1- (1-p)^{\phi_{k}}$} \\
(\mu_{k}, \phi_{k}+1) & \mbox{~with~probability~$(1-p)^{\phi_{k}}$} 
\end{cases}
\end{eqnarray*}
We consider only the second component. This is clearly a positive recurrent,
irreducible, aperiodic, Markov process with state space the positive integers.
Thus a stationary distribution does exist with limiting probabilities 
\begin{eqnarray*}
p_{j} = \lim_{n \to \infty} \mathbb{P}(\theta_{n}=j)
\end{eqnarray*}
and when $\theta_{k+1}=j+1$, then this could
come about from $\theta_{k}$ being $j$ (with probability $p(1-p)^{j}$) or
from $\theta_{k}$ being $j+1$ (with probability $(1-p)^{(j+1)+1}=(1-p)^{j+2})$. 
The solution giving the stationary distribution
$p_{j}=\displaystyle \lim_{n \to \infty} \mathbb{P}(\theta_{n}=j)$ for the first case is
\begin{align*}
p_{j+1} & =p(1-p)^{j}p_{j}+(1-p)(1-p)^{j+1}p_{j+1}  \\
\Longrightarrow p_{j+1} & =\dfrac {p(1-p)^{j}}{ 1-(1-p)^{j+2}  }\cdot p_{j} \\
\Longrightarrow p_{j+1}& =\prod_{i=1}^{j}\dfrac {p(1-p)^{i}}{ 1-(1-p)^{i+2}  }\cdot p_{1}. 
\end{align*}
The value of $p_{1}$ can be found by the following equation
\begin{eqnarray*}
p_{1}\sum_{j=1}^{\infty}\Biggl ( \prod_{i=1}^{j-1}\dfrac
 {p(1-p)^{i}}{ 1-(1-p)^{i+2}  } \Biggr ) = 1.
\end{eqnarray*}
The expected height increment generally for $p \in (0, 1)$ is
\begin{align*}
p_{1}\sum_{j=1}^{\infty}\Biggl ( \biggl ( \prod_{i=1}^{j-1}\dfrac 
{p(1-p)^{i}}{ 1-(1-p)^{i+2} } \biggr )
\biggl (1-(1-p)^{j} \biggr ) \Biggr ).
\end{align*}
Substituting $p_{1}$, which is a function of $p$, the average height increment is
\begin{eqnarray*}
f(p)=\dfrac{1}{\displaystyle \sum_{j=1}^{\infty}\Biggl (\displaystyle \prod_{i=1}^{j-1}
\dfrac {p(1-p)^{i}}{ 1-(1-p)^{i+2}  } \Biggr )}\sum_{j=1}^{\infty}\Biggl 
( \biggl ( \prod_{i=1}^{j-1}\dfrac {p(1-p)^{i}}{ 1-(1-p)^{i+2}} 
\biggr )\biggl (1-(1-p)^{j} \biggr ) \Biggr ),
\end{eqnarray*}
which is equal to 
\begin{eqnarray*}
1 - \dfrac{ \displaystyle \sum_{j=1}^{\infty}\Biggl ( \biggl ( \displaystyle \prod_{i=1}^{j-1}
\dfrac {p(1-p)^{i}}{ 1-(1-p)^{i+2}  } \biggr )\biggl 
((1-p)^{j} \biggr ) \Biggr )}{\displaystyle \sum_{j=1}^{\infty}
\Biggl (\displaystyle \prod_{i=1}^{j-1}
\dfrac {p(1-p)^{i}}{ 1-(1-p)^{i+2}  } \Biggr )}.
\end{eqnarray*}

For the second (overestimate) case, the stationary distribution 
obeys the following recursive relation
\begin{align*}
p_{j+1} & =(1-p)^{j}p_{j}  \\
\Longrightarrow p_{j+1} & =(1-p)^{j} \cdot (1-p)^{j-1} \cdot \ldots \cdot (1-p)p_{1} 
=(1-p)^{\frac{j(j+1)}{2}}p_{1}.
\end{align*}
Thus, the value of $p_{1}$ can be retrieved by the following equation
\begin{eqnarray*}
\sum_{j=1}^{\infty}p_{j}=p_{1}\sum_{j=1}^{\infty}(1-p)^{\frac{j(j-1)}{2}}=1
\end{eqnarray*}
and the average height increment is
\begin{eqnarray*}
p_{1}\sum_{j=1}^{\infty}\biggl ( (1-p)^{\frac{j(j-1)}{2}} \bigr (1-(1-p)^{j} \bigr ) \biggr ),
\end{eqnarray*}
which further simplified yields the simple expression 
\begin{eqnarray*}
p_{1}\sum_{j=1}^{\infty}\left ((1-p)^{\frac{j(j-1)}{2}} -(1-p)^{\frac{j(j+1)}{2}} \right) = p_{1}.
\end{eqnarray*}
Therefore, the average height increment in the overestimate case is 
\begin{eqnarray*}
h(p) = \dfrac{1}{\displaystyle \sum_{j=1}^{\infty}(1-p)^{\frac{j(j-1)}{2}}}.
\end{eqnarray*}
The argument is complete. 
\end{proof}

We can rewrite the inverse of the overestimate as
\begin{eqnarray*}
\sum_{j=1}^{\infty}(1-p)^{\frac{j(j-1)}{2}} &=&
(1-p)^{0} + (1-p)^{1}+(1-p)^{1+2} + (1-p)^{1+2+3} + \ldots \\
&=& 1 +(1-p) +(1-p)(1-p)^{2}+(1-p)(1-p)^{2}(1-p)^{3}+ \ldots \\
&=& \sum_{j=1}^{\infty} \left (\prod_{i=1}^{j-1}(1-p)^{i} \right).
\end{eqnarray*}
The average height increment
can be computed in terms of $\theta$ elliptic functions. This class
of functions arises in many different areas of Mathematics and a
comprehensive account of their analyses can be found at \cite{abr}, \cite{whi}.
The overestimate height increment assumes the simple form of 
$$\dfrac{2\sqrt [8] {1-p}}{\theta_{2}(0, \sqrt{1-p})},$$ where the
$\theta_{2}$ function is defined by, (see section 16.27 in Abramowitz and Stegun \cite{abr}),
\begin{eqnarray*}
\theta_{2}(z, q) = 2q^{1/4}\sum_{n=0}^{\infty}q^{n(n+1)}\cos \bigl ((2n+1)z \bigr),
\end{eqnarray*}
with $q=1-p$. 

The previous equation can be expressed as a product, instead of summation. 
It holds \cite{whi},
\begin{eqnarray*}
\theta_{2}(z, q) = 2q^{1/4}G \cos (z) \prod_{n=1}^{\infty} \left (1 + 2q^{2n} \cos(2z) + q^{4n} \right ),
\end{eqnarray*}
where $G = \displaystyle \prod_{n=1}^{\infty} ( 1 - q^{2n} )$. Therefore, 
the expected overestimate height increment becomes
\begin{eqnarray*}
\dfrac{1}{\displaystyle \prod_{n=1}^{\infty} \bigl ( 1 - (1-p)^{n} \bigr 
)\displaystyle \prod_{n=1}^{\infty} \bigl ( 1 + 2(1-p)^{n} +(1-p)^{2n} \bigr )} \\
= \dfrac{1}{\displaystyle \prod_{n=1}^{\infty} \biggl 
( \bigl (1 -(1-p)^{n} \bigr ) \bigl ( 1 + (1 -p)^{n} \bigr )^{2} \biggr)}.
\end{eqnarray*}

The following \textsc{Maple} graph plots the underestimate and 
the overestimate functions. Note that 
for $p \geq 0.7$, both $f(p)$ and the overestimate $h(p)$ 
are in fact very close to $p$.
\vspace{.5cm}
\begin{figure}[H]
\centering
\includegraphics[width=9cm, height=9cm, keepaspectratio=true]{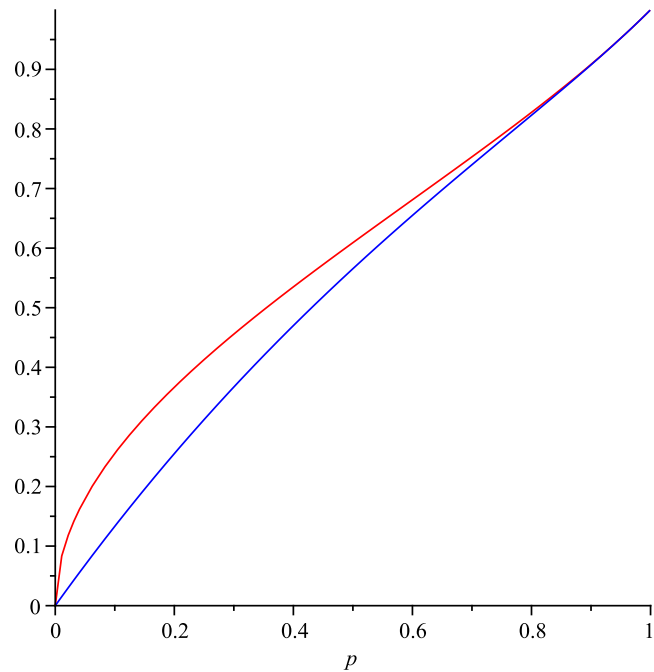}
\caption{Plot of underestimate (blue) and 
overestimate (red) functions of the expected height increment of a random graph order.}
\end{figure}
In the next section, 
we give sharp bounds for the average height.

\section{Bounds on the height of random graph orders}

We have managed to obtain functions $f(p)$, 
which is the underestimate increment of the expected height and
$h(p)$, which is the overestimate of the expected height of a random
graph partial order. We first
give a lower bound on the underestimate function, which is easier
to work with (more tractable). 
\begin{Lemma}
For all $p \in (0,1)$,
$f(p) \geq p$.
\end{Lemma}
\begin{proof} 
The probability of an edge arising between a newly added vertex $(k+1)$ and any of the 
$k$ existing vertices in a random graph is $p$, independently of the other edges.
Thus, with probability $p$ the length of a greedy chain will increase by one and
with probability $(1-p)$, the length will remain unchanged.
Thus, the expected increase in the height of a chain is $p$ and the claim follows immediately.
\end{proof}
Here is an upper bound on the overestimate function.
\begin{Lemma} 
It holds that
\begin{eqnarray*}
h(p)\leq p+ \dfrac{(p-1)^{2}}{2-p}. 
\end{eqnarray*}
\end{Lemma}
\begin{proof} 
To obtain an upper bound on
$$h(p)=\frac{1}{1+(1-p)+(1-p)^{3}+\ldots }$$
it is enough to bound below the denominator $1+(1-p)+(1-p)^{3}+\ldots$ ~. \newline
A crude lower bound
is $2-p$ (as all other terms are positive). Thus 
\begin{eqnarray*}
h(p)\leq \dfrac{1}{2-p}.
\end{eqnarray*} 
This in turn is
\begin{align*}
\frac{1}{2-p} & =\frac{1}{1-(p-1)}=1+(p-1)+(p-1)^{2} + \ldots \\
& =p+(p-1)^{2}+(p-1)^{3} + \ldots \\
& =p+\dfrac{(p-1)^{2}}{1-(p-1)} \\
& =p+\dfrac{(p-1)^{2}}{2-p}. \qedhere
\end{align*}
\end{proof}

Note that as $p$ tends from the left to $1$ and this fact explains the convergence 
of the underestimate and overestimate, seen previously in Figure $6.1$.
Collecting more terms, we will be able to obtain a sharper bound in the following manner.
It holds that 
\begin{eqnarray*}
1+(1-p)+(1-p)^{3}+ \ldots > 2- p + (1 - p)^{3} + \ldots + ( 1 - p)^{K},
\end{eqnarray*}
for a finite number $K = \dfrac{m(m-1)}{2}$, where $m\in\mathbf N$, 
so the bound has the form
\begin{eqnarray*}
\dfrac{1}{2-p + \displaystyle \sum_{j=3}^{m} \left ( (1-p)^{\frac{j(j-1)}{2}} \right)}.
\end{eqnarray*}
Its sharper than the previous one but again infinitely many terms are discarded. In order 
to obtain a bound using all terms, note that $j(j-1)/2 \leq j^{2}$ and 
\begin{eqnarray*}
1+(1-p)+(1-p)^{3}+ \ldots \geq (1-p)+(1-p)^{4}+(1-p)^{9} + \ldots~.
\end{eqnarray*}
The overestimate can be bounded above by another theta function. 
Therefore, a sharper bound is 
\begin{eqnarray*}
h(p) \leq \dfrac{2}{\theta_{3}(0, 1-p)+1}.
\end{eqnarray*}
The $\theta_{3}$ function is defined \cite{abr},
\begin{eqnarray*}
\theta_{3}(z, q) = 1 + 2\sum_{n=1}^{\infty}q^{n^{2}}\cos(2nz).
\end{eqnarray*}
Thus, we obtained suitable bounds, such that
\begin{eqnarray*}
p \leq f(p) < h(p) \leq \dfrac{2}{\theta_{3}(0, 1-p)+1}.
\end{eqnarray*}

In the following section, we present bounds on the number of linear extensions. 

\section{Expected number of linear extensions}

Here we obtain bounds on the average number of linear extensions of a random partial order. 
We start by quoting a useful Theorem from \cite{abbj},
\begin{theorem}[Alon {\em et al.} \cite{abbj}]
Let a random variable $Y$ be geometrically distributed, i.e. 
$\mathbb{P}(Y=k) = pq^{k-1}$, for $k=1, 2, \ldots$~. Then,
\begin{align*}
\mathbb{E}\bigl(\log_{e}(Y) \bigr) & = \sum_{k=1}^{\infty}\log_{e}(k)pq^{k-1} < \mu(p) \\
& \qquad{} \leq \bigl(1-k(p) \bigr)\log_{e}\dfrac{1/p - k(p)}{1-k(p)} \\
&  \qquad{} < \log_{e}\dfrac{1}{p} = \log_{e}\bigl (\mathbb{E}(Y) \bigr),
\end{align*}
\end{theorem}
where $k(p)$ is defined as,
\begin{equation*}
k(p) = \prod_{k=1}^{\infty}(1 - q^{k}),
\end{equation*}
with $q=1-p$.

The following Theorem regarding the expected number of linear extensions
is the main contribution in this Chapter.
\begin{theorem}
Let $p \in (0, 1)$. The average increment of the natural logarithm of the 
number of linear extensions $\mu(p)= \dfrac{\mathbb{E} \bigl(\log_{e}\bigl(e(P)\bigr) \bigr)}{n}$ 
is {\bf whp.} bounded below by:
$$\mu(p) \geq -\dfrac{\log_{2}h(p)\log_{e}(2)}{2}.$$
\end{theorem}
\begin{proof}
We consider the longest chain $C$. We know that {\bf whp.} 
$\mathbb{E}(\vert C \vert) \leq h(p)n$ by Theorem 6.3.1. Now, by Lemma $5$ of
Cardinal {\em et al.} \cite{Cardinal}, which states that
$\vert C\vert \geq 2^{-H(\tilde{P})}n$, we deduce that 
\begin{eqnarray*}
h(p) \geq \mathbb {E} \bigl(2^{-H(\tilde{P})}\bigr),
\end{eqnarray*}
where $H(\tilde{P})$ denotes the entropy of the incomparability graph. By
Lemma $4$ of that paper, we also have that 
\begin{eqnarray*}
nH(\tilde {P}) \leq 2\log_{2} \left (e(P) \right).
\end{eqnarray*}
Therefore
\begin{eqnarray*}
h(p) \geq \mathbb {E} \bigl(e(P)^{-2/n}\bigr).
\end{eqnarray*}
Taking on both sides logarithms and applying Jensen's inequality \cite{jens}, we get
that {\bf whp.}
\begin{align*}
\log_{2} h(p) & \geq \log_{2} \mathbb {E} \left (e(P)^{-2/n} \right ) \\
& \geq \mathbb {E} \left (\log_{2}(e(P)^{-2/n}) \right ) \\
& \quad {}= -\frac{2}{n} \mathbb {E} \bigl (\log_{2}(e(P) \bigr) \\
& \quad{} = -\frac{2}{\log_{e}(2)} \mu(p).
\end{align*}
Thus, $\mu(p)$ is {\bf whp.} bounded below by
\begin{eqnarray*}
\mu(p) \geq -\dfrac{\log_{2}h(p)\log_{e}(2)}{2}
\end{eqnarray*}
and the proof of the Theorem is complete.
\end{proof}

As we can see from the following graph, the bound isn't tight compared 
with the one of Alon {\em et al.} \cite{abbj}. However, its derivation provides
an insight into the relation of the height of a random graph order
with the number of linear extensions.
\vspace{.5cm}
\begin{figure}[H]
\centering
\includegraphics[width=9cm, height=9cm, keepaspectratio=true]{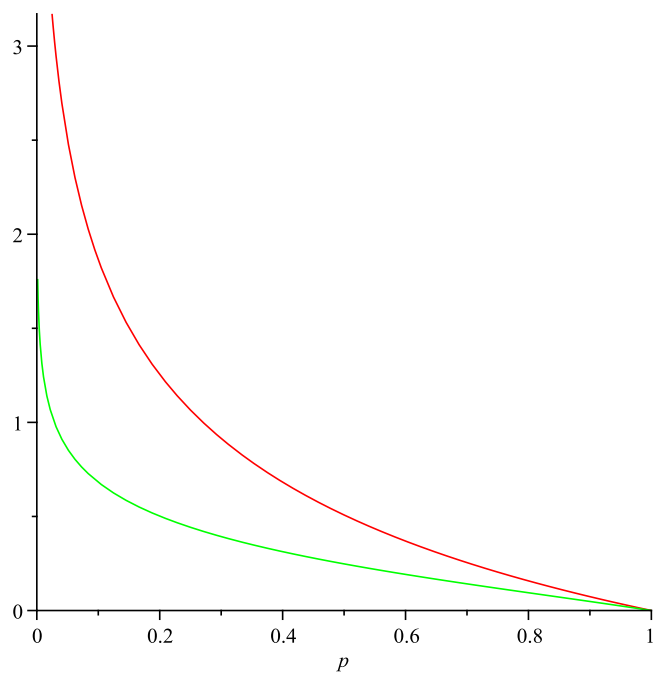}
\caption {Plot of Alon {\em et al.} \cite{abbj} bound (red) and 
of bound of Theorem 6.5.2 (green) on $\mu(p)$.
}
\end{figure}

\chapter{Conclusions and future research directions}

At the last Chapter of this thesis, the conclusions of the research and 
possible future directions are discussed. In the first
section, we consider the sorting of partially ordered sets and in the second
section, the fast merging of chains.

\section{Sorting partially ordered arrays}

Central to the analysis of the time complexity of
sorting partially ordered sets, was the number of linear extensions, as a 
measure of the `presortedness' of the array. Recall
that the quantity $\log_{2}(n!)$ is the lower bound 
of comparisons needed to sort an array of $n$ keys, with no prior information.
In all cases of partially ordered sets considered in Chapter $5$, the constant $1/2\log_{2}(e)$
appeared to the asymptotic number of comparisons. A future direction to research
might be the sharpening of these results. For example, one might ask, 
what is the average number of key exchanges or the computation of 
exact expected costs. 

Generalising Albert--Frieze argument \cite{albfr} and using entropy arguments, a new 
result was the lower bound on the number of linear extensions of a random graph order.
However, we have seen that it does not directly compete the bounds of Alon 
{\em et al.} \cite{abbj}, 
thus there is space for further improvement of this bound or to derivation
of new sharper ones and this might be a suitable topic for further research.

A different bound on the number of linear extensions of a random interval
order can be derived as follows. By Theorem 5.6.5, one can deduce 
that {\bf whp.}, the size of the 
longest chain C, is
\begin{equation*}
\vert C \vert = \dfrac{2\sqrt{n}}{\sqrt{\pi}}.
\end{equation*}
Using the result $\vert C\vert \geq 2^{-H(\tilde{P})}n$ from Cardinal {\em et al.},
as in Theorem 6.5.2, we have
\begin{eqnarray*}
\dfrac{2}{\sqrt{n \pi}} \geq 2^{-H(\tilde{P})}.
\end{eqnarray*}
Taking logarithms, a lower bound for the entropy of the incomparability
random interval graph is,
\begin{equation*}
H(\tilde{P}) \geq \dfrac{1}{2} \log_{2}(n) + O(1).
\end{equation*}
By the following inequality \cite{Cardinal}
\begin{equation*}
nH(\tilde {P}) \leq 2\log_{2} \bigl (e(P) \bigr),
\end{equation*}
$\log_{2} \bigl(e(P) \bigr)$ is {\bf whp.} bounded below by
\begin{equation*}
\log_{2}\bigl(e(P) \bigr) \geq \dfrac{n}{4}\log_{2}(n) + O(n).
\end{equation*}
Thus, {\bf whp.}
\begin{align*}
\log_{2} \bigl (e(P) \bigr) & \geq \dfrac{n}{4}\log_{2}(n)\bigl(1+o(1) \bigr) \\
& \quad {} = \dfrac{n}{4\log_{e}(2)}\log_{e}(n)\bigl(1+o(1) \bigr) \\
& \qquad {} \approx 0.36 n \log_{e}(n)\bigl(1+o(1) \bigr).
\end{align*}
\begin{remark}
The speed up factor of the derived bound of the number of linear extensions of a random interval order is
$1/4\log_{e}(2) \approx 0.36$ -- exactly half of the constant stated in Corollary 5.6.6.
This fact shows that the bound in this Chapter performs rather poorly, comparing
with the results in section 5.6. Note that,
with further information about the entropy, improvement might be possible.
\end{remark}

\section{Merging chains using Shellsort}

In the last section, we consider the merging of chains. 
This problem is analysed in the paper of
Hwang and Lin \cite{hwang}, where an efficient algorithm is presented.

Here, we discuss some preliminary ideas, that might be worthwhile for further study.
Specifically, we propose the application of Shellsort for the merging of linearly
ordered sets.
Shellsort was invented by Donald Shell \cite{Shell} in 1959 and is
based on insertion sort. The algorithm runs from left to right,
by comparing elements at a given gap or increment $d \in \mathbf N$ and exchanging them, if they
are in reverse order, so in the array $\{a_{1}, a_{2}, \ldots , a_{n} \}$
the $d$ subarrays $\{a_{j}, a_{j+d}, a_{j+2d}, \ldots \}$, for $j=1, 2, \ldots, d$ are separately
sorted. At the second pass, Shellsort runs on 
smaller increment, until after a number of passes, the increment becomes $d=1$. 
This final insertion sort completes the sorting of the array.

The sequence of the increments is crucial for the running time of the 
algorithm, as the pivot selection is important to Quicksort. Shell \cite{Shell}
proposed the sequence $\left \lfloor\frac{n}{2} \right \rfloor, 
\left \lfloor\frac{n}{4}\right \rfloor, \ldots, 1$,
which leads to quadratic time. Pratt \cite{Pratt}
suggested a sequence of the form
$2^{a}3^{b} < n$, where $a, b \in \mathbf N$, 
which yields $\Theta \bigl(n\log_{2}^{2}(n)\bigr)$ time.
Incerpi and Sedgewick \cite{inc} have shown that there do
exist $\log_{a}(n)$ increments, for which the
running time of the algorithm is $O(n^{1+\frac{\epsilon}{\sqrt{\log_{2}(n)}}})$, 
with $a=2^{\epsilon^{2}/8}$ and $\epsilon > 0$.
Despite the extensive analysis of Shellsort, there are 
many open problems, as whether the algorithm can achieve
on the average $O \bigl(n \log_{2}(n) \bigr)$ run-time.

In our problem, Shellsort can be fruitfully applied to merging chains, 
which can be done quite fast, 
using the knowledge of the partial order. Consider two chains 
$C_{1}$ and $C_{2}$,
\begin{eqnarray*}
&& C_{1} = \{a_{1} < a_{2} < \ldots < a_{n} \} \\
&& C_{2} = \{b_{1} < b_{2} < \ldots < b_{m} \}.
\end{eqnarray*}
Starting from $a_{1}$, with initial increment 
$d=\mathrm {\max} \{n, m \}$, the 
algorithm separately sorts $d$ subarrays. 
At the second pass, the algorithm iterates
from $a_{2}$
with increment $d-1$. The final comparison of the element
occupying the location $d$ with its adjacent key in the position $d+1$, 
terminates the merging process.

We illustrate this informal idea,
with a simple example. Suppose that we want to merge the chains
$C_{1}=\{5, 7, 9, 11, 12 \}$ and $C_{2}= \{4, 6, 10 \}$. We start
with the array $\{5, 7, 9, 11, 12, 4, 6, 10 \}$
and initial increment $\mathrm {\max} \{5, 3 \}=5$. 
Then, the following subarrays are 
independently sorted: $\{5, 4 \}$, $\{7, 6 \}$, $\{9, 10 \}$,
so at the end of the first iteration, the array becomes
$\{4, 6, 9, 11, 12, 5, 7, 10 \}$. Starting from $6$, with increment 
equal to $4$, the algorithm proceeds to the subarrays
$\{6, 5 \}$, $\{9, 7 \}$, $\{11, 10 \}$, so we obtain
$\{4, 5, 7, 10, 12, 6, 9, 11 \}$. In the same manner, starting
from $7$, with gap equal to $3$, the subarrays $\{7, 6 \}$, 
$\{10, 9 \}$, $\{12, 11 \}$ are sorted, giving $\{4, 5, 6, 9, 11, 7, 10, 12 \}$. 
Then, the subarrays $\{9, 7, 12 \}$ and $\{11, 10 \}$ are sorted,
yielding $\{4, 5, 6, 7, 10, 9, 11, 12\}$. The final comparison 
to $\{10, 9 \}$ returns the merged chain. This algorithm took $13$ 
comparisons for the merging of $8$ keys. A different
increment sequence might speed up the process,
noting that there are some redundant comparisons, e.g. the
comparison of $\{9, 7 \}$ in the second pass. Its appealing feature
is that it merges `in-place', without the need of auxiliary memory.

Central to the argument is the number of inversions of 
a permutation of $n$ distinct keys $\{a_{1}, a_{2}, \ldots, a_{n}\}$.
An inversion 
is a pair of elements, such that for $i < j$, $a_{i} > a_{j}$.
Obviously, an upper bound for the number of inversions is 
$1+2+ \ldots + (n-1)=\dbinom{n}{2}$. On 
the other hand, a sorted array has $0$ inversions. 
In other words, the number of inversions determine
the `amount' of work needed for the complete sorting.
Generally, when one has to merge $k$ chains, 
with cardinalities $m_{1}, m_{2}, \ldots, m_{k}$ 
and $\sum_{j=1}^{k}m_{j}=n$, an 
upper bound to the number of inversions is
\begin{eqnarray*}
\dbinom{n}{2} - \Biggl ( \dbinom{m_{1}}{2}+ 
\dbinom{m_{2}}{2} + \ldots + \dbinom{m_{k}}{2} \Biggr ).
\end{eqnarray*}
Note that this bound corresponds to the case, where
the chains are presented in completely wrong order, 
e.g. the elements of a chain $C_{j}$ are greater from the elements of chains, 
which lie to its right. In practice, this case occurs rarely, 
so the bound can be greatly improved. 

The application of Mergesort, as proposed by Cardinal 
{\em et al.} \cite{Cardinal} for the merging of chains completes the sorting in
$(1+\epsilon)\log_{2} \bigl (e(P) \bigr ) + O(n)$ time -- see their Theorem $3$.
Cardinal {\em et al.} remark in their paper that their Mergesort algorithm is better than
an earlier one of Kahn and Kim \cite{kk}, provided $\log_{2}\bigl(e(P)\bigr)$ is super--linear.
As we have seen in this Chapter, $\log_{2} \bigl (e(P)\bigr )$ is linear,
thus the application of Shellsort might constitute an alternative choice
for the merging of chains.

\begin{appendices}

\chapter{\textsc{Maple} Computations} 

Here is the \textsc{Maple} worksheet for the computation of the 
variance of the number of key comparisons of
dual pivot Quicksort.
\begin{verbatim}
> restart;
> f_{n}(z):= 1/binomial(n, 2)sum(sum(z^2n-i-2f_{i-1}(z)f_{j-i-1}(z)
  f_{n-j}(z), j=i+1..n), i=1..n-1);
> diff(f_{n}(z), z$2);
> subs(z = 1, diff(f_{n}(z), z$2));
\end{verbatim}
\begin{align*}
f''_{n}(1) & = \dfrac{2}{n(n-1)} \biggl ( \sum_{i=1}^{n-1}\sum_{j=i+1}^{n}(2n-i-2)^{2}- 
\sum_{i=1}^{n-1}\sum_{j=i+1}^{n}(2n-i-2) \\
~ & +2\sum_{i=1}^{n-1}\sum_{j=i+1}^{n}(2n-i-2)\mathbb{E}(C_{i-1, 2}) + 
2\sum_{i=1}^{n-1}\sum_{j=i+1}^{n}(2n-i-2)\mathbb{E}(C_{j-i-1, 2}) \\
~ & +2\sum_{i=1}^{n-1}\sum_{j=i+1}^{n}(2n-i-2)\mathbb{E}(C_{n-j, 2})+
2\sum_{i=1}^{n-1}\sum_{j=i+1}^{n}\mathbb{E}(C_{i-1, 2})\mathbb{E}(C_{j-i-1, 2}) \\
~ & + 2\sum_{i=1}^{n-1}\sum_{j=i+1}^{n}\mathbb{E}(C_{i-1, 2})\mathbb{E}(C_{n-j, 2}) +
2\sum_{i=1}^{n-1}\sum_{j=i+1}^{n}\mathbb{E}(C_{j-i-1, 2})\mathbb{E}(C_{n-j, 2}) \\
~ & + \sum_{i=1}^{n-1}\sum_{j=i+1}^{n}f''_{i-1}(1) +  
\sum_{i=1}^{n-1}\sum_{j=i+1}^{n}
f''_{j-i-1}(1)+\sum_{i=1}^{n-1}\sum_{j=i+1}^{n}f''_{n-j}(1) \biggr ). \, \, \tag{A.1} 
\end{align*}
\newpage
The following \textsc{Maple} commands compute the sums of Eq. (A.1) 
\begin{verbatim}
> sum(sum((2n-i-2)^2, j=i+1..n), i=1..n-1);
> sum(sum((2n-i-2), j=i+1..n), i=1..n-1);
> 2(sum(sum((2n-i-2)(2iharmonic(i-1)-4(i-1)), j = i+1..n), i = 1..n-1));
> simplify(%);
> 2(sum(sum((2n-i-2)((2(j-i))harmonic(j-i-1)-4(j-i-1)),
  j=i+1..n), i=1..n-1));
> sum(4 n harmonic(k)(k+i+1)-4nharmonic(k)i-2iharmonic(k)(k+i+1)
  +2harmonic(k) i^2-4harmonic(k)(k+i+1)+4iharmonic(k), k=0..n-i-1);
> 2(sum(3n-3i+2nharmonic(n-i)(n-i)^2 -iharmonic(n-i)(n-i)
  +2nharmonic(n-i)(n-i)-(3(n-i))n+(3/2(n-i))i-iharmonic(n-i)(n-i)^2
  -2harmonic(n-i) (n-i)^2
  +(1/2)i(n-i)^2-2harmonic(n-i)(n-i)-n (n-i)^2+(n-i)^2, i=1..n-1));
> 2(sum(4nharmonic(j)(n-j)^2-5n^2harmonic(j)(n-j)-2nharmonic(j)
  +(n-j)harmonic(j)n+(2(n-j))harmonic(j)-(n-j)^2 
  harmonic(j)+2n^3 harmonic(j)-(n-j)^3 harmonic(j), j=1..n-1));
> sum((2(k+1))harmonic(k)-4k, k = 0..n-i-1);
> 2(sum((2iharmonic(i-1)-4(i-1))(2binomial(n-i+1, 2)
  harmonic(n-i)+(n-i-5(n-i)^2)(1/2)), i=1..n-1));
> sum((8(n-k))binomial(k+1, 2)harmonic(k), k=1..n-1);
> sum(8binomial(k+1, 2)harmonic(k), k=1..n-1);
\end{verbatim}
Using these results, 
the second-order derivative evaluated at $1$ is given recursively by:
\begin{align*}
f''_{n}(1)& = 2(n+1)(n+2)(H^{2}_{n} - H^{(2)}_{n}) - H_{n} \left (\dfrac{17}{3}n^{2} 
+ \dfrac{47}{3}n + 6 \right ) + \dfrac{209}{36}n^{2} \\ 
& \quad{} + \dfrac{731}{36}n + \dfrac{13}{6} 
+ \dfrac{6}{n(n-1)}\sum_{i=1}^{n-1}(n-i)f''_{i-1}(1). \, \tag{A.2}
\end{align*}
The \textsc{Maple} input for the solution of the recurrence is as follows
\begin{verbatim}
> f''_n(1)=2 (harmonic(n))^2 n^2-17/3 harmonic(n) n^2+209/36 n^2
  -2 harmonic(n,2) n^2+731/36 n-6 harmonic(n,2) n-47/3 harmonic(n) n
  +6 (harmonic(n))^2 n+4 (harmonic(n))^2-4 harmonic(n,2)-6 harmonic(n)
  +13/6+6/(n(n-1))sum((n-i)f''_i-1(1), i=1..n-1);
> binomial(n, 2)f''_{n}(1);
> binomial(n+1, 2)f''_{n+1}(1)- binomial(n, 2)f''_{n}(1);
> binomial(n+2, 2)f''_{n+2}(1)-2binomial(n+1, 2)f''_{n+1}(1)
  + binomial(n, 2)f''_{n}(1);
\end{verbatim}
The output of these commands gives the second--order difference
\begin{align*}
\Delta^{2} \dbinom{n}{2}f''_{n}(1) & = 12(n+1)(n+2)(H^{2}_{n} - H^{(2)}_{n}) -
H_{n}(20n^{2} + 32n - 12)  \\
& \quad{} + 17n^{2} +37n + 3f''_{n}(1) \, \, \tag{A.3}
\end{align*}
and after standard manipulations of Eq. (A.3), (recall subsection 4.1.1), 
the solution of the recurrence is 
\begin{equation*}
f''_{n}(1) = 4(n+1)^{2}(H^{2}_{n+1} - H^{(2)}_{n+1}) - 
4 H_{n+1}(n+1)(4n+3) + 23n^{2} + 33n + 12.
\end{equation*}

\end{appendices}

\end{document}